\documentclass{article}
\usepackage{graphicx}
\usepackage[utf8]{inputenc}
\usepackage{braket}
\usepackage{amsmath}
\usepackage[margin=1in]{geometry}
\usepackage{amsfonts}
\usepackage[utf8]{inputenc}
\usepackage[T1]{fontenc}
\usepackage{verbatim}
\usepackage{hyperref}
\usepackage{graphicx}
\usepackage[T1]{fontenc}
\usepackage{amsthm}
\usepackage{amssymb}
\newtheorem{definition}{Definition}[section]

\newtheorem{claim}{Claim}[section]
\newtheorem{theorem}{Theorem}[section]

\newtheorem{lemma}{Lemma}[section]

\newtheorem{fact}{Fact}[section]
\theoremstyle{remark}
\newtheorem*{remark}{Remark}
\usepackage{vwcol}
\usepackage{tikz}
\usepackage{tikz-cd}
\usepackage{mathtools}
\usepackage{tasks}
\usepackage{appendix}
\usepackage{algorithm}
\usepackage{algpseudocode}
\usepackage{makecell}
\usepackage{wrapfig}
\usepackage{xspace}
\usepackage{makecell}
\usetikzlibrary{calc}
\usetikzlibrary{decorations.pathreplacing,calligraphy}

\newcommand{\qx}{\ensuremath{q_X}\xspace}
\newcommand{\qz}{\ensuremath{q_Z}\xspace}
\newcommand{\wx}{\ensuremath{w_X}\xspace}
\newcommand{\wz}{\ensuremath{w_Z}\xspace}
\newcommand{\OO}{\ensuremath{\mathcal{O}}}
\newcommand{\BB}{\ensuremath{\mathcal{B}}\xspace}

\newcommand{\QQ}{\ensuremath{\mathcal{Q}}\xspace}
\newcommand{\BBI}{\ensuremath{\mathcal{B}_i}\xspace}
\newcommand{\BBARI}{\ensuremath{\overline{\mathcal{B}_i}}\xspace}
\newcommand{\SSI}{\ensuremath{\mathcal{S}_i}\xspace}
\newcommand{\QQI}{\ensuremath{\mathcal{Q}_i}\xspace}
\newcommand{\XXI}{\ensuremath{\mathcal{X}_i}\xspace}
\newcommand{\RRI}{\ensuremath{\mathcal{R}_i}\xspace}
\DeclareMathOperator{\im}{im}

\title{Tradeoff Constructions for Quantum Locally Testable Codes}
\author{Adam Wills\thanks{Hon Hai Research Institute, Taipei. Email: \texttt{adamjwills7248@gmail.com}.} \and Ting-Chun Lin\thanks{Department of Physics, University of California San Diego, CA, and Hon Hai Research Institute, Taipei. Email: \texttt{til022@ucsd.edu}.} \and Min-Hsiu Hsieh\thanks{Hon Hai Research Institute, Taipei. Email: \texttt{min-hsiu.hsieh@foxconn.com}.}} 

\begin{document}
\maketitle
\begin{abstract}
    In this work, we continue the search for quantum locally testable codes (qLTCs) of new parameters by presenting three constructions that can make new qLTCs from old. The first analyses the soundness of a quantum code under Hastings' weight reduction construction for qLDPC codes \cite{hastings2021quantum} to give a weight reduction procedure for qLTCs. Secondly, we describe a novel `soundness amplification' procedure for qLTCs which can increase the soundness of any qLTC to a constant while preserving its distance and dimension, with an impact only felt on its locality. Finally, we apply the AEL distance amplification construction \cite{alon1995linear} to the case of qLTCs for the first time which can turn a high-distance qLTC into one with linear distance, at the expense of other parameters.

    These constructions can be used on as-yet undiscovered qLTCs to obtain new parameters, but we also find a number of present applications to prove the existence of codes in previously unknown parameter regimes. In particular, applications of these operations to the hypersphere product code \cite{hastings2016quantum} and the hemicubic code \cite{leverrier2022towards} yield many previously unknown parameters. Additionally, soundness amplification can be used to produce the first asymptotically good testable quantum code (rather than locally testable) - that being one with linear distance and dimension, as well as constant soundness. Lastly, applications of all three results are described to an upcoming work.
\end{abstract}

\section{Introduction}


Informally, a code, be it classical or quantum, is described as locally testable if one can determine, by acting only on a small subset of the bits or qubits (ideally constant in size), whether a given word is a valid codeword or not. There are consequences of this notion in many areas, most importantly in the theory of probabilistically checkable proofs (PCPs), where locally testable codes (LTCs) find their origin. Indeed, the celebrated PCP theorem was proved with the aid of certain classical LTCs \cite{arora1998proof,arora1998probabilistic,dinur2007pcp}. Quantum locally testable codes (qLTCs) were thus first introduced \cite{aharonov2015quantum} in the hope that they might aid in the conversion of the qPCP conjecture into a qPCP theorem \cite{aharonov2013guest}, although it is very much open whether this is possible, or indeed whether the conjecture is true. qLTCs garnered further attention in 2015 when it was proved in \cite{eldar2017local} that the existence of a qLTC of certain parameters would imply the well-known NLTS conjecture \cite{freedman2013quantum}. While the NLTS conjecture was resolved independently \cite{anshu2023nlts}, such qLTCs are still desired to provide an alternative proof of the NLTS theorem. Lastly, there are connections between the properties of qLTCs and passive quantum memories \cite{eczoo_self_correct} as well as single-shot decoding \cite{gu2023single}.

For the benefit of these problems and others, a search has been initiated for quantum locally testable codes of new parameters, in the hope of eventually writing down an optimal qLTC. In this paper, we will present constructions that can turn one qLTC into another that exhibits different parameters. These methods find immediate applications on some of the few qLTCs that are known to exist, producing qLTCs of previously unknown parameters, and will likely continue to find applications as more qLTCs are discovered.

\subsection{Overview of Previous Results}

When talking about qLTCs, one considers four parameters: the soundness, dimension, distance and locality, and asks how each of these parameters scale for a given code as the code length diverges to infinity. The dimension and distance are both standard parameters of error-correcting codes, being respectively the number of encoded logical qubits and the minimum weight of a logical error. We deal exclusively in quantum locally testable CSS codes, for which the distance is often expressed as $d = \min(d_X,d_Z)$, where $d_X$ and $d_Z$ are the X - and Z - distances: the minimum weights of X - type and Z - type logical errors respectively.

Soundness and locality are both important to us as they measure the degree of a code's local testability, although locality is also very important in the area of fault tolerance. For a CSS code with parity-check matrices $H_X$ and $H_Z$ respectively associated with the X - stabilisers and Z - stabilisers, the locality is defined as the maximum row or column weight of either $H_X$ or $H_Z$. For soundness, we first give a definition for classical codes and then the definition for quantum CSS codes will follow. We say that a classical code with parity-check matrix $H \in \mathbb{F}_2^{s \times t}$ is locally testable with soundness $\rho$ if for every $x \in \mathbb{F}_2^t$, $\frac{|Hx|}{s} \geq \rho\frac{d(x,\ker(H))}{t}$, where $d(x,\ker(H)) = \min\{|z| \text{ s.t. } x + z \in \ker(H)\}$. There is then a general definition of local testability for quantum codes given in \cite{aharonov2015quantum}, but this is rather involved. For now, it suffices to simply say that there is a result specialising this definition (Fact 17 of \cite{eldar2017local}) which tells us that, up to constant factors, a quantum CSS code is locally testable with soundness $\rho$ if and only if the classical codes defined by $H_X$ and $H_Z$ are locally testable with soundness $\rho$ according to the previous classical definition.

For an optimal qLTC, the existence of which is currently very much unknown, the soundness is as large as a constant, whereas the locality is as small as a constant. Having these two both being constant gives a ``true''  locally testable code, because the number of qubits acted on in testing a given word is constant. Additionally, for an optimal code, the distance and dimension are both as large as linear. While few qLTCs are known to exist, the history of qLDPC construction, where we care about only dimension, distance and locality, is full and rich due to the hope of using qLDPCs for fault-tolerant quantum computation.

Work on constructing quantum LDPC codes originated with Kitaev's toric code \cite{kitaev2003fault}. Some 20 years of effort went into the improvement of the dimension and distance of these codes under the condition of their being LDPC (having constant locality), until the eventual discovery of good qLDPC codes\footnote{A code is said to be good if it has linear dimension and linear distance.}, that being codes with optimally scaling dimension, distance and locality, in \cite{panteleev2022asymptotically}. Further constructions were found following this, both full \cite{leverrier2022quantum, dinur2023good}, as well as partial \cite{lin2022good}. A major feature of the journey towards better qLDPC codes was the so-called ``square-root distance barrier'', which was the difficulty experienced in constructing an LDPC quantum code with distance meaningfully exceeding a square-root, specifically, exceeding $\sqrt{N}\text{polylog}(N)$. Aside from the toric code itself, the codes of Freedman, Meyer and Luo in \cite{freedman2002z2} exhibit this, giving distance $\sqrt{N}\sqrt[\leftroot{-2}\uproot{2}4]{\log(N)}$, which was the highest known distance for a qLDPC code for around 20 years, however this code has constant dimension. Meanwhile, the hypergraph product codes of Tillich and Z\'emor \cite{tillich2013quantum} gave constant locality with square-root distance and linear dimension. The square-root distance barrier was eventually broken with the discovery of the fiber bundle codes due to Hastings, Haah and O'Donnell \cite{hastings2021fiber}, followed by the introduction of the lifted product \cite{panteleev2021quantum} and the balanced product \cite{breuckmann2021balanced}; good qLDPC codes were discovered shortly after \cite{panteleev2022asymptotically}.

On the classical side, the existence of good LDPC codes has been known for a long time \cite{gallager1962low}, although the existence of optimal classical locally testable codes, known as $c^3$ - LTCs, was not proved until 2021 in \cite{dinur2022locally}, shortly before the independent resolution in \cite{panteleev2022asymptotically}, after which further constructions were exhibited \cite{leverrier2022quantum,lin2022c}.

There are only four papers prior to the present work that aimed to construct qLTCs. The first two papers were from Hastings in 2016 \cite{hastings2016quantum} and Leverrier et al.~in 2019 \cite{leverrier2022towards}, which presented respectively the hyersphere product code and the hemicubic code. As is suggested by their names, these are both obtained from geometric constructions. Because of the similarity in the nature of their construction, and the similarity in their parameters, it makes sense for us to group these two together under the label ``geometric constructions''. Their parameters are shown in the central column of Table \ref{geomConstructionsParams}.

\renewcommand{\arraystretch}{1.5}
\begin{table}[h]

\centering
\begin{tabular}{c||c|c}
& \makecell{Geometric \\Constructions \cite{hastings2016quantum,leverrier2022towards} }& \makecell{Double Distance Balancing\\Applied to the Geometric Constructions \cite{wills2023general}}\\\hline\hline
Physical Qubits & $N$&$\Theta(Nt^2)$\\
Soundness & $\Omega\left(\frac{1}{\text{polylog}(N)}\right)$&$\Omega\left(\frac{1}{\text{polylog}(N)t^2}\right)$\\
Distance & $\Theta\left(\sqrt{N}\right)$&$\Theta(\sqrt{N}t)$\\
Dimension & $\Theta(1)$&$\Theta(t^2)$\\
Locality & $\OO(\log(N))$ & $\OO(\log(N))$
\end{tabular}\caption{The qLTCs of \cite{hastings2016quantum} and \cite{leverrier2022towards} are grouped under the name ``geometric constructions'' for brevity, although it would be remiss of us to not mention that the latter code exhibits a hard-won improvement in its soundness by a logarithmic factor over the former: from $\Omega\left(\frac{1}{\log(N)^2}\right)$ to $\Omega\left(\frac{1}{\log(N)}\right)$.}\label{geomConstructionsParams}
\end{table}
\renewcommand{\arraystretch}{1}

The third paper aiming to construct qLTCs \cite{cross2022quantum} had the primary aim of constructing codes with constant soundness. The authors start with the CSS code defined by the parity-check matrices 
\begin{equation}
H_Z = [I_n|I_n] \hspace{1cm} H_X = [\hat{H}|\hat{H}]
\end{equation}
where $I_n$ is an $n \times n$ identity matrix and $\hat{H}$ is a parity-check matrix for a classical $c^3$ - LTC. This code has constant soundness, constant locality and linear dimension, although constant X - distance and linear Z - distance, meaning it has the major deficiency of having constant overall distance. This is presently the only known example of a ``true'' quantum locally testable code: one with constant soundness and locality. By performing various operations on this code, the authors obtain other codes of improved distance, although for all these codes, as the distance grows, the locality increases (or the soundness drops) by the same amount as the distance is improved, meaning that polynomial distance cannot be obtained without either a polynomial locality or an inverse polynomial soundness. In particular, the authors obtain another code with constant soundness and non-constant locality (although constant \textit{average} locality) by distance balancing the above code using a modified parity-check matrix for the repetition code\footnote{The method behind distance balancing is described in Section \ref{homProduct}.}. This distance balancing procedure allows one to raise the overall distance of the quantum code, and is shown to preserve its soundness\footnote{Here, and throughout this work, saying a quantity is `preserved' means that it is preserved up to constant factors.}, but does so at the expense of its locality. Additionally, this result on soundness preservation applies only to the particular classical and quantum codes in question.

A previous result on distance balancing \cite{hastings2016weight} showed that when distance balancing \textit{any} quantum code with the usual parity-check matrix for the $t$-bit repetition code, we obtain the lower bound of $\Omega\left(\frac{\rho}{t}\right)$ in terms of the soundness of the original quantum code $\rho$. We generalised this in our previous work \cite{wills2023general} to say that when distance balancing any quantum code with any classical code of length $t$, the same lower bound may be obtained - $\Omega\left(\frac{\rho}{t}\right)$ - on the soundness of the new quantum code. This allows one to distance balance using a good classical LDPC code, which has the benefit over distance balancing with the repetition code that the dimension scaling of the quantum code will improve over the course of distance balancing, rather than worsening. As a distance balancing procedure alone, this construction may yield new parameters, for example when applied to the code described at the start of the previous paragraph. However, there is another, more general application of this idea that can potentially yield many more new parameters, which goes by the name `double distance balancing'.

Double distance balancing is a procedure that was originally used in the construction of quantum LDPC codes \cite{panteleev2021quantum}, but we give it a name here for the first time simply so we have something to refer to it by. Applied to any code, even one with already-balanced distances, one may distance balance once, and then again dually, using a good classical LDPC code. This has the effect of growing the dimension of the quantum code at hand, and causing its distance to tend towards a square-root, regardless of where it starts from. This procedure was applied to the geometric constructions in \cite{wills2023general} to obtain the parameters shown in the right-hand column of Table \ref{geomConstructionsParams}. This procedure can be applied to any qLTC, either currently existing, or yet to be discovered, to potentially discover many codes, whose parameters exhibit a tradeoff compared to those of the inputted code.

\subsection{Overview of Our Results}

The motivation behind the present work is to present constructions that can facilitate the rapid exploration of the parameter space of quantum locally testable codes. While it is hoped that one day a truly optimal quantum locally testable code will be written down, it is unclear how long this would take, and it is worth noting that it took some 30 years to construct optimal classical locally testable codes. Therefore, it is worth having tools that allow us to fill out the four-dimensional parameter space as quickly as possible, allowing for the identification of difficult parameter regimes in which research can be focused much more quickly than would be possible by constructing many individual codes.

We exhibit three constructions for this task: weight reduction, soundness amplification and distance amplification. Putting these on a par with the double distance balancing construction applied to qLTCs \cite{wills2023general}, we obtain a suite of four ``tradeoff'' constructions that can be applied to existing, or future quantum locally testable codes, to potentially quickly discover many new parameters. With the exception of the distance amplification construction, the procedures are very general and can be applied to most qLTCs. The distance amplification construction can be used to produce qLTCs of linear distance, although would not produce a code of interest if applied to a code that did not already have quite high distance, for example whose relative distance (distance divided by code length) does not exceed inverse polynomial. Applications will be discussed of the constructions to known qLTCs and to an upcoming work \cite{dinur2024towards} (in which all three constructions find applications).

For the first, and most difficult to prove, of our results, we analyse the soundness of a quantum code as it undergoes the weight reduction procedure of Hastings published in 2021 \cite{hastings2021quantum}, which we must modify slightly to be usable in our case. \cite{hastings2021quantum} replaces the former weight reduction paper \cite{hastings2016weight} from 2016, which contained an error. \cite{hastings2016weight} did, however, analyse the soundness of quantum codes under its procedures, and applied this result to the hypersphere product code to conclude the existence of a new qLTC of constant locality, although this does not now hold given the error in this paper. The newer paper \cite{hastings2021quantum} does not contain an analysis of soundness under its constructions, possibly because the constructions in this latter work are significantly more involved than those of the former. Such an analysis is performed in Section \ref{wtRedSection} to obtain the result of Theorem \ref{wtRedThm}.

\begin{theorem}[Informal]
Given a quantum locally testable CSS code on $N$ physical qubits, there is, under reasonable assumptions on this code, another qLTC of constant locality, whose dimension is equal to that of the first code, and whose number of physical qubits, distances and soundness differ from those of the first code by at most a polynomial factor in the original code's weights.
\label{wtRedThm}\end{theorem}
A formal statement of Theorem \ref{wtRedThm} is presented in Lemma \ref{wtRedThmFormal} in Section \ref{fullWeightRed}.
\\

For our second result, we design a novel ``soundness amplification'' procedure for quantum locally testable codes that can amplify the soundness of any qLTC to a constant from a sub-constant value. It is worth emphasising that this is a non-trivial thing to do in this setting. On the classical side, the notion of local testability used to involve an algorithm that tested for code membership by querying several bits from a purported codeword. In this case, soundness amplification is trivial by a simple repeated application of this algorithm. However, the rather strong, combinatorial notion of local testability now used in the quantum case prohibits this. However, our result is as follows.

\begin{theorem}[Informal]
Given a quantum locally testable CSS code of soundness $\rho$ and locality $w$, there is another qLTC with the same number of physical qubits, the same distance, same dimension and constant soundness that has locality at most $w\;\text{poly}(1/\rho)$.
\label{SAInformalThm}\end{theorem}

A formal statement of Theorem \ref{SAInformalThm} is presented in Lemma \ref{SALemma} in Section \ref{soundnessAmplificationSection}.
\\

For our third and final result, we apply the AEL distance amplification procedure \cite{alon1995linear} to qLTCs for the first time, which finds an application to the code in the upcoming work \cite{dinur2024towards}.

\begin{theorem}[Informal]
Given a quantum locally testable CSS code on $N$ physical qubits with locality $w$ and non-constant distance $d$, let $\Delta = \frac{d}{N}$ be the relative distance of this code. There is, under reasonable assumptions on this code, another qLTC of the same number of physical qubits, up to a constant, the same dimension, linear distance, with locality and soundness differing by at most a polynomial factor in $\frac{1}{\Delta}$ and $w$.
\label{DAInformalThm}\end{theorem}

A formal statement of Theorem \ref{DAInformalThm} is presented in Lemma \ref{DALemma} in Section \ref{DASection}.
\\

A very brief overview of the available tradeoff constructions is presented in Table \ref{constructionOverview}.

\renewcommand{\arraystretch}{1.3}
\begin{table}[h]

\centering
\begin{tabular}{c||c|c|c|c}
& \makecell{Weight\\Reduction}&\makecell{Soundness\\Amplification}&\makecell{Distance\\Amplification}&\makecell{Double\\Distance Balancing}\\\cline{2-5}
&\makecell{Theorem \ref{wtRedThm}\\Present Work}&\makecell{Theorem \ref{SAInformalThm}\\Present Work}&\makecell{Theorem \ref{DAInformalThm}\\Present Work}&\makecell{Corollary 1.1\\\cite{wills2023general}}\\\hline\hline
Soundness & $\searrow$ & $\nearrow$& $\searrow$&$\searrow$\\
Distance & $\searrow$& $\longrightarrow$& $\mapsto \Theta(N)$&$\longrightarrow\Theta(\sqrt{N})$\\
Dimension & $\searrow$ &$\longrightarrow$ & $\longrightarrow$&$\nearrow$\\
Locality & $\mapsto \Theta(1)$&$\nearrow$ & $\nearrow$&$\longrightarrow$
\end{tabular}\caption{This table is intended to give only the essential 
 details of the available tradeoff constructions. Arrows $\nearrow$, $\longrightarrow$ and $\searrow$ indicate a rise, preservation, or falling of the \textit{scaling} of a given parameter in terms of the number of physical qubits (not the actual value) for a typical qLTC, under the assumptions of the given result. $\mapsto X$ indicates that the scaling becomes $X$ immediately, whereas $\longrightarrow X$ indicates that the scaling tends towards $X$, irrespective of its initial value.}\label{constructionOverview}
\end{table}
\renewcommand{\arraystretch}{1}

While we believe that the main use of these results will be to as-yet undiscovered codes, there are several new parameter regimes in which we may immediately construct codes using these results. One use case for the weight reduction procedure are the geometric constructions. Technically, applying the tradeoff constructions to these two codes separately gives new parameters in each case, because the hypersphere product code has a better locality by a double logarithmic factor, whereas the hemicubic code has a better soundness by a logarithmic factor, but this is unnecessary detail, and we apply the construction only to the hemicubic code for simplicity. Both weight reduction and soundness amplification may (separately) be applied, and further double distance balancing may be used in each case to obtain many new parameters. The results of this are shown in Table \ref{newGeometricParams}.

\renewcommand{\arraystretch}{1.5}
\begin{table}[h]

\centering
\begin{tabular}{c||c|c|c|c}
& Hemicubic Code \cite{leverrier2022towards}& \makecell{WR and DDB\\Applied to\\Hemicubic Code}&\makecell{SA and DDB\\Applied to\\Hemicubic Code}&\makecell{DDB and SA\\Applied to\\Hemicubic Code}\\\hline\hline
\makecell{Physical\\Qubits} & $N$&$\mathcal{O}(N\text{polylog}(N)t^2)$&$\Theta(Nt^2)$&$\Theta\left(Nt^2\right)$\\
Soundness & $\Omega\left(\frac{1}{\log(N)}\right)$&$\Omega\left(\frac{1}{\text{polylog}(N)t^2}\right)$ & $\Omega\left(\frac{1}{t^2}\right)$&$\Theta(1)$\\
Distance & $\Theta\left(\sqrt{N}\right)$&$\Omega\left(\frac{\sqrt{N}t}{\text{polylog}(N)}\right)$&$\Theta\left(\sqrt{N}t\right)$&$\Theta\left(\sqrt{N}t\right)$\\
Dimension & $\Theta(1)$&$\Theta(t^2)$&$\Theta\left(t^2\right)$&$\Theta\left(t^2\right)$\\
Locality & $\OO(\log(N))$ & $\Theta(1)$ & $\mathcal{O}\left(\log(N)^{2+\delta}\right)$& $\mathcal{O}\left(\log(N)^{2+\delta}t^{2(1+\delta)}\right)$
\end{tabular}\caption{The third, fourth and fifth columns show the new code parameters that can be obtained by applying our results in different combinations to the hemicubic code. We note the great many of them demonstrates the capacity of the tradeoff constructions to generate many new code parameters from one input code. In the third column, weight reduction (WR) and double distance balancing (DDB) is applied (the order is irrelevant), where the classical codes in the latter construction are of length $t$. $t$ may be set to any value, for example a constant, or some function of $N$. The latter two columns show the result when soundness amplification (SA) and then DDB are applied, in this order and the reverse order, respectively. In the localities, $\delta > 0$ can be chosen to be any constant. The hemicubic code is omitted but similar, albeit previously unknown parameters can also be obtained.}\label{newGeometricParams}
\end{table}
\renewcommand{\arraystretch}{1}

As a second immediate application, the soundness amplification procedure may be applied to asymptotically good (LDPC) quantum codes. Since this operation preserves distance and dimension, we obtain a code of constant soundness, linear distance and linear dimension - producing the first known asymptotically good testable quantum code\footnote{It would interesting to determine if these parameters are also achievable via a random construction.} (rather than locally testable), to use the terminology of \cite{chapman2023efficiently}, for example.

Thirdly, each one of our three results finds applications in an upcoming work \cite{dinur2024towards}. Here, qLTCs are discovered with every parameter within a polylogarithm of optimality, but whose existence is contingent upon a conjecture in classical coding. This code and our applications to it are summarised in Table \ref{upcomingNewParams}.

\renewcommand{\arraystretch}{1.5}
\begin{table}[h]

\centering
\begin{tabular}{c||c|c|c}
& \makecell{Conjectured\\Codes of \cite{dinur2024towards}}& \makecell{Weight Reduction\\Applied to \cite{dinur2024towards}}&\makecell{Distance Amplification and\\Soundness Amplification\\Applied to \cite{dinur2024towards}}\\\hline\hline
Soundness & $\Omega\left(\frac{1}{\text{polylog}(N)}\right)$&$\Omega\left(\frac{1}{\text{polylog}(N)}\right)$ & $\Theta(1)$\\
Distance & $\Omega\left(\frac{N}{\text{polylog}(N)}\right)$&$\Omega\left(\frac{N}{\text{polylog}(N)}\right)$&$\Theta(N)$\\
Dimension & $\Theta(N)$&$\Omega\left(\frac{N}{\text{polylog}(N)}\right)$&$\Theta(N)$\\
Locality & $\OO(\text{polylog}(N))$ & $\Theta(1)$ & $\OO(\text{polylog}(N))$
\end{tabular}\caption{The parameters of the conjectured codes of \cite{dinur2024towards} followed by their parameters with weight reduction applied, and then distance amplification and soundness amplification (in this order). We note that because soundness amplification preserves distance and dimension, these latter two constructions can be applied together.}\label{upcomingNewParams}
\end{table}
\renewcommand{\arraystretch}{1}

\subsection{Structure of the Paper}

The paper is structured as follows. Section \ref{prelimsSection} presents the necessary preliminaries. In particular, Sections \ref{CSSPreliminaries} and \ref{locTestPrelims} give the broad background on quantum CSS codes and local testability, while Section \ref{homProduct} defines the homological product, which is necessary to understand the constructions of Sections \ref{wtRedSection}. Section \ref{wtRedPrelims} gives the necessary preliminary information on Hastings' weight reduction constructions. This is necessary to understand Section \ref{wtRedSection}, which gives the analysis of a qLTC under this weight reduction procedure, along with the necessary minor modifications to this procedure for our case. Section \ref{soundnessAmplificationSection} then describes our soundness amplification procedure and proves the claimed change in parameters of a qLTC under this operation. The AEL distance amplification construction is defined and applied to qLTCs in Section \ref{DASection}. Lastly, Appendix \ref{WRFullParams} gives tables relevant to Section \ref{wtRedSection} to present the full, detailed, changes in parameters.

\subsection{Further Directions and Discussion}

The search for qLTC constructions of new parameters that may, or may not, be inputtable to the tradeoff constructions of the present work to produce codes of previously unknown parameters is of great interest. In particular, an important next step is to obtain a qLTC of constant soundness, constant locality and scaling distance, where currently the only know code of constant soundness and constant locality also has constant distance \cite{cross2022quantum}. 

On the tradeoff constructions themselves, it would be interesting to see if the weight reduction procedure applied here to qLTCs could be used for a weight reduction of more general chain complexes while obtaining a lower bound on properties analogous to soundness. In addition, the presented weight reduction method for qLTCs turns a code of soundness $\rho$ and locality $w$ into a code of soundness $\Omega(\rho/\text{poly}(w))$, where the polynomial is of quite large degree in the worst case. On the other hand, the soundness amplification procedure turns the same code into one of constant soundness and weight $\mathcal{O}(w/\rho^{1+\delta})$ for arbitrary $\delta > 0$. It would be interesting in both cases to attempt to improve the analysis and/or the constructions to achieve soundness $\Omega(\rho/w)$ in the first case and weight $\mathcal{O}(w/\rho)$ in the second case which seem to be the natural scalings.

We wonder also about the applications of these constructions to Hamiltonian complexity theory. For example, operations along similar lines to the weight reduction and soundness amplifications found here may one day be useful ingredients in a gap amplification procedure for a qPCP theorem, in analogy to Dinur's combinatorial proof of the classical PCP theorem \cite{dinur2007pcp}, although this is highly speculative at the moment. One point to mention here is the explicitness of the constructions. While the weight reduction method is explicit, the soundness and distance amplification results rely on probabilistic methods, although we expect that it may be possible to make them explicit with further work.

We note that a previous version of the paper included the `identity product' construction - a homological product construction that is able to grow the dimension of a qLTC at the expense of its distance while preserving its soundness and locality. This was included despite the achievability of the same code transformation via the simple operation of copying the code because we hoped it would lead to further tradeoff constructions interpolating between the parameter transformations achievable by double distance balancing and the identity product. While we found such constructions, their parameter transformations turn out to be achievable via combinations of copying the code and double distance balancing. The identity product has therefore been removed and replaced with soundness amplification in this version.

Nevertheless, we comment that via code copying, more parameters are achievable. For example, applying this to the codes discovered in Table \ref{newGeometricParams} can achieve polynomial dimension and distance, along with either inverse polylogarithmic soundness and constant locality, or constant soundness and polylogarithmic locality.

\section{Preliminaries}\label{prelimsSection}
\subsection{Quantum CSS Codes}\label{CSSPreliminaries}

We review the relevant information on the homological formulation of quantum CSS codes but recommend \cite{breuckmann2021quantum}, along with its references, for further detail.

A quantum CSS code on qubits is equivalent to a three-term chain complex of vector spaces over the binary field, which may be written generally as

\begin{equation}
   C = \left(\mathbb{F}_2^{N_Z}\overset{H_Z^T}{\longrightarrow}\mathbb{F}_2^{N}\overset{H_X}{\longrightarrow}\mathbb{F}_2^{N_X}\right).\label{standardCC}
\end{equation}
where $N_Z$, $N$ and $N_X$ are respectively the number of Z - stabilisers, qubits, and X - stabilisers of the code. The idea behind this equivalence is to define parity-check matrices $H_Z \in \mathbb{F}_2^{N_Z \times N}$ and $H_X \in \mathbb{F}_2^{N_X \times N}$ by the stabilisers of this code; a row of $H_Z$, for example, denotes the support of one Z - stabiliser. Then the condition $H_X\cdot H_Z^T = 0$ is necessary and sufficient for commutativity of the stabilisers. Throughout, when describing a quantum CSS code by a three-term chain complex, the bottom `zeroth' space will always be that of the X - stabilisers, the `first' space will be that of the qubits and the top `second' space will be that of the Z - stabilisers. Given a code or chain complex such as $C$, the `code associated with the X - stabilisers of $C$' is the classical code $\ker(H_X)$ and similarly for the Z - stabilisers. Given that a quantum CSS code is nothing more than two classical codes satisfying a compatibility condition, it will sometimes be natural to conflate the terms `qubits' and `bits' in certain proofs, where appropriate, because often our proofs are only proofs about classical codes associated with various sets of stabilisers, and the bits of these classical codes are representative of qubits in the quantum code.

It is possible to `take the dual' of the chain complex $C$ to give the new chain complex

\begin{equation}
   C^* = \left(\mathbb{F}_2^{N_X}\overset{H_X^T}{\longrightarrow}\mathbb{F}_2^{N}\overset{H_Z}{\longrightarrow}\mathbb{F}_2^{N_Z}\right).\label{standardCCDual}
\end{equation}
which has the effect of swapping the role of X and Z. Taking the dual of more general chain complexes can be done by reversing the order of the spaces and taking the dual of all the linear maps. 

Throughout, we will endeavour to avoid using excessive homological language - nevertheless we define the following terms. We recall that the term `chain' is standard terminology for an element of a space in a chain complex, whereas a `cell' refers to an element of some special basis of one of the spaces, which is always the standard basis of a binary vector space for us. Therefore, for example, a 1-cell in the chain complex of Equation \eqref{standardCC} refers to an individual qubit, whereas a 0-chain in this chain complex refers to a collection of X - stabilisers. It will be important in certain places to ensure a distinction between a `stabiliser', which is a generator of the stabiliser code represented by individual 0-cells or 2-cells, and a `member of the stabiliser group', which is a product of generators of the stabiliser code represented by 0-chains or 2-chains. We recall also that, for a chain complex $X$ with spaces $X_i$ and boundary maps $\partial_i : X_i \rightarrow X_{i-1}$, we refer to i-cycles as the i-chains that are in $\ker \partial_i$ and i-cocyles as the i-chains that are in $\ker \partial_{i+1}^T$. i-boundaries and i-coboundaries, also called trivial i-cycles and i-cocyles, are respectively the i-chains that are elements of $\im\partial_{i+1}$ and $\im\partial_i^T$. For the chain complex $C$, we note that 1-cycles represent Z - operators that commute with all the X - stabilisers, while 1-boundaries represent members of the Z - stabiliser group. Similarly, 1-cocyles represent X - operators that commute with all the Z - stabilisers, while 1-coboundaries represent members of the X - stabiliser group.

It will be common for the spaces of chain complexes such as $C$ to be direct sums of smaller spaces, for example $\mathbb{F}_2^N = \mathbb{F}_2^A \oplus \mathbb{F}_2^B$ or $\mathbb{F}_2^{N_X} = \mathbb{F}_2^C\oplus\mathbb{F}_2^D$. In this case, when we refer to ``qubits in $\mathbb{F}_2^A$'', we are referring to the qubits represented by basis elements of $\mathbb{F}_2^A$, and similarly ``X - stabilisers in $\mathbb{F}_2^C$'' is shorthand for the X - stabilisers represented by basis elements of $\mathbb{F}_2^C$, etc.

We now define the parameters of our codes that are important to us. The X - and Z - distances of the code are the minimum weights of logical X - type and Z - type errors. These are

\begin{equation}
    d_X = \min\left\{|v| : v \in \ker(H_Z) \setminus \text{im}(H_X^T)\right\}
\end{equation}
and
\begin{equation}
    d_Z = \min\left\{|v| : v \in \ker(H_X) \setminus \text{im}(H_Z^T)\right\}
\end{equation}
which may be described homologically as the minimum weights of non-trivial elements of first cohomology and first homology respectively. The distance of the code is then the minimum of these two: $d = \min\{d_X,d_Z\}$, which is the minimum weight of a logical error of our code. Another important parameter of the code is its dimension, the number of logical qubits which it may encode, which is

\begin{equation}
    K = \dim\left(\ker(H_X)/\im(H_Z^T)\right)
\end{equation}
i.e. the dimension of the first homology group. Sometimes, the literature refers to rate and relative distance, which are respectively dimension divided by code length and distance divided by code length. 

Important quantities in this work will be the `weights' of the quantum code, also sometimes referred to as the localities. Following the notation of \cite{hastings2016weight} and \cite{hastings2021quantum}, we define

\begin{align}
    \wx &= \text{max row weight}(H_X)\\
    \wz &= \text{max row weight}(H_Z)\\
    \qx &= \text{max column weight}(H_X)\\
    \qz &= \text{max column weight}(H_Z).
\end{align}
\wx is therefore the maximum weight of an X - stabiliser, and similarly for \wz, whereas \qx is the maximum number of X - stabilisers acting on a given qubit, and similarly for \qz. The overall `locality' of the code is defined as

\begin{equation}
    w = \max\{\wx,\wz,\qx,\qz\}.
\end{equation}
In places, we will describe a code as `high weight' or `low weight', which simply means that it has non-constant locality or constant locality respectively. Similarly, to describe a given stabiliser as `high weight' means that its support is super-constant in size.

All of the quantities $N$, $N_X$, $N_Z$, $d_X$, and so on, will be standard for all quantum codes. We will often refer to the parameters of various codes with different decorations, for example decorating all parameters of a particular code with a prime or a tilde. Typically, the term `code' really refers to an infinite family of codes whose length (number of physical qubits) tends to infinity. We measure the size of the parameters that are important to us in terms of their asymptotic scaling in terms of the code length, $N$. The optimal scaling of dimension and distance is $\Theta(N)$, i.e. linear, whereas the optimal scaling of locality is constant, $\Theta(1)$, which is to say that it is not ``scaling''. A code with these three parameters ``scaling'' simultaneously optimally, is referred to as a good LDPC code.

\subsection{Classical and Quantum Local Testability}\label{locTestPrelims}

We will now review the relevant definitions in the area of local testability for both classical and quantum codes.

There are multiple notions of local testability on the classical side, with relations between them, but one in particular has become adopted in the recent quantum literature \cite{leverrier2022towards, panteleev2022asymptotically, cross2022quantum}. Given a classical code described by some parity-check matrix $H \in \mathbb{F}_2^{s \times t}$, the code is said to be locally testable with soundness $\rho$ if for every $x \in \mathbb{F}_2^t$,

\begin{equation}
    \frac{|Hx|}{s} \geq \rho \frac{d(x,\ker(H))}{t}\label{cLTCDef}
\end{equation}
where $d(x,\ker(H)) = \min\{|z| \text{ s.t. } x+z \in \ker(H)\}$ is the distance of the word $x$ from the code. The idea behind this is that one may check, up to a success probability $\frac{d(x,\ker(H))}{t}$, whether a given word $x$ is in the codespace or not, by checking a subset of the checks of $H$ chosen at random. If we see even one that is violated, we may conclude that $x \notin \ker(H)$, whereas if we see them all satisfied, we will return $x \in \ker(H)$. Since $\frac{|Hx|}{s}$ is the probability that a randomly chosen check is violated, in order to be correct with the desired success probability, we check $\OO\left(\frac{1}{\rho}\right)$ checks. We say also that the given classical code has locality $w$ if all of the rows and columns of $H$ have weight at most $w$\footnote{It is common in the case of classical codes to only define locality in terms of maximum row weight. However, it makes sense to define it here in terms of column weights too, particularly for when we take products. For example, with this definition, we can say that the locality of a quantum code defined via a homological product between a quantum code and a classical code scales with the worse of the two localities of the inputted codes; see Section \ref{homProduct}.}.

Quantum locally testable codes were first defined in \cite{aharonov2015quantum} for general quantum codes, although the definition specialises to the case of stabiliser codes as follows \cite{eldar2017local}. Given the generators of a stabiliser group $g_1, ..., g_m$, one may define a stabiliser code. With the projectors $\Pi_i = (I-g_i)/2$ being onto the $+1$-eigenspaces of each $g_i$, the codespace is exactly the 0-eigenspace of $\sum_{i=1}^m\Pi_i$. Calling this codespace $C$, the t-fattening of $C$ is the vector space defined by

\begin{equation}
    C_t = \text{span}\{(A_1\otimes...\otimes A_n)\ket{\psi} \text{ s.t. } \ket{\psi} \in C, |\{i \in \{1, ..., n\}: A_i \neq I\}| \leq t\}.
\end{equation}
We then let $\Pi_{C_t}$ be the projector onto the vector space $C_t$, from which we define the distance observable from $C$ as
\begin{equation}
    D_C = \sum_{t \geq 1}t\left(\Pi_{C_t}-\Pi_{C_{t-1}}\right).
\end{equation}
This allows us to say that the stabiliser code is locally testable with soundness $\rho$ if the operator inequality
\begin{equation}
    \frac{1}{m}\sum_{i=1}^m\Pi_i \succeq \frac{\rho}{n}D_C
\end{equation}
holds true. When dealing with CSS codes, which we do exclusively, the following lemma allows us to only worry about the above classical definition of local testability, rather than the general definition for stabiliser codes, or any other classical definition.

\begin{lemma}[Fact 17 of \cite{eldar2017local}]
    Let $C$ be a quantum CSS code defined by the parity-check matrices $H_X$ and $H_Z$. If $C$ is locally testable with soundness $\rho$, then the classical codes defined by the parity-check matrices $H_X$ and $H_Z$ are locally testable with soundness at least $\rho/2$. Conversely, if the classical codes defined by the parity-check matrices $H_X$ and $H_Z$ are locally testable with soundness $\rho$, then $C$ is locally testable with soundness $\rho$.
\end{lemma}

Throughout, we use $\rho_X$ to denote the soundness of the X - operators of a given CSS code, and $\rho_Z$ to denote the soundness of the Z - operators of a given CSS code. Therefore, $\rho_X$ is the soundness of the code associated with the Z - stabilisers, $\ker(H_Z)$, whereas $\rho_Z$ is the soundness of the code associated with the X - stabilisers, $\ker(H_X)$\footnote{This is the opposite convention to that used in our previous work, \cite{wills2023general}, however we feel it is the best one because it brings the soundnesses into line with the distances - $d_X$ is the minimum weight of a logical X - type error, etc.}. As with distance, dimension and locality, we measure the size of the soundness in terms of its asymptotic scaling in terms of the code length. Ideally, the soundness is $\Theta(1)$, and therefore does not ``scale''. A code with simultaneously optimally ``scaling'' distance, dimension, locality and soundness, is called a $c^3$ - LTC.

Many of our proofs will therefore be showing the classical soundness of $\ker(H)$ for some parity-check matrix $H \in \mathbb{F}_2^{s \times t}$. In order to show that this code has soundness $\rho$, we will usually simply consider a general bit string $u \in \mathbb{F}_2^t$ and show that $\frac{|Hu|}{s} \geq \rho\frac{d(u,\ker(H))}{t}$. However, it will be common for us, at the start of such a proof, to add some other bit string $v$ to $u$ and claim that we can do so without loss generality. We can do this because these bit strings will always be in the space $\ker(H)$ for whatever $H$ we are talking about; for such a bit string $v$, $d(u+v, \ker(H)) = d(u,\ker(H))$ and $H(u+v) = Hu$. In terms of Pauli operators, if we are considering $u$ to represent some Z - operator, we may add without loss of generality bit strings $v$ corresponding to Z - operators that commute with all X - stabilisers. This means both Z - logical operators and Z - stabilisers. Note also that if $u$ represents some Z - operator, then $|H_Xu|$ is the number of violated X - stabilisers. Many of our proofs will be in terms of a number of violated stabilisers.

We add a final note on our definition of classical local testability, Equation \eqref{cLTCDef}. This definition includes the `dimensions' of the parity-check matrix $s$ and $t$. This is now standard in the literature, but in other places they may be simply omitted, for example \cite{hastings2016weight}, where the definition is replaced with $|Hx| \geq \rho d(x, \ker(H))$. This makes some of our results look slightly different to other works, like \cite{hastings2016weight}, that do not include them. We emphasise that the difference is usually negligible, because typically $s$ will scale at the same rate, or nearly the same rate, as $t$. For us, it means that many of our soundness results look more complicated than they otherwise would, including many dimensions such as $N$, $N_X$ and $N_Z$. This is mostly more of an annoyance than anything, and reasonable assumptions on the codes in question, for example $N = \Theta(N_X) = \Theta(N_Z)$, which is common, may lead to such terms being forgettable. The most important features of our various soundness bounds are almost always not the dimensions.

\subsection{The Homological Product}\label{homProduct}

The homological product is a tool that has been used, along with its generalisations, for over a decade now to construct quantum CSS codes. It first arose implicitly in the work of Tillich and Z\'emor in 2009 \cite{tillich2013quantum} and has since been generalised to more elaborate product constructions~\cite{panteleev2021quantum,breuckmann2021balanced} that yield better parameters. The basic homological product takes in two chain complexes, $X$ and $Y$, and returns another chain complex $X \times Y$. Supposing $X$ and $Y$ are both chain complexes over $\mathbb{F}_2$, the resulting chain complex will be over the same field. We will also assume that $X_p = 0$ for $p < N$ and $Y_q = 0$ for $q < N$, for some $N \in \mathbb{Z}$, where $A_p$ denotes the $p$-th space of the chain complex $A$. The spaces of $X \times Y$ are then defined via

\begin{equation}
    (X \times Y)_p = \bigoplus_{i=0}^p X_i \otimes Y_{p-i}.
\end{equation}
Denoting the boundary maps of the two input chains as $\partial_p^X : X_p \rightarrow X_{p-1}$ and $\partial_p^Y : Y_p \rightarrow Y_{p-1}$, the boundary map $\partial_p^{X \times Y}$ of $X \times Y$ acts on product elements $u \otimes v \in X_i \otimes Y_{p-i}$ via

\begin{equation}
    \partial_p^{X \times Y}(u \otimes v) = \left(\partial_i^X(u) \otimes v\right) \oplus \left(u \otimes \partial_{p-i}^Y(v)\right).
\end{equation}
Its action on the whole space may then be determined by enforcing linearity. It may be checked that this definition indeed yields a valid chain complex $X \times Y$.

The only homological product that we will take in this work is that of a 3-term chain (a quantum CSS code) with a 2-term chain (a classical code) in order to achieve the `distance balancing' construction. This will arise as a step in Hastings' weight reduction methods in Section \ref{thickenAndChooseHeightsPrelimSection}\footnote{Note that while the construction employed is distance balancing, the intention behind using it is not to `balance distances'.}. The idea behind the distance balancing construction is to address the case in which a code has, for example, a very poor X - distance and a better Z - distance, or vice versa. Because the overall distance is defined as the minimum of these two, improving the X - distance at the expense of the Z - distance yields an improved overall distance.

In order to distance balance a quantum code given by a chain complex $C = \left(\mathbb{F}_2^{N_Z}\overset{H_Z^T}{\longrightarrow}\mathbb{F}_2^{N}\overset{H_X}{\longrightarrow}\mathbb{F}_2^{N_X}\right)$, we use an auxiliary classical code with a parity-check matrix $H \in \mathbb{F}_2^{s \times t}$, which may be thought of as a 2-term chain complex $R = \left(\mathbb{F}_2^t \overset{H}{\longrightarrow} \mathbb{F}_2^{s}\right)$. The first step in the distance balancing procedure is to take the dual of the latter chain complex to yield $R^* = \left(\mathbb{F}_2^s \overset{H^T}{\longrightarrow} \mathbb{F}_2^{t}\right)$. The second step is then to take the homological product of $C$ with $R^*$ to yield a 4-term chain complex $C \times R^*$. Finally, we define our distance balanced version of $C$, which we may call $C'$, via the chain complex formed from the last three terms of $C \times R^*$: $\left(C \times R^*\right)_2$, $\left(C \times R^*\right)_1$ and $\left(C \times R^*\right)_0$. This procedure was first considered by Hastings in 2016 with the ingredient classical code $R$ being only the repetition code \cite{hastings2016weight}, and then generalised to the case of all classical codes by Evra, Kaufman and Z\'emor in 2020 \cite{evra2022decodable}. The scaling of the non-soundness parameters is given in the following lemma. The requirement that the parity-check matrix of the classical code has independent checks (meaning its rows are linearly independent) is technically important, although generally non-restrictive.

\begin{lemma}[{\cite[Theorem 4.2]{evra2022decodable}}]In the above distance balancing construction, as long as the classical code has independent checks,
    \begin{align}
        \dim(C') &= \dim(C)\dim(R)\\
        d_X(C') &= d_X(C)d(R)\\
        d_Z(C') &= d_Z(C).
    \end{align}
Furthermore, under the reasonable assumption that $N_X = \OO(N)$, the number of physical qubits of $C'$ is $\Theta(Nt)$. Finally, the locality of the resulting quantum code scales with the worse of the localities of the inputted quantum code $C$ and the classical code $R$, where here by the locality of $R$ we mean the maximum row or column weight of $H$.
\label{DBresults}\end{lemma}

Notice that while the Z - distance appears to remain unchanged, it will become worse in real terms (in its asymptotic scaling) since the number of physical qubits rises. The above procedure may also be performed `dually' (meaning we take the dual of the quantum code before and after the procedure) to raise the Z - distance at the expense of the X - distance. As for the soundness, it was shown in our previous work \cite{wills2023general} that under reasonable assumptions on the quantum code, the soundness of $C'$ may be written in terms of that of $C$ as 

\begin{equation}
    \rho(C') = \Omega\left(\frac{\rho(C)}{t}\right)
\end{equation}
which gives a direct generalisation of the same result by Hastings \cite{hastings2016weight} for the case of $R$ being the t-bit repetition code.
\subsection{Weight Reduction}\label{wtRedPrelims}

In 2016, Hastings published the paper \cite{hastings2016weight} providing constructions that could weight reduce quantum codes i.e. reduce \qx, \qz, \wx and \wz to \OO(1). In that paper, this went via two constructions: the X-type generator splitting of Section IIA and the Z-type qubit splitting of Section IIC. The first can reduce \wx to \OO(1) and the second can reduce \qz to \OO(1). It is then claimed in the paper that both of these constructions preserve constant \wz and \qx\footnote{A construction is said to preserve the constancy of a parameter if the constancy of that parameter before the construction implies its constancy after the construction.}. This would lead to these two constructions alone being sufficient for a full weight reduction of all four parameters to constant, as follows:

\begin{enumerate}
    \item Use the X-type generator splitting to reduce \wx to \OO(1).
    \item Use the Z-type qubit splitting to reduce \qz to \OO(1). This construction does indeed preserve the constancy of \wx.
    \item Use the X-type generator splitting dually\footnote{To apply a construction dually means to take the dual of the CSS code chain complex before and after the procedure, thereby applying the construction with the roles of X and Z swapped.} to reduce \wz to \OO(1). This would then preserve the constancy of \wx and \qz.
    \item Use the Z-type qubit splitting dually to reduce \qx to \OO(1), preserving the constancy of all other parameters.
\end{enumerate}

Unfortunately, as observed by Z\'emor, Lemma 1 of that paper contains an error, meaning that, during (non-dual) X-type generator splitting, constant \wz is not necessarily preserved. These constructions alone are therefore insufficient for a full weight reduction of a quantum code of reducing all four parameters to constant, because in step 3, reducing \wz may cause \wx to grow large. Fortunately, Hastings updated the paper with the publication of \cite{hastings2021quantum}. In this paper, the X-type generator splitting and Z-type qubit splitting are both re-used, but supplemented with additional constructions. The X-type generator splitting is renamed to ``gauging'' and the Z-type qubit splitting is renamed to ``thickening and choosing heights''\footnote{This construction really comes in two steps: first thickening, and second choosing heights. These are two separate mathematical steps, but often the whole construction is referred to by Hastings as simply ``thickening''. It will be important for us to emphasise the existence of both of these steps later on and so we will refer to the whole construction as ``thickening and choosing heights''.}. The two added constructions are called ``copying'' and ``coning''. It is shown in \cite{hastings2021quantum} that, as well as reducing the weights, each of these constructions causes the number of physical qubits, dimension and distance to change by a factor that is at most a polynomial in the original weights and the logarithm of the original length. Then, these four together are sufficient for a full reduction of all four parameters to constant as follows:

\begin{enumerate}
    \item Use copying to reduce \qx to a constant.
    \item Use gauging to reduce \wx to a constant. This preserves constant \qx.
    \item Use thickening and choosing heights to reduce \qz to a constant. This preserves constant \qx and \wx.
    \item Use coning to reduce \wz to a constant. This preserves constant \qx, \wx and \qz\footnote{The reader may wonder whether all four of these constructions are, in fact, necessary. Could we not, for example, forget about copying, and instead use thickening and choosing heights dually to reduce \qx in the first step? This is not possible, because dual thickening and choosing heights cannot reduce \qx to a constant as long as there is a non-constant \wz. Overall, this order of four constructions is quite delicately constructed.}.
\end{enumerate}

We comment that Hastings presents copying and gauging together, in one combined construction, in \cite{hastings2021quantum}, because using these two constructions together leads to lower overhead. It will make sense for proving our results to extricate these two constructions, but note that potentially slightly better parameters in some cases could be obtained by considering them together. More generally, our primary goal is to demonstrate a full weight reduction of any qLTC, under reasonable restrictions, incurring a decrease in the soundness by a factor that is at most a polynomial in the original weights and the logarithm of the original length. Note, however, that in general this polynomial is likely to have very large exponents. Particular cases may get better parameters with closer analysis in cases where, for example, a subset of the four weight parameters are already constant, there are only a few stabilisers acting on a large number of qubits, or similar. Before showing the constructions, we finally comment that in presenting some of the Lemmas from \cite{hastings2016weight} and \cite{hastings2021quantum}, small changes have been made in the pursuit of greater detail, but no changes made are particularly consequential.

Section \ref{wtRedSection} contains the necessary analyses of the soundness of quantum codes under these constructions and there is a full analysis of putting all these constructions together in Section \ref{fullWeightRed} for a full weight reduction of a quantum locally testable code.

\subsubsection{Copying}\label{copyingPrelimSection}

Copying is a simple construction presented by Hastings to reduce \qx to a constant (indeed, it can reduce \qx to at most 3). Again, this is presented in \cite{hastings2021quantum} in conjunction with gauging, but we separate the two constructions here.

Given an initial code $C$ for which we wish to reduce \qx\footnote{We will say that a parameter is ``reduced'' as shorthand for saying that it is ``reduced to a constant''.} to produce another code $\tilde{C}$, copying goes as follows:

\begin{enumerate}
    \item Concatenate $C$ with a repetition code of length \qx in the X - basis i.e. a stabiliser code on \qx qubits with stabilisers $X_1X_2$, $X_2X_3$, ..., $X_{q_X-1}X_{q_X}$. Explicitly, \begin{enumerate} 
    \item For every qubit $q$ of $C$, $\tilde{C}$ will have qubits labelled by $(q,j)$ for $j \in \{1, ..., \qx\}$. These are the only qubits of $\tilde{C}$.
    \item For every Z - stabiliser in $C$, $Z_{q_1}...Z_{q_w}$, $\tilde{C}$ will have a Z - stabiliser $\tilde{Z}_{q_1}...\tilde{Z}_{q_w}$, where by $\tilde{Z}_{q_i}$ we mean $\Pi_{j=1}^{\qx} Z_{(q_i,j)}$. These are the only Z - stabilisers of $\tilde{C}$.
    \item For every X - stabiliser in $C$, $S = X_{q_1}...X_{q_w}$, $\tilde{C}$ will have an X - stabiliser $\tilde{S} = X_{(q_1,j(S,1))}...X_{(q_w,j(S,w))}$ for some $j(S,i) \in \{1, ..., \qx\}$ to be specified. These are called the copied X - stabilisers. Additionally, for each qubit $q$ in $C$, $\tilde{C}$ will also have stabilisers $X_{(q,1)}X_{(q,2)}, X_{(q,2)}X_{(q,3)}, ..., X_{(q,q_x - 1)}X_{(q,q_X)}$. These are called the new X - stabilisers. The copied and new X - stabilisers are the only X - stabilisers of $\tilde{C}$.
    \end{enumerate}
    \item The $j(S,i)$ are chosen such that each qubit is acted on by at most one copied X - stabiliser.
\end{enumerate}

We explain the idea here in a more worded way, through which it will become clear how the choice in step 2 is possible. The logical Z - operator of the repetition code in the X - basis is the product of Z on all qubits i.e. $\Pi_{j=1}^{\qx}Z_j$. For the logical X - operator, we have a choice of an individual X operator on any one of the qubits i.e. $X_j$ for any $j$. Therefore, when we concatenate our code $C$ with this repetition code, each Z - stabiliser in $\tilde{C}$ is formed by simply repeating every physical Z in any Z - stabiliser of $C$ \qx times. On the other hand, each copied X - stabiliser will have the same weight as the old X - stabiliser it came from - we just take care that for every X - stabiliser acting on a given qubit $q$ in $C$, of which there are at most \qx, each resulting copied X - stabiliser acts on a different $(q,j)$.

We then have the following.

\begin{lemma}[Extricated from Lemma 1, on Copying and Gauging, in \cite{hastings2021quantum}] Let parameters with and without a tilde be those of $\tilde{C}$ and $C$ respectively, i.e. those after and before copying. Then

\begin{enumerate}
    \item $\tilde{N} = N\qx$, $\tilde{N}_X = N_X + \Theta(N\qx)$, $\tilde{N}_Z = N_Z$.
    \item $\tilde{K}$ = K.
    \item $\tilde{q}_X \leq 3$.
    \item $\tilde{w}_X = \max(w_X,2)$.
    \item $\tilde{q}_Z = \qz$.
    \item $\tilde{w}_Z = q_X\wz$.
    \item $\tilde{d}_Z = d_Zq_X$.
    \item $\tilde{d}_X = d_X$.
\end{enumerate}
\end{lemma}

In Section \ref{copyingSection}, we will analyse the soundness of $\tilde{C}$ in terms of that of $C$.

\subsubsection{Gauging}

Gauging has already been fully analysed on its own in \cite{hastings2016weight} and so we take it here as a black-box construction to reduce \wx. The reader can either see this construction in conjunction with copying, in \cite{hastings2021quantum} or, what we recommend, on its own, under the name of ``X-type generator splitting'' in Section IIA of \cite{hastings2016weight}, and taking our Lemma \ref{correctedLemma} to replace Lemma 1 in that paper. The soundness of gauging was also analysed in \cite{hastings2016weight} in Lemma 6. This is the reason that we wanted to separate copying and gauging: gauging already has a soundness analysis while copying does not. Note that using the presented ingredients, it would be straightforward to analyse the soundness of the combined copying and gauging procedure of \cite{hastings2021quantum}, and again this would likely lead to better parameters, but we do not do this to avoid repeating too much of the gauging soundness proof.

\begin{lemma}[From \cite{hastings2016weight}, Lemma 1 ($\tilde{w}_Z$ is corrected) and Lemma 6]\label{correctedLemma} Let parameters with and without a tilde be those of the code after and before gauging respectively. Then
\begin{enumerate}
    \item $\tilde{N} = \OO(N\qx)$ $(\tilde{N} \geq N)$, $\tilde{N}_X = \OO(N_X + N\qx)$, $\tilde{N}_Z = N_Z$.
    \item $\tilde{K} = K$.
    \item $\tilde{q}_X = \max(q_X,2)$.
    \item $\tilde{w}_X \leq 3$.
    \item $\tilde{q}_Z \leq w_X\qz$.
    \item $\tilde{w}_Z \leq \wz(1+\wx\qx)$.
    \item $\tilde{d}_Z \geq d_Z$.
    \item $\tilde{d}_X \geq d_X \Theta\left(\frac{1}{\wx}\right)$.
    \item $\tilde{\rho}_Z \geq \frac{\tilde{N}}{\tilde{N}_X}\frac{N_X}{N}\rho_Z\frac{1}{1+\wx\left(\qx+\frac{N_X}{N}\rho_Z\right)}$.
    \item $\tilde{\rho}_X \geq \frac{\tilde{N}}{N} \rho_X$.
\end{enumerate}
Note then that under the reasonable assumptions that $N_X = \Theta(N)$, $\rho_Z = \OO(1)$, and assuming that copying has already been performed to get $\qx = \OO(1)$, we have $\tilde{\rho}_Z = \Omega\left(\frac{\rho_Z}{\wx}\right)$ and $\tilde{\rho}_X \geq \rho_X$.
\end{lemma}
\subsubsection{Thickening and Choosing Heights}\label{thickenAndChooseHeightsPrelimSection}

Thickening and Choosing Heights was presented as ``Z-type Qubit Splitting'' in Section IIC of \cite{hastings2016weight} and under the present name in Section II of \cite{hastings2021quantum} and allows for the reduction of the parameter \qz. Both of these papers present a full analysis of the construction's non-soundness parameters. Further, \cite{hastings2016weight} presents an analysis of the soundness of the ``thickening only'' portion of the construction, without a choice of heights being made, in its Lemma 7
. Therefore, for thickening and choosing heights, our contribution is to analyse the soundness under the choice of heights. It is because of this that we emphasise the two steps in the name ``thickening and choosing heights'', as mentioned, rather than simply ``thickening'', which is how Hastings often refers to both steps collectively.

The first step of this construction, thickening, is exactly the distance balancing construction, explained in Section \ref{homProduct}, where the classical code being used is the repetition code of length $l$\footnote{Indeed, the discovery of thickening for weight reduction doubled as the discovery of the distance balancing procedure. This was followed by the distance balancing method of Evra et al. \cite{evra2022decodable} that allowed for the classical code to be general, not just the repetition code. This allows good classical LDPC codes to be used, meaning that the dimension of the quantum code does not drop when distance balancing.}. We therefore take the chain complex of the inputted quantum code: $C = \left(\mathbb{F}_2^{N_Z}\overset{H_Z^T}{\longrightarrow}\mathbb{F}_2^{N}\overset{H_X}{\longrightarrow}\mathbb{F}_2^{N_X}\right)$ and the chain complex of the classical repetition code, $R = \left(\mathbb{F}_2^l \overset{H}{\longrightarrow} \mathbb{F}_2^{l-1}\right)$, where

\begin{equation}
    H = \begin{pmatrix}1 & 1 & 0 & 0 &\cdots&0 & 0\\
    0 & 1 & 1 & 0 &\cdots & 0 & 0\\
    0 & 0 & 1 & 1 & \cdots & 0 & 0\\
    &&&&\ddots&&\\
    0 & 0 & 0 & 0 & \cdots & 1 & 1\end{pmatrix}\label{repCodePCM}
\end{equation}
and take the homological product of $C$ with the dual of the latter complex, $R^*$. This gives us the four-term chain complex shown in Figure \ref{thickenCC}.

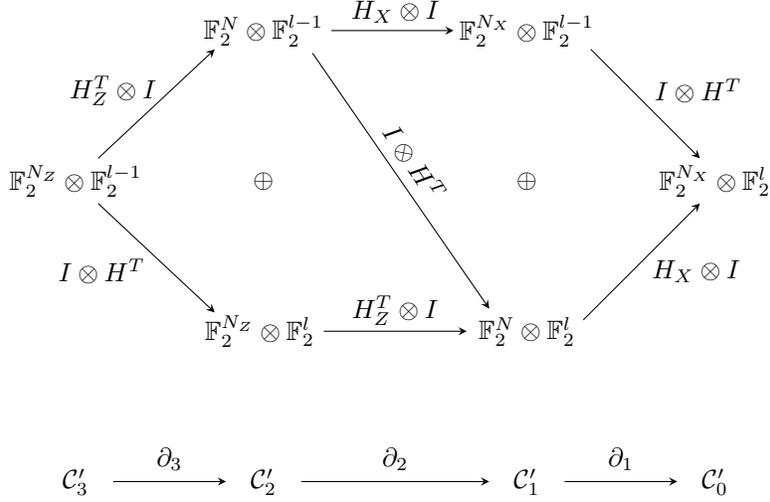
\begin{figure}[h]
\begin{center}
\begin{tikzpicture}

\node at (2.5,0) {$\mathbb{F}_2^{N_Z}\otimes \mathbb{F}_2^{l-1}$};
\node at (5,2) {$\mathbb{F}_2^{N}\otimes \mathbb{F}_2^{l-1}$};
\node at (8.5,2) {$\mathbb{F}_2^{N_X}\otimes \mathbb{F}_2^{l-1}$};
\node at (11,0) {$\mathbb{F}_2^{N_X}\otimes \mathbb{F}_2^l$};
\node at (4.95,-2) {$\mathbb{F}_2^{N_Z}\otimes \mathbb{F}_2^l$};
\node at (8.5,-2) {$\mathbb{F}_2^{N}\otimes \mathbb{F}_2^l$};
\node at (5,0) {$\oplus$};
\node at (8.5,0) {$\oplus$};

\draw [-stealth](2.8,0.3) -- (4.35,1.75);
\draw [-stealth](2.8,-0.3) -- (4.35,-1.75);
\draw [-stealth](5.9,2) -- (7.525,2);
\draw [-stealth](5.8,-2) -- (7.7,-2);
\draw [-stealth](9.35,1.75) -- (10.8,0.3);
\draw [-stealth](9.25,-1.85) -- (10.8,-0.3);
\draw [-stealth](5.65,1.7) -- (8,-1.7);

\node at (3,1.2) {$H_Z^T \otimes I$};
\node at (2.85,-1.2) {$I \otimes H^T$};
\node at (6.75,-1.725) {$H_Z^T \otimes I$};
\node at (6.75,2.25) {$H_X \otimes I$};
\node [rotate = -56] at (7,0.25) {$I \otimes H^T$};
\node at (10.78,1.2) {$I \otimes H^T$};
\node at (10.75,-1.2) {$H_X \otimes I$};

\node at (2.5,-4) {$\mathcal{C}'_3$};
\node at (5,-4) {$\mathcal{C}'_2$};
\node at (8.5,-4) {$\mathcal{C}'_1$};
\node at (11,-4) {$\mathcal{C}'_0$};

\draw [-stealth](3,-4) -- (4.5,-4);
\draw [-stealth](5.5,-4) -- (8,-4);
\draw [-stealth](9,-4) -- (10.5,-4);

\node at (3.75,-3.7) {$\partial_3$};
\node at (6.75,-3.7) {$\partial_2$};
\node at (9.75,-3.7) {$\partial_1$};

\end{tikzpicture}
\end{center}\caption{The homological product of $C$ with $R^*$, denoted $C'$. These diagrams can be ``read off vertically'', meaning that, for example $\mathcal{C}_2' = \left(\mathbb{F}_2^{N_Z}\otimes\mathbb{F}_2^l\right) \oplus \left(\mathbb{F}_2^N\otimes\mathbb{F}_2^{l-1}\right)$. The linear maps between $\mathcal{C}_i'$ can also be read off vertically, summing the images of any maps that map to the same space.}\label{thickenCC}
\end{figure}

Finally, we define another quantum code from the final three terms of this four-term chain complex i.e. $\mathcal{C}_2'\overset{\partial_2}{\longrightarrow}\mathcal{C}_1'\overset{\partial_1}{\longrightarrow}\mathcal{C}_0'$. For clarity, Z - stabilisers are associated with basis elements of $\mathcal{C}_2'$, qubits are associated with basis elements of $\mathcal{C}_1'$ and X - stabilisers are associated with basis elements of $\mathcal{C}_0'$.

At this stage, \qz has not been reduced. Indeed, there will be a qubit represented by some basis element in $\mathbb{F}_2^N \otimes \mathbb{F}_2^l$ that will be acted on by \qz Z - stabilisers represented by basis elements in $\mathbb{F}_2^{N_Z}\otimes\mathbb{F}_2^l$. What we now do, however, is to remove most of the Z - stabilisers associated with basis elements of $\mathbb{F}_2^{N_Z}\otimes\mathbb{F}_2^l$. This is the `choosing heights' step. Writing the standard basis of $\mathbb{F}_2^l$ as $\left(w_i\right)_{i=1}^l$, for each basis element $c \in \mathbb{F}_2^{N_Z}$, we pick one $k \in \{1, ..., l\}$ and remove all basis elements $c \otimes w_{k'}$ for $k' \neq k$, therefore keeping $c \otimes w_k$\footnote{The reason for the name `choosing heights' is that the chain complex for the repetition code may also be thought of as a cellulation of an interval, and so taking this homological product can be imagined as geometrically thickening the quantum code complex. In this second `height choice' step, we are choosing a coordinate along that interval at which to attach all of the 2-cells represented by basis elements of $\mathbb{F}_2^{N_Z}$, rather than attaching them at every coordinate/height.}.

After the choice of heights is made, the Z - stabilisers in $\mathbb{F}_2^{N_Z}\otimes\mathbb{F}_2^l$ that do act on a qubit $x_a \otimes w_k$ in $\mathbb{F}_2^N\otimes\mathbb{F}_2^l$ are exactly those from the linear combination of basis elements $H_Z(x_a)\otimes w_k$ that are kept, rather than removed. Hastings shows in both \cite{hastings2016weight} and \cite{hastings2021quantum} that it is possible to take a large enough $l$, and make a particular height choice, such that the \qz resulting from this procedure is constant. Indeed, Hastings gives the following lemma.

\begin{lemma}[Lemmas 2 and 4 of \cite{hastings2021quantum}] Let parameters with and without a tilde be respectively those of the code after and before ``thickening and choosing heights''. Then, for any $\epsilon > 0$, one may take $l = \Theta\left(\qz^{1+\epsilon}\min(\qz\wz,N)^{\OO(\epsilon)}\right)$ and then there is a choice of heights such that
    \begin{enumerate}
        \item $\tilde{N} = Nl + N_X(l-1)$, $\tilde{N}_X = N_Xl$, $\tilde{N}_Z = N_Z + N(l-1)$.
        \item $\tilde{K} = K$.
        \item $\tilde{q}_X = \max(\qx,2)$.
        \item $\tilde{w}_X = \wx + 2$ if $l \geq 3$ and $\tilde{w}_X = \wx + 1$ if $l = 2$.
        \item $\tilde{q}_Z = \max\left(\OO(1), \wx\right)$.
        \item $\tilde{w}_Z = \max(\wz,\qx+2)$.
        \item $\tilde{d}_Z = d_Z$.
        \item $\tilde{d}_X = ld_X$.
    \end{enumerate}
\end{lemma}

As for the soundness, Hastings analyses the thickening portion only in \cite{hastings2016weight}. We also note that, since thickening is exactly distance balancing with the repetition code, the following lemma can also be seen as a special case of our previous result \cite{wills2023general}, which applies to distance balancing with any classical code.

\begin{lemma}[Lemma 7 of \cite{hastings2016weight}] Let parameters with and without a prime be respectively those of the code after and before thickening (only). Then, 
    \begin{enumerate}
        \item $\rho_Z' \geq \frac{N'}{N_X'}\min\left(\frac{N_X\rho_Z}{N},1\right)\frac{1}{l}$.
        \item $\rho_X' \geq \frac{N'}{N_Z'}\min\left(\frac{N_Z\rho_X}{N},1\right)\frac{1}{l}$.
    \end{enumerate}
    where $N' = Nl + N_X(l-1)$, $N_X' = N_Xl$ and $N_Z' = N_Zl + N(l-1)$. Note then that under the reasonable assumptions that $N = \Theta\left(N_X\right) = \Theta\left(N_Z\right)$ and $\rho_Z$, $\rho_X = \OO(1)$, we have $\rho_Z' = \Omega\left(\frac{\rho_Z}{l}\right)$ and $\rho_X' = \Omega\left(\frac{\rho_X}{l}\right)$.\label{thickeningSoundnessLemma}
\end{lemma}
The analysis of the soundness of the code under the ``choosing heights'' step can be found in Section \ref{thickeningAndChoosingHeightsSection}.
\subsubsection{Coning}\label{coningPrelimSection}

The coning construction is introduced in \cite{hastings2021quantum} and forms the primary technical contribution of that paper. This is the last of the four constructions and can reduce \wz to constant, while preserving constant \qx, \wx and \qz. The soundness of this construction is not analysed prior to the present work: this analysis will be performed in Section \ref{coningSection}. The reader may find figures \ref{coneCodeChainComplex} to \ref{reducedConeCodeChainComplex} in that section helpful for the understanding of the exposition in this section.

Hastings gives the coning construction for all quantum CSS codes, but gives the construction for ``reasonable'' and ``unreasonable'' codes (to be defined) separately. The construction is most simple in the case of reasonable codes, and these are the only codes that we will consider, although this is not a major restriction. A code is called ``reasonable'' if there is no Z - logical operator whose support is contained entirely within the support of some Z - stabiliser. Therefore, any code whose Z - distance is at least equal to \wz is reasonable. Any reasonable code has the following property:

\begin{claim}\label{reasonableClaim}
    For a reasonable code, any Z - operator which is supported only on a subset of the support of one of the Z - stabilisers, and commutes with all the X - stabilisers, is a member of the stabiliser group\footnote{We recall the difference between a ``stabiliser'' and a ``member of the stabiliser group'', as defined in Section \ref{CSSPreliminaries}.}.
\end{claim}

The proof of this is immediate, because any Z - operator commuting with all X - stabilisers must be a Z - logical operator, or a member of the stabiliser group, by definition.

As well as only considering reasonable codes, for the sake of simplicity in our soundness proofs, we will present a slightly abridged version of Hastings' coning construction in which there are no ``direct Z - stabilisers'' (to use the language of \cite{hastings2021quantum}). Practically, what this means, is that we work with a ``worst-case scenario'' in which we weight reduce all Z - stabilisers. In a specific use-case, where there are only a few, or one, high weight Z - stabilisers requiring weight reduction, an analysis more specific to that particular case would likely yield better bounds; see the original construction in \cite{hastings2021quantum} for this. We will attempt to keep our notation and language as similar as possible to that paper for the sake of comparability.

To perform the coning construction, we must first start by defining new ``auxiliary'' chain complexes that will ultimately allow us to ``induce'' the effect of the high weight Z - stabilisers by only using constant weight Z - stabilisers. Let us start by defining $N_Z$ two-term chain complexes, $\left(\BB_i\right)_{i=1}^{N_Z}$, which we will say each have 1-cells and 0-cells. For each Z - stabiliser of the original code, $C$, denote the set of qubits on which that Z - stabiliser acts as $Q_i$. The 1-cells of \BBI are then defined to be $Q_i$ and the space they span is denoted $\QQ_i$, i.e. $\left(\BB_i\right)_1 = \QQ_i$, where all our vector spaces are over the binary field.

Consider now, for each $i$, the set $S_i$, of X - stabilisers with some support (not necessarily their whole support) on $Q_i$. By assumption on commutativity, each element of $S_i$ must have even support on $Q_i$. Therefore, for each stabiliser $S \in S_i$, we can choose some pairing of the qubits in $Q_i$ on which it acts. Define the 0-cells in \BBI to be the set of tuples $(S,j)$ such that $S \in S_i$ and $j$ is a pair of qubits in $Q_i$ on which $S$ acts. We will denote the set of 0-cells of \BBI by $X_i$ and the binary vector space they span as \XXI $\left(=\left(\BB_i\right)_0\right)$. The boundary operator of the complex going from $\left(\BB_i\right)_1$ to $\left(\BB_i\right)_0$ is then the obvious one: the boundary of a qubit in $\left(\BB_i\right)_1$ is then given by all the tuples $(S,j)$ for which the given qubit is in the pair $j$. Notice that this means that the size of the coboundary of any 0-cell is equal to 2.

We will now define a further chain complex, \BBARI, from each \BBI. \BBARI will be a 3-term chain complex, considered to have 1-cells, 0-cells and -1-cells. In \BBARI, the 1-cells, 0-cells and the boundary operator going between them will be identical to those of \BBI. We add in -1-cells to make \BBARI, as well as a boundary operator going from the 0-cells to the -1-cells. The purpose of this is to construct a chain complex, \BBARI, which has trivial zeroth homology.

Let us see how to do this. First observe the following. Given that the coboundary of any 0-cell in \BBI has size 2, we can define a graph $G_i$. This graph will have vertices identified with the 1-cells of \BBI and edges identified with the 0-cells of \BBI. An edge will then join the two 1-cells which it has in its coboundary. With this, we observe that the elements of zeroth cohomology in \BBI\footnote{By an element of zeroth cohomology in \BBI, we mean a co-closed 0-chain.} then correspond exactly to closed edge chains in $G_i$. Employing the decongestion lemma \cite{freedman2021building}, Hastings then argues that a basis of simple cycles can be found for these closed edge chains. The ``weight'' of one of these simple cycles is the number of edges that it contains; the sum of the weights of the elements of this basis turns out to be $\OO\left(\left|Q_i\right|\log\left|Q_i\right|\right)$, noting that the graph has $\left|Q_i\right|$ vertices. Further, it turns out that each edge appears at most $\OO(\log\left|Q_i\right|)$ times in the basis.

We may then construct \BBARI by adding a set of -1-cells given by this basis. The coboundary of any of these -1-cells is given by the set of edges in that simple cycle (recalling that 0-cells correspond to edges). The set of -1-cells is denoted $R_i$ and the binary vector space which they span is denoted $\RRI$ $\left(=\left(\BBARI\right)_{-1}\right)$. \BBARI is constructed to have trivial zeroth cohomology, because any co-closed 0-chain corresponds to a closed edge chain in the graph $G_i$, which then can be written as some combination of simple cycles, meaning that the co-closed 0-chain is also a 0-coboundary. Because \BBARI has trivial zeroth cohomology, it has trivial zeroth homology.

Adding these -1-cells to form \BBARI admits a geometrical interpretation. As explained, the complex \BBI can be thought of geometrically as the graph $G_i$. Adding these simple cycles as -1-cells is then like geometrically adding 2-discs to the graph $G_i$ to form a complex $\SSI$. This is potentially a little confusing as the vertices, edges and discs in $\SSI$ are respectively 0-, 1- and 2-dimensional geometrical objects, but correspond to 1-, 0-, and -1-cells in $\BBARI$; nevertheless, the geometrical picture will be useful.

From here, we may define the ``cone code'', $C_{cone}$. This will be a code that has the high weight Z - stabilisers induced by low weight Z - stabilisers, and will have all of the old X - stabilisers and qubits, but will now have additional qubits and X - stabilisers. Unfortunately, these added X - stabilisers may have high weight and, further, the new qubits may be acted on by a large (non-constant) number of new X - stabilisers. Fortunately, we will display a remedy for this by \textit{reducing} the cone code to form the ``reduced cone code''.

The cone code has spaces corresponding to Z - stabilisers, qubits and X - stabilisers respectively as follows:

\begin{align}
    \left(\mathcal{C}_{cone}\right)_2 &= \bigoplus_{i=1}^{N_Z} \mathcal{Q}_i\\
    \left(\mathcal{C}_{cone}\right)_1 &= \mathcal{C}_1 \oplus \left(\bigoplus_{i=1}^{N_Z}\mathcal{X}_i\right)\\
    \left(\mathcal{C}_{cone}\right)_0 &= \mathcal{C}_0 \oplus \left(\bigoplus_{i=1}^{N_Z}\mathcal{R}_i\right)
\end{align}
where $\mathcal{C}_1$ and $\mathcal{C}_0$ are the spaces of qubits and X - stabilisers of the original code, respectively. We may then define the boundary maps of the cone code by defining maps between the spaces $\QQI$, $\mathcal{C}_1$, $\XXI$, $\mathcal{C}_0$ and $\RRI$, and then the boundary maps follow as usual by summing images in the same space. The chain complex of the cone code is shown in Figure \ref{coneCodeChainComplex} in Section \ref{coningSection}. There are maps as follows:

\begin{enumerate}
    \item There is a map from \QQI to \XXI for each $i$ as inherited from \BBARI, i.e. $\partial_1^{\BBARI}$.
    \item There is a map from \XXI to \RRI for each $i$ as inherited from \BBARI, i.e. $\partial_0^{\BBARI}$.
    \item There is a map from $\mathcal{C}_1$ to $\mathcal{C}_0$ as inherited from the original code $C$, i.e. $H_X$.
    \item There is the obvious projection map from \QQI to $\mathcal{C}_1$ for each $i$ (any qubit $q \in \QQI$ is mapped to $q \in \mathcal{C}_1$), denoted $\pi_q^{(i)}$.
    \item There is the obvious projection map from \XXI to $\mathcal{C}_0$ for each $i$ (any tuple $(S,j)$ is mapped to $S \in \mathcal{C}_0$), denoted $\pi_X^{(i)}$.
\end{enumerate}
The boundary maps of the cone code, $\partial_i^{C_{cone}}$, are formed from these maps. At this point, the cone code does have the action of the high weight Z - stabilisers induced by low weight Z - stabilisers, but now, as mentioned, there are two problems. First, X - stabilisers in \RRI may act on up to $\left|Q_i\right|$ qubits in \XXI, making them high weight stabilisers. Second, qubits in \XXI may be acted on by up to $\OO(\log\left|Q_i\right|)$ stabilisers in \RRI, meaning that they are potentially acted on by a non-constant number of X - stabilisers. It is possible to ``reduce the cone code'' to solve both of these problems. The second problem has an easy solution, and we solve this one first. The solution is simply to thicken and choose heights dually, to reduce \qx (where we only need to choose heights for the stabilisers coming from \RRI, not $\mathcal{C}_0$). This procedure is described (non-dually, for the reduction of \qz) in Section \ref{thickenAndChooseHeightsPrelimSection}. Hastings argues that it is possible that for any $\epsilon > 0$, one may thicken by an amount $l = \Theta\left(\log\left(\wz\right)^{2+2\epsilon}\wz^\epsilon\right)$ and choose heights to make it such that at most one stabiliser in $\RRI\otimes\mathcal{E}_0$ is attached to any given qubit in $\XXI\otimes\mathcal{E}_0$. The resulting chain complex is depicted in Figure \ref{thickendAndHeightChosenConeCodeChainComplex}.

All that is left to address is some potentially high weight X - stabilisers, coming from \RRI, each attached at one height. We may think of these geometrically as potentially high weight discs, attached to a large number of edges coming from \XXI. The idea is to ``cellulate'' these discs, adding edges going across the disc, and potentially vertices in the disc, to break up the large disc into many smaller faces, thus re-introducing a constant \wx. There are many ways this cellulation could be done. One example given by Hastings is, if the vertices within a disc are labelled $0, ..., w-1$, one may add an edge between vertices $j$ and $w-j$ for all $j$ with $0 < j < w-j-1$. For this cellulation, no additional vertices need to be added. This particular cellulation will be important to us, and so it is depicted in Figure \ref{cellulationFigure}, and the chain complex that results is depicted in Figure \ref{reducedConeCodeChainComplex}. Note that every disc must be cellulated and they are each cellulated at whatever height they are attached, thus the edges (and vertices if they are added) are only added at the heights at which they are needed. The end result is the ``reduced cone code''.

From this construction, we have the following:

\begin{lemma}[Adapted\footnote{Note that we have the lemma in this form because we do not take any direct stabilisers, as mentioned. Furthermore, Hastings gives $\tilde{d}_Z$ in terms of a parameter $\lambda$, for which we take a trivial lower bound $\frac{1}{\wz}$.} from Lemma 8 of \cite{hastings2021quantum}]\label{coningNonSoundnessLemma} Let parameters with and without a tilde be respectively those of the reduced cone code and the code pre-coning. Then, for any $\epsilon > 0$, letting $l = \Theta\left(\log\left(\wz\right)^{2+2\epsilon}\wz^\epsilon\right)$, we have

    \begin{enumerate}
        \item $\tilde{N} = \OO\left(l\left(N+\wz\qx\wx N_Z\right)\right)$, $\tilde{N}_X = \OO\left(l\left(N + \wz\qx\wx N_Z + N_X\right)\right)$, $\tilde{N}_Z = \OO(N_Z\wz l)$.
        \item $\tilde{K} = K$.
        \item $\tilde{q}_X = \max\left(\qx+\OO(1),\OO(1)\right)$.
        \item $\tilde{w}_X = \max\left(\OO\left(\wx^2\qz\right),\OO(1)\right)$.
        \item $\tilde{q}_Z = \max\left(\qz,\OO(1)\right)$.
        \item $\tilde{w}_Z \leq \qx + \OO(1)$.
        \item $\tilde{d}_Z \geq \frac{d_Zl}{\wz}$.
        \item $\tilde{d}_X \geq d_X$.
    \end{enumerate}
\end{lemma}
The analysis of the soundness of this construction, under the assumption of the code being reasonable, will be performed in Section \ref{coningSection}. We also mention that in order to draw our conclusions on the soundness of the code under coning, we take a larger value of $l$ here, $l = \Theta\left(\wz\log(\wz)\qx\right)$, where these un-decorated parameters are those of the code pre-coning, as is explained in Section \ref{reducedConeCodeSoundnessSection}.

\section{Quantum Soundness Under Weight Reduction}\label{wtRedSection}

Our objective is to analyse the soundness of a quantum code under each of the weight reduction constructions, as necessary, so that we may obtain a lower bound on the soundness of a quantum code after a full weight reduction in terms of its soundness before weight reduction. A full analysis of a weight reduction of a qLTC will be presented in Section \ref{fullWeightRed}.

It will be common to represent X and Z - operators by bit strings in the obvious way: a Z - operator (for example) on $\tilde{N}$ qubits may be represented by a bit string $u \in \mathbb{F}_2^{\tilde{N}}$, where there is a Pauli Z in the Z - operator exactly wherever $u$ has a 1. With this equivalence, it makes sense to conflate our thinking about $u$ as a bit string and as an operator, often conflating the bits on which the bit string is defined with the qubits on which the operator is defined. We then refer to $u$ as being ``in a given quantum code'' if it commutes with all the X - stabilisers of that quantum code i.e. it is in the classical code defined by the X - stabilisers of that code. We will also talk about ``putting u in a quantum code'' to refer to the process of flipping bits in $u$ such that it is ``in the quantum code''. We will often specify a process by which this will occur. For example, many of our soundness proofs will go via a line of reasoning as follows:

\begin{enumerate}
    \item Specify a method by which we put a bit string in a post-construction code, for example the post-copying code. The number of bits flipped in this method provides an upper bound on the distance of $u$ from the post-construction code.

    \item We will then argue that flipping this many bits in $u$ implies that a certain number of stabilisers must have been violated by $u$ in the first place, from which we can deduce soundness.
    
\end{enumerate}
The equivalent of all of the above will be true for X - operators also. Note that the soundness in question is deduced because in our definition of soundness, $\frac{|Hx|}{m} \geq \rho \frac{d(x,\ker H)}{n}$, $|Hx|$ translates to a number of violated stabilisers of a given type.

Finally, note that we give our soundness lower bounds in full for the sake of giving the greatest detail, but we say that this is at risk of hiding the most important features of each bound. In particular, the presence of code lengths and numbers of stabilisers often adds greater detail than will be needed - generally $N$, $N_X$ and $N_Z$ will scale at the same rate, or nearly at the same rate. Assumptions like this, as well as the reasonable assumption that the soundness is at most a constant at all stages of the construction, will bring out the important features of the lower bound, which is how many factors of each weight the given soundness is decreased by in some step of the construction.
\subsection{Copying}\label{copyingSection}

We now analyse the soundness of a quantum code under the copying construction of Section \ref{copyingPrelimSection}. Let us recall here that $\rho_X$ is the soundness of the classical code associated with the Z - stabilisers and $\rho_Z$ is the soundness of the classical code associated with the X - stabilisers.

\begin{lemma}
Let parameters with and without a tilde be those of the code after and before copying respectively. Then

\begin{equation}
    \tilde{\rho}_Z \geq \frac{\tilde{N}}{\tilde{N}_X}\frac{\rho_Z}{\qx\rho_Z+\frac{N}{N_X}\qx\left(\qx^2+1\right)}.
\end{equation}
\end{lemma}
\begin{proof}
    Let us denote the quantum code before and after copying as $C$ and $\tilde{C}$ respectively. Consider a Z - operator on the new set of qubits represented by some bit string $u \in \mathbb{F}_2^{\tilde{N}}$. We label the $\tilde{N}$ physical qubits of the new code as $(q,j)$ where $q$ labels a qubit in $C$ and $j \in \{1, ..., \qx\}$. For a given $q$, we refer to all the qubits $(q,j)$ as a block of qubits in $\tilde{C}$. We can then label blocks of qubits in $\tilde{C}$ by $q$ i.e. the same set of labels as those that label qubits in $C$.

    We say that $u$ is uniform on a given block $q$ if

    \begin{equation}
        u(q,j) = 1 \text{ for some $j$ } \implies u(q,i) = 1 \text{ for all } i \in \{1, ..., \qx\}
    \end{equation}
    i.e. it is either all zeros or all ones on that block. Then, in order for a bit string $u$ to be in the code $\tilde{C}$, it is necessary that it is uniform on every block. This is because it must commute with all the stabilisers of the form $X_{(q,i)}X_{(q,i+1)}$. Given a bit string $u$ that is uniform on every block, we may define a bit string $\hat{u} \in \mathbb{F}_2^N$, where $\hat{u}(q) = u(q,j)$. Then $u$ is in the code if and only if it is uniform on every block and the corresponding $\hat{u}$ is in the code $C$ (the second condition ensures commutativity with the encoded stabilisers of $C$).

    We can put $u$ into the code $\tilde{C}$ as follows. For every block $q$ on which $u$ is not uniform, set $u(q,k) = 1$ for all $k$ if at least half of the bits in that block are 1, and set $u(q,k) = 0$ for all $k$ otherwise. In doing this, suppose that $f$ bits are flipped. The resulting bit string is uniform on every block and so we may define a corresponding $\hat{u}$. It is then possible to put $\hat{u}$ in the code $C$ by flipping as few bits as are needed. This entails flipping $d(\hat{u}, \ker H_X)$ bits in $\hat{u}$, where $\ker H_X$ refers to the classical code defined by the X - stabilisers of $C$. We may then finish putting $u$ in the code by flipping, for every bit $q$ in $\hat{u}$ that gets flipped, all the bits $(q,i)$ in $u$ for every $i$. Denoting also the classical code defined by the X - stabilisers of $\tilde{C}$ as $\ker \tilde{H}_X$, we then have

    \begin{equation}
        d(u, \ker \tilde{H}_X) \leq f + \qx d(\hat{u},\ker H_X).
    \end{equation}
    Let us consider two cases. First, consider the case for which $f \geq \alpha\qx d(\hat{u},\ker H_X)$, for some positive number $\alpha$ to be later determined. Then

    \begin{equation}
        d(u, \ker \tilde{H}_X) \leq f\left(\frac{\alpha+1}{\alpha}\right).
    \end{equation}
    When we flip the $f$ bits in $u$, there must be bits flipped in at least $\frac{f}{\qx}$ blocks. If $u$ is non-uniform in a given block $q$, at least one of the stabilisers of the form $X_{(q,i)}X_{(q,i+1)}$ is violated. Therefore, $u$ violates at least

    \begin{equation}
        \frac{f}{\qx} \geq \frac{\alpha}{\alpha+1}\frac{d(u,\ker \tilde{H}_X)}{\qx}\label{copyingZVio1}
    \end{equation}
    stabilisers. Now consider the case $f \leq \alpha \qx d(\hat{u},\ker H_X)$, for which we have

    \begin{equation}
        d(u, \ker \tilde{H}_X) \leq (1+\alpha)\qx d(\hat{u},\ker H_X).
    \end{equation}
    $\hat{u}$ is at a distance $d(\hat{u},\ker H_X)$ from $\ker H_X$, and therefore violates at least $\frac{\rho_ZN_X}{N}d(\hat{u},\ker H_X)$ X - stabilisers of $C$.

    Because there are $f$ bits flipped in $u$ in the first stage of putting $u$ in the code, there are bits flipped in at most $f$ blocks of $u$. This means that $u$ is non-uniform in at most $f$ blocks. Thus there are at least $N-f$ uniform blocks in $u$. On any uniform block $q$, we can deduce the value of $\hat{u}(q)$. Consider the set of such qubits in $C$ for which $u$ is uniform on the corresponding block. Call this set $S$, where we know that $|S| \geq N - f$. There are at most $\qx\left(N-|S|\right)$ X - stabilisers of $C$ supported anywhere outside of $S$, and therefore at least $N_X - \qx\left(N-|S|\right) \geq N_X - \qx f$ X - stabilisers of $C$ supported entirely within $S$. There are therefore at least

    \begin{equation}
        \frac{\rho_ZN_X}{N}d(\hat{u},\ker H_X) - \qx f \geq \left(\frac{\rho_ZN_X}{N} - \qx^2\alpha\right)d(\hat{u},\ker H_X)
    \end{equation}
     X - stabilisers of $C$ that are violated by $\hat{u}$ and supported entirely within $S$. For every such X - stabiliser of $C$, the corresponding encoded X - stabiliser in $\tilde{C}$ is violated by $u$. As such, we know that at least

    \begin{equation}
         \left(\frac{\rho_ZN_X}{N} - \qx^2\alpha\right)d(\hat{u},\ker H_X) \geq \left(\frac{\rho_ZN_X}{N} - \qx^2\alpha\right)\frac{d(u,\ker \tilde{H}_X)}{\qx(1+\alpha)}\label{copyingZVio2}
    \end{equation}
    X - stabilisers in $\tilde{C}$ are violated by $u$. We now choose $\alpha$ to make the right-hand sides of Equations \eqref{copyingZVio1} and \eqref{copyingZVio2} equal. This is $\alpha = \frac{\rho_ZN_X}{N\left(\qx^2+1\right)}$. From this, we have that at least

    \begin{equation}
        \frac{\rho_ZN_X}{\qx\rho_ZN_X + N\qx\left(\qx^2+1\right)}d(u,\ker \tilde{H}_X)
    \end{equation}
    stabilisers in $\tilde{C}$ are violated in both cases, from which we deduce the required soundness.
\end{proof}

\begin{lemma}
    Let parameters with and without a tilde be those of the code after and before copying respectively. Then

    \begin{equation}
        \tilde{\rho}_X \geq \qx\rho_X.
    \end{equation}
\end{lemma}

\begin{proof}
    We use the same language and notation as in the previous lemma of blocks of qubits etc. Consider an X - operator on the new set of qubits represented by some bit string $u \in \mathbb{F}_2^{\tilde{N}}$. Without loss of generality, $u$ is supported at most once on each block. The reason for this is that $\tilde{C}$ has stabilisers of the form $X_{(q,i)}X_{(q,i+1)}$. For some $q$, the bit strings corresponding to these stabilisers for this $q$ may be added together to give any bit string of even weight on this block. Adding such a bit string to $u$ within each block $q$ can reduce the support of $u$ within that block to at most one point.

    From $u$, we may define a corresponding bit string $\hat{u} \in \mathbb{F}_2^N$. If $u$ is supported in some place on a block $q$, then we set $\hat{u}(q) = 1$. Otherwise, $\hat{u}(q) = 0$. Then, $u$ is in the code $\tilde{C}$ if and only if the corresponding $\hat{u}$ is in the code $C$. From this, we may put $\hat{u}$ in the code $C$ by flipping as few bits in it as are necessary. This will be $d(\hat{u},\ker H_Z)$, where $\ker H_Z$ denotes the classical code defined by the Z - stabilisers of $C$. We may then put $u$ in the code $\tilde{C}$ following the way that $\hat{u}$ is put in the code $C$: if a bit $q$ in $\hat{u}$ is flipped from 1 to 0, then flip the bit that is set to 1 for $u$ in the block $q$ to 0. If a bit $q$ in $\hat{u}$ is flipped from 0 to 1, flip any bit in $u$ in the block $q$ to 1. Thus, $d(u,\ker \tilde{H}_Z) \leq d(\hat{u}, \ker H_Z)$.

    $\hat{u}$ violates at least $\frac{\rho_XN_Z}{N}d(\hat{u},\ker H_Z)$ Z - stabilisers in $C$. For every Z - stabiliser that $\hat{u}$ violates in $C$, $u$ violates one in $\tilde{C}$. Thus, $u$ violates at least $\frac{\rho_XN_Z}{N}d(u,\ker \tilde{H}_Z)$ Z - stabilisers in $\tilde{C}$, from which we deduce the required soundness using $\tilde{N} = \qx N$ and $\tilde{N}_Z = N_Z$.
\end{proof}
\subsection{Thickening and Choosing Heights}\label{thickeningAndChoosingHeightsSection}

Our objective now is to have an analysis of soundness under the full ``Thickening and Choosing Heights'' construction used to reduce \qz. We recall from Section \ref{thickenAndChooseHeightsPrelimSection} that the soundness of a quantum code under the thickening portion is already known. It is therefore left for us to analyse the soundness of the quantum code under the choice of heights. We recall that our original quantum code chain complex $C = \left(\mathbb{F}_2^{N_Z}\overset{H_Z^T}{\longrightarrow}\mathbb{F}_2^{N}\overset{H_X}{\longrightarrow}\mathbb{F}_2^{N_X}\right)$ becomes the chain complex $\mathcal{C}_2'\overset{\partial_2}{\longrightarrow}\mathcal{C}_1'\overset{\partial_1}{\longrightarrow}\mathcal{C}_0'$ during thickening, as shown in Figure \ref{thickeningAndChoosingHeightsCC}, where $H$ is the usual parity-check matrix associated with a length $l$ classical repetition code, as shown in Equation \eqref{repCodePCM}. We then, in the height-choice step, for each basis element $v \in \mathbb{F}_2^{N_Z}$, pick one $k \in \{1, ..., l\}$, and remove all basis elements $v \otimes w_{k'}$ (where $\left(w_i\right)_{i=1}^l$ forms the standard basis for $\mathbb{F}_2^l$) in $\mathbb{F}_2^{N_Z} \otimes \mathbb{F}_2^l$ for which $k' \neq k$, but keeping $v \otimes w_k$. We therefore keep all the qubits and X - stabilisers of the thickened code, but keep only a subset of the Z - stabilisers. Therefore, we have immediately that

\begin{fact}
    $\tilde{\rho}_Z = \rho'_Z$, where $\tilde{\rho}_Z$ refers to the soundness of the Z - operators of the code after thickening and choosing heights and $\rho'_Z$ refers to the same for the code after thickening but before choosing heights.
\end{fact}

\begin{figure}[]
\begin{center}
\begin{tikzpicture}

\node at (2.5,0) {$\mathbb{F}_2^{N_Z}\otimes \mathbb{F}_2^{l-1}$};
\node at (5,2) {$\mathbb{F}_2^{N}\otimes \mathbb{F}_2^{l-1}$};
\node at (8.5,2) {$\mathbb{F}_2^{N_X}\otimes \mathbb{F}_2^{l-1}$};
\node at (11,0) {$\mathbb{F}_2^{N_X}\otimes \mathbb{F}_2^l$};
\node at (4.95,-2) {$\mathbb{F}_2^{N_Z}\otimes \mathbb{F}_2^l$};
\node at (8.5,-2) {$\mathbb{F}_2^{N}\otimes \mathbb{F}_2^l$};
\node at (5,0) {$\oplus$};
\node at (8.5,0) {$\oplus$};

\draw [-stealth](2.8,0.3) -- (4.35,1.75);
\draw [-stealth](2.8,-0.3) -- (4.35,-1.75);
\draw [-stealth](5.9,2) -- (7.525,2);
\draw [-stealth](5.8,-2) -- (7.7,-2);
\draw [-stealth](9.35,1.75) -- (10.8,0.3);
\draw [-stealth](9.25,-1.85) -- (10.8,-0.3);
\draw [-stealth](5.65,1.7) -- (8,-1.7);

\node at (3,1.2) {$H_Z^T \otimes I$};
\node at (2.85,-1.2) {$I \otimes H^T$};
\node at (6.75,-1.725) {$H_Z^T \otimes I$};
\node at (6.75,2.25) {$H_X \otimes I$};
\node [rotate = -56] at (7,0.25) {$I \otimes H^T$};
\node at (10.78,1.2) {$I \otimes H^T$};
\node at (10.75,-1.2) {$H_X \otimes I$};

\node at (2.5,-4) {$\mathcal{C}'_3$};
\node at (5,-4) {$\mathcal{C}'_2$};
\node at (8.5,-4) {$\mathcal{C}'_1$};
\node at (11,-4) {$\mathcal{C}'_0$};

\draw [-stealth](3,-4) -- (4.5,-4);
\draw [-stealth](5.5,-4) -- (8,-4);
\draw [-stealth](9,-4) -- (10.5,-4);

\node at (3.75,-3.7) {$\partial_3$};
\node at (6.75,-3.7) {$\partial_2$};
\node at (9.75,-3.7) {$\partial_1$};

\end{tikzpicture}
\end{center}\caption{}\label{thickeningAndChoosingHeightsCC}
\end{figure}

We now turn to the soundness of the X - operators.

\begin{lemma}
Let parameters with and without a tilde refer to those of the code after and before ``thickening and choosing heights'' respectively. Let parameters with a prime refer to those of the code after thickening but before choosing heights. Then

\begin{equation}
    \tilde{\rho}_X \geq \frac{N_Z'}{\tilde{N}_Z}\frac{\rho_X'}{1+\wz\qz l}.
\end{equation}
and since $\tilde{N}_Z \leq N'_Z$, we have $\tilde{\rho}_X \geq \Omega\left(\frac{\rho'_X}{\wz\qz l}\right)$.
\end{lemma}
\begin{proof}
    Even though some Z - stabilisers are removed during the choosing of heights, it turns out that the thickened code, $C'$, and the code after thickening and choosing heights, $\tilde{C}$, have the same Z - stabiliser group. This fact was used by Hastings to determine the dimension and the distance of $\tilde{C}$, and we use it here to determine the soundness. The idea is as follows. Consider a stabiliser that is removed from the list of stabilisers of $C'$ when forming $\tilde{C}$, $v \otimes w_m$. We will show that this is still in the Z - stabiliser group of $\tilde{C}$, even if it is not now a ``stabiliser'' (i.e. not in the designated list of stabilisers). Indeed, consider an element $v \otimes p \in \mathcal{C}'_3 = \mathbb{F}_2^{N_Z}\otimes\mathbb{F}_2^{l-1}$, where $p$ is of the form $(0, 0, ..., 0, 0, 1, 1, ..., 1, 1, 0, 0, ..., 0, 0)$ and satisfies $H^T(p) = w_k + w_m$\footnote{For example, if $m > k$ then $p$ has 1's from the $k$'th to the $(m-1)$'th place.}. We then have

    \begin{equation}
        \partial_3(v \otimes p) = \left(v \otimes w_k + v \otimes w_m, H_Z^T(v) \otimes p\right).
    \end{equation}
    All of the basis elements in this linear combination represent Z - stabilisers in the thickened code. Note that there are two basis elements within $v \otimes w_k + v \otimes w_m$ but generally more than two in $H_Z^T(v) \otimes p$. Furthermore, the Z - stabiliser $v \otimes w_k$ is kept during the choosing of heights, as are all the Z - stabilisers within the linear combination of basis elements $H_Z^T(v) \otimes p$. This implies that the Z - stabiliser given by $v \otimes w_m$ is still in the stabiliser group after choosing heights because $\partial_2\partial_3(v \otimes p) = 0$ and so

    \begin{equation}
        \partial_2\left((v\otimes w_m,0)\right) = \partial_2\left(v\otimes w_k,H_Z^T(v) \otimes p\right).
    \end{equation}
    Thus, the Z - stabiliser corresponding to $v \otimes w_m$ is the same as the product of the Z - stabilisers corresponding to $v \otimes w_k$ and all those in the linear combination $H_Z^T(v) \otimes p$. Thus, $\ker H_Z'$ and $\ker \tilde{H}_Z$ are the same spaces, where the former refers to the classical code defined by the Z - stabilisers of the code after thickening (only), and the latter refers to the same thing after choosing heights.

    We may now consider an X - operator represented by some bit string $u \in \mathbb{F}_2^{\tilde{N}}$ i.e. some physical X - operator on the code after choosing heights. We may just as well consider this to be some physical X - operator on the code after thickening and indeed $d(u, \ker H_Z') = d(u, \ker \tilde{H}_Z)$ because, again, $\ker H_Z' = \ker \tilde{H}_Z$. In $C'$, $u$ will violate at least 

    \begin{equation}
        \frac{N_Z'}{N'}\rho_X'd(u, \ker H_Z')
    \end{equation}
    Z - stabilisers. The question now is how many of these stabilisers are removed when choosing heights. Consider a Z - stabiliser $v \otimes w_m$ that is violated in $C'$ but removed when forming $\tilde{C}$. At least one of the stabilisers in the linear combination of basis elements $\left(v\otimes w_k,H_Z^T(v) \otimes p\right)$ must be violated, because the Z - stabilizer corresponding to $v \otimes w_m$ is  the product of the Z - stabilizers corresponding to $\left(v\otimes w_k,H_Z^T(v) \otimes p\right)$, and the X - operator $u$ cannot anti-commute with a product of Z - stabilisers without anti-commuting with at least one of them. This observation will allow us to translate a number of violated stabilisers in $C'$ to a number of violated stabilisers in $\tilde{C}$, but there is one complication. Given two stabilisers violated in $C'$, $v_1 \otimes w_{m_1}$ and $v_2 \otimes w_{m_2}$, but removed when forming $\tilde{C}$, these two stabilisers may have some overlap in their linear combinations $\left(v_1\otimes w_{k_1},H_Z^T\left(v_1\right) \otimes p_1\right)$ and $\left(v_2\otimes w_{k_2},H_Z^T\left(v_2\right) \otimes p_2\right)$, and so we cannot necessarily conclude that for each stabiliser violated in $C'$ but removed when forming $\tilde{C}$, there is at least one violated in $\tilde{C}$. We must adjust this argument as follows.

    We first ask: when can there be overlap between the linear combinations $\left(v_1\otimes w_{k_1},H_Z^T\left(v_1\right) \otimes p_1\right)$ and $\left(v_2\otimes w_{k_2},H_Z^T\left(v_2\right) \otimes p_2\right)$? If $v_1 = v_2$, then $v_1\otimes w_{k_1} = v_2\otimes w_{k_2}$, and there is certainly overlap. Thus assume that $v_1 \neq v_2$. In this case, in order for there to be overlap between the two linear combinations, there must be some overlap between $H_Z^T(v_1)$ and $H_Z^T(v_2)$. Given a $v_1 \in \mathbb{F}_2^{N_Z}$, there may be at most $\wz\qz$ $v_2's$ such that $H_Z^T(v_1)$ and $H_Z^T(v_2)$ have some overlap.

    Thus, given a stabiliser $v_1 \otimes w_{m_1}$ that is violated in $C'$ but removed when forming $\tilde{C}$, we may conclude that there is some Z - stabiliser violated in $\tilde{C}$. Then, however, we \textit{exclude} all of the stabilisers $v_2 \otimes w_{m_2}$ that are violated in $C'$ but removed when forming $\tilde{C}$ that have some overlap in their linear combination with $v_1 \otimes w_{k_1}$. This means excluding at most $\wz\qz l$ stabilisers. After this, we may look at the non-excluded stabilisers of the form $v_3 \otimes w_{m_3}$ that are violated in $C'$ but removed when forming $\tilde{C}$. Picking one of these, its violation must imply the violation of some Z - stabiliser in $\tilde{C}$. We then exclude all of the stabilisers that have some overlap in their linear combination with $v_1 \otimes w_{m_1}$ or $v_3 \otimes w_{m_3}$, and continue in this way. Ultimately, we may conclude that having a certain number of stabilisers violated in $C'$ that are removed when forming $\tilde{C}$ implies the violation of at least a $\frac{1}{\wz\qz l}$ fraction of this many stabilisers in $\tilde{C}$.

    To conclude the proof, let us say that the X - operator $u$ violates $V_1 + V_2$ Z - stabilisers in $C'$, where $V_1$ is the number that are removed when forming $\tilde{C}$ and $V_2$ is the number that are kept. If $V_1 \leq \wz\qz l V_2$, then the number of Z - stabilisers violated in $\tilde{C}$ is at least 

    \begin{equation}
        V_2 \geq \frac{V_1 + V_2}{1+\wz\qz l} \geq \frac{N_Z'}{N'}\rho_X'\frac{d(u,\ker H_Z')}{1+\wz\qz l}
    \end{equation}
    and we get the required soundness in this case. If instead $V_1 \geq \wz\qz l V_2$, then the number of Z - stabilisers violated in $\tilde{C}$ is at least

    \begin{equation}
        \frac{V_1}{\wz\qz l} \geq \frac{V_1 + V_2}{1+\wz\qz l} \geq \frac{N_Z'}{N'}\rho_X'\frac{d(u,\ker H_Z')}{1+\wz\qz l}
    \end{equation}
    again giving the required soundness.
\end{proof}

\subsection{Coning}\label{coningSection}

We now provide a full analysis of soundness of a quantum code under the coning construction, which is the construction used to reduce \wz while preserving the constancy of all other parameters. This construction was presented in Section \ref{coningPrelimSection}. We will do this incrementally, first determining a lower bound on the soundness of the cone code in terms of that of the code pre-coning, and then determining a lower bound on the soundness of the reduced cone code. We recall that the reduced cone code is formed from the cone code first by (dually) thickening and choosing heights to reduce \qx, and then by cellulating the added 2-discs. We recall that we have an analysis of the soundness of a quantum code under the process of thickening and choosing heights from Lemma \ref{thickeningSoundnessLemma} and from Section \ref{thickeningAndChoosingHeightsSection}. It will therefore suffice to analyse the soundness of the cone code, and then analyse the soundness under the cellulation step. However, we mention that we will end up taking a larger value of $l$ for the thickening and choosing heights here than the one used by Hastings to reduce \qx to a constant, $l = \Theta\left(\log\left(\wz\right)^{2+2\epsilon}\wz^\epsilon\right)$, as in Lemma \ref{coningNonSoundnessLemma}. This will be for the benefit of our proof of soundness under the cellulation step, as we will go on to explain.

\subsubsection{Soundness of the Cone Code}

\begin{lemma}
Let parameters with and without a prime refer to those of the cone code and the code pre-coning respectively. Then

\begin{equation}
    \rho'_Z \geq \frac{N'}{N_X'}\frac{\rho_Z}{\wz\qx\wx\rho_Z + \frac{N}{N_X} + \wz\qx\wx\frac{N}{N_X}}.
\end{equation}
\end{lemma}

\begin{proof}
For ease of reference, the structure of the chain complex of the cone code is shown in Figure \ref{coneCodeChainComplex}.

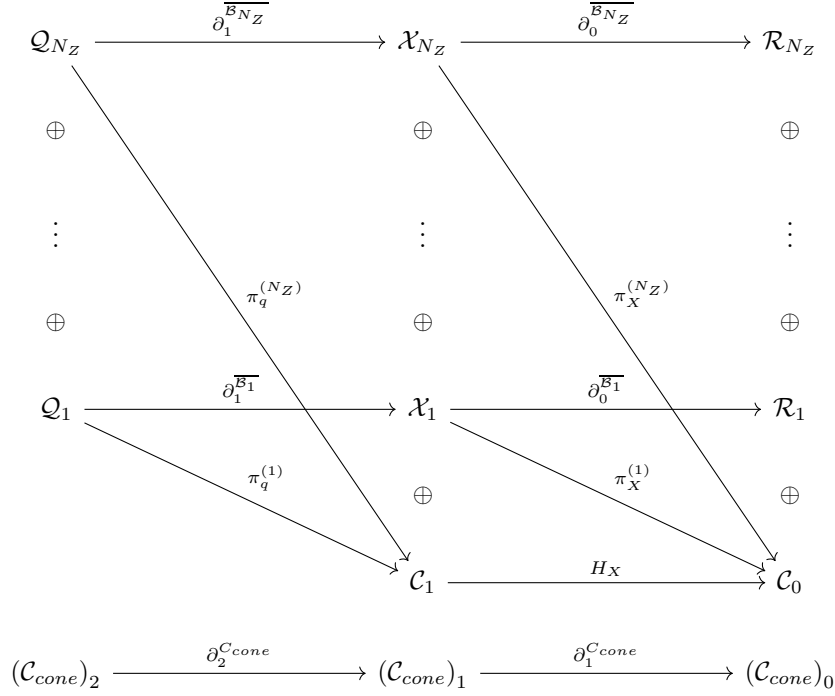
\begin{figure}[h]
\begin{center}
\begin{tikzcd}
\mathcal{Q}_{N_Z} \arrow[rrrr, "\partial^{\overline{\mathcal{B}_{N_Z}}}_1"] \arrow[rrrrdddddd, "\pi_q^{(N_Z)}"] &  &  &  & \mathcal{X}_{N_Z} \arrow[rrrr, "\partial^{\overline{\mathcal{B}_{N_Z}}}_0"] \arrow[rrrrdddddd, "\pi_X^{(N_Z)}"] &  &  &  & \mathcal{R}_{N_Z}                 \\
\oplus                                                                                                          &  &  &  & \oplus                                                                                                          &  &  &  & \oplus                            \\
\vdots                                                                                                          &  &  &  & \vdots                                                                                                          &  &  &  & \vdots                            \\
\oplus                                                                                                          &  &  &  & \oplus                                                                                                          &  &  &  & \oplus                            \\
\mathcal{Q}_1 \arrow[rrrr, "\partial^{\overline{\mathcal{B}_{1}}}_1"] \arrow[rrrrdd, "\pi_q^{(1)}"]             &  &  &  & \mathcal{X}_1 \arrow[rrrr, "\partial^{\overline{\mathcal{B}_{1}}}_0"] \arrow[rrrrdd, "\pi_X^{(1)}"]             &  &  &  & \mathcal{R}_1                     \\
                                                                                                                &  &  &  & \oplus                                                                                                          &  &  &  & \oplus                            \\
                                                                                                                &  &  &  & \mathcal{C}_1 \arrow[rrrr, "H_X"]                                                                               &  &  &  & \mathcal{C}_0                     \\
\left(\mathcal{C}_{cone}\right)_2 \arrow[rrrr, "\partial^{C_{cone}}_2"]                                         &  &  &  & \left(\mathcal{C}_{cone}\right)_1 \arrow[rrrr, "\partial^{C_{cone}}_1"]                                         &  &  &  & \left(\mathcal{C}_{cone}\right)_0
\end{tikzcd}
\end{center}\caption{The chain complex of the cone code. The spaces of the cone code, $\left(\mathcal{C}_{cone}\right)_i$, can be read off by ``direct summing vertically''. The boundary maps can also be determined by reading off the maps vertically and summing images in the same space. $\partial_1^{\overline{\mathcal{B}_i}}$ and $\partial_0^{\overline{\mathcal{B}_i}}$ are the boundary maps of the chain complex \BBARI, while $H_X$ is a boundary map coming from the chain complex of the pre-coning code. $\pi_q^{(i)}$ and $\pi_X^{(i)}$ are the obvious projection maps.}\label{coneCodeChainComplex}
\end{figure}
Let us consider some physical Z - operator in the cone code, represented by some bit string $u \in \left(\mathcal{C}_{cone}\right)_1$. Hereafter, we will refer to the cone code as $C'$, rather than $C_{cone}$, for brevity. Let us write $u = (u_1, u_2)$ for $u_1 \in \mathcal{C}_1$ and $u_2 = \bigoplus_{i=1}^{N_Z}u_2^{(i)} \in \bigoplus_{i=1}^{N_Z}\XXI$, where $u_2^{(i)} \in \XXI$. For each $\BBARI$, consider $d(u_2^{(i)}, \ker \partial_0^{\BBARI})$. $\exists$ $v_2^{(i)} \in \ker \partial_0^{\BBARI}$ such that $\left|u_2^{(i)} + v_2^{(i)}\right| = d(u_2^{(i)}, \ker \partial_0^{\BBARI})$. \BBARI has trivial zero-th homology, and therefore $v_2^{(i)} \in \im\partial_1^{\BBARI}$ i.e. $\exists$ $w_2^{(i)} \in \QQI$ such that $\partial_1^{\BBARI}w_2^{(i)} = v_2^{(i)}$. Adding the Z - stabiliser corresponding to $\bigoplus_{i=1}^{N_Z} w_2^{(i)}$ then shows us that without loss of generality, we may say that our Z - operator $u$ satisfies

\begin{equation}
    \left|u_2^{(i)}\right| = d(u_2^{(i)}, \ker \partial_0^{\BBARI})
\end{equation}
and, in particular, $u_2^{(i)} \in \ker \partial_0^{\BBARI} \iff u_2^{(i)} = 0$.

We may then put $u$ in the cone code as follows. Suppose that there are exactly $f$ bits in $u_2$ set to 1. We flip all of these bits to 0. Then, we flip the smallest number of bits in $u_1$ required to put it in $\ker H_X$, which will be $d(u_1, \ker H_X)$ bits. The result will be in the cone code. Therefore, we have

\begin{equation}
    d(u, \ker H'_X) \leq d(u_1, \ker H_X) + f
\end{equation}
where $H'_X$ is the parity-check matrix associated with the X - stabilisers of the cone code, i.e. $H'_X = \partial_1^{C_{cone}}$.

Consider first the case that $f \leq \alpha d(u_1,\ker H_X)$ for some positive number $\alpha$ to be determined. Then

\begin{equation}
    d(u, \ker H'_X) \leq (1 + \alpha)d(u_1, \ker H_X).
\end{equation}
$u_1$ violates at least $\rho_Z\frac{N_X}{N}d(u_1, \ker H_X)$ X - stabilisers in the pre-coning code, $C$. Consider the $f$ bits that get flipped from 1 to 0 in $u_2$ when we put $u$ in the code. Each of these bits is in some set $X_i$ (the set of basis elements of \XXI) and is labelled by some tuple $(T,j)$ for T some X - stabiliser of $C$ and j some qubit pair. Consider the following subset, S, of the X - stabilisers of $C$: an X - stabiliser, $T$, is in S if none of the $f$ bits in $u_2$ contain $T$ in their tuple. We have $|S| \geq N_X - f$ and then the number of X - stabilisers of $C$ that are violated by $u_1$ and that are in S is at least

\begin{align}
    \rho_Z\frac{N_X}{N}d(u_1, \ker H_X) - f &\geq \left(\rho_Z\frac{N_X}{N} - \alpha\right)d(u_1, \ker H_X)\\
    &\geq \frac{1}{1+\alpha}\left(\rho_Z\frac{N_X}{N} - \alpha\right)d(u, \ker H_X')\label{coningSoundnessVio1}.
\end{align}
At least this many X - stabilisers are violated by $u$ in the cone code, because an X - stabiliser in $C$ violated by $u_1$ which is also in $S$ must be violated by $u$ (as a basis element of $\mathcal{C}_0$ in the cone code).

Consider now the case where $f \geq \alpha d(u_1, \ker H_X)$. Then

\begin{equation}
    d(u, \ker H_X') \leq f\left(\frac{1+\alpha}{\alpha}\right).
\end{equation}
There are at most $\wz\qx\wx$ elements in $X_i$. Therefore, there are at least $\frac{f}{\wz\qx\wx}$ of the \XXI on which $u_2$ is supported. If $u_2$ is supported in some \XXI then there must be at least one violated stabiliser in \RRI, again because $u_2^{(i)} = 0 \iff u_2^{(i)} \in \ker \partial_0^{\BBARI}$. There are therefore at least

\begin{equation}
    \frac{f}{\wz\qx\wx} \geq \frac{\alpha}{\alpha+1}\frac{1}{\wz\qx\wx}d(u, \ker H_X')\label{coningSoundnessVio2}
\end{equation}
X - stabilisers violated by $u$ in $C'$. We let $\alpha$ be such that the right-hand sides of Equation \eqref{coningSoundnessVio1} and \eqref{coningSoundnessVio2} are equal, which turns out to be

\begin{equation}
    \alpha = \frac{\wz\qx\wx}{1+\wz\qx\wx}\frac{\rho_ZN_X}{N}.
\end{equation}
With this choice, there are at least

\begin{equation}
\frac{\rho_Z}{\wz\qx\wx\rho_Z + \frac{N}{N_X}+\wz\qx\wx\frac{N}{N_X}}d(u, \ker H_X')
\end{equation}
X - stabilisers violated by $u$, from which we deduce the required soundness.
\end{proof}

\begin{lemma}
    Let parameters with and without a prime refer to those of the cone code and the code pre-coning respectively. Then, assuming that the code pre-coning is reasonable,

    \begin{equation}
        \rho_X' \geq \frac{N'}{N_Z'}\frac{N_Z}{N}\frac{\rho_X}{1+\frac{N_Z}{N}\rho_X\wz\qx\wx + \qz\wz\qx\wx}.
    \end{equation}
\end{lemma}

\begin{proof}
    Consider some X - operator in the cone code represented by some bit string $u \in \left(\mathcal{C}_{cone}\right)_1$ where, again, we will refer to the cone code as $C'$ rather than $C_{cone}$ from here on. We again write $u = (u_1, u_2)$ where $u_1 \in \mathcal{C}_1$ and $u_2 = \bigoplus_{i=1}^{N_Z}u_2^{(i)} \in \bigoplus_{i=1}^{N_Z}\XXI$, for $u_2^{(i)} \in \XXI$. Given $u_1$, consider the projection of $u_1$ into $\mathcal{Q}_i$. Call it $q_i$. Then $u$ is in the cone code if and only if

    \begin{equation}
        q_i = \left(\partial_1^{\BBARI}\right)^T\left(u_2^{(i)}\right).
    \end{equation}
    We note that this is only possible if $q_i$ has even support for every $i$, because the coboundary of every basis element in \XXI has size 2 and $q_i$ is expressed as a linear combination of such coboundaries. Note further that $q_i$ has even support if and only if $u_1$ satisfies the stabiliser labelled by $i \in \{1, ..., N_Z\}$. With $H_Z$ the parity-check matrix associated with the Z - stabilisers of the pre-coning code, $u_1$ violates at least $\frac{N_Z}{N}\rho_Xd(u_1,\ker H_Z)$ Z - stabilisers in the pre-coning code, $C$.

    We may now put $u_1$ in the code $C'$ as follows. First flip $d(u_1,\ker H_Z)$ bits in $u_1$ to put $u_1$ in the pre-coning code. Let the codeword of $\ker H_Z$ that it becomes be denoted $u_1'$. Let the projection of $u_1'$ into $\mathcal{Q}_i$ be $q_i'$. We want to argue that $\exists$ $u_2'^{(i)} \in \XXI$ such that $q_i' = \left(\partial_1^{\BBARI}\right)^T\left(u_2'^{(i)}\right)$. This is true, as follows.
    
    First, $q_i'$ must have even support because $u_1'$ is a codeword of $\ker H_Z$ and so commutes with all stabilisers of $C$. Consider the graph $G_i$ as defined in Section \ref{coningPrelimSection}. $q_i'$ represents a set of vertices in that graph of even size. Therefore, if $G_i$ is connected, we can always find a set of edges in $G_i$ (i.e. an element of \XXI), $u_2'^{(i)}$, such that $q_i' = \left(\partial_1^{\BBARI}\right)^T\left(u_2'^{(i)}\right)$. What if $G_i$ is disconnected? We need to have it such that $q_i'$ has even support on every connected component of $G_i$, or we will not be able to find any valid $u_2'^{(i)}$. However, $q_i'$ must have even support on every connected component of $G_i$ because of the following claim.

    \begin{claim}
        The product of Z's over all the vertices (qubits) in a connected component of $G_i$ is a member of the stabiliser group of $C$.
    \end{claim}

    \begin{proof}
        The Z - operator in question has an even overlap with every X - stabiliser of $C$ by construction of $G_i$. By Claim \ref{reasonableClaim}, this Z - operator is a member of the stabiliser group of $C$ because it is supported entirely on a subset of some Z - stabiliser, and we assume the code to be reasonable.
    \end{proof}

    The X - operator represented by the bit string $u_1'$ must therefore commute with the product of Z - operators over the vertices of any connected component of $G_i$, because $u_1'$ is a codeword of $\ker H_Z$, and as a result $q_i'$ must have even support on every connected component of $G_i$. Thus, indeed $\exists$ $u_2'^{(i)} \in \XXI$ such that $\left(\partial_1^{\BBARI}\right)^T\left(u_2'^{(i)}\right) = q_i'$.

    Having flipped $d(u_1, \ker H_Z)$ bits in $\mathcal{C}_1$ to turn $u_1$ into $u_1'$, we finish putting $u$ in the cone code by flipping the minimum number of bits in $u_2$ required to put $(u_1', u_2)$ in the cone code, which we have demonstrated is indeed possible by turning each $u_2^{(i)}$ into $u_2'^{(i)}$. Suppose that in doing this, we flip $f$ bits in $u_2$. Then

    \begin{equation}
        d(u, \ker H_Z') \leq d(u_1, \ker H_Z) + f
    \end{equation}
    where $H_Z'$ is the parity-check matrix associated with the Z - stabilisers of the cone code. Let us suppose first that $f \leq \alpha d(u_1, \ker H_Z)$ for some positive number $\alpha$ to be determined. Then $d(u, \ker H_Z') \leq (1+\alpha)d(u_1, \ker H_Z)$.

    As mentioned, $u_1$ violates at least $\frac{N_Z}{N}\rho_Xd(u_1, \ker H_Z)$ Z - stabilisers in $C$. Then, for each of these stabilisers labelled by some $i \in \{1, ..., N_Z\}$, $q_i$ has odd support so $\nexists$ $u_2^{(i)}$ such that $q_i = \left(\partial_1^{\BBARI}\right)^T\left(u_2^{(i)}\right)$ and therefore for each such $i$, at least one Z - stabiliser of $C'$ in \QQI is violated. Thus, at least

    \begin{equation}
        \frac{N_Z}{N}\rho_Xd(u_1, \ker H_Z) \geq \frac{1}{1+\alpha}\frac{N_Z}{N}\rho_Xd(u, \ker H_Z')\label{coningXVio1}
    \end{equation}
    Z - stabilisers are violated in $C'$.

    Suppose now that $f \geq \alpha d(u_1, \ker H_Z)$. Then $d(u, \ker H_Z') \leq f\left(\frac{\alpha+1}{\alpha}\right)$. Consider the set $S$ of Z - stabilisers of $C$ which have no support on any of the $d(u_1, \ker H_Z)$ bits that get flipped to turn $u_1$ into $u_1'$. Since at most $\qz d(u_1, \ker H_Z)$ Z - stabilisers may have support on those bits, $|S| \geq N_Z - \qz d(u_1, \ker H_Z)$. Furthermore, since $|X_i| \leq \wz\qx\wx$, there must be at least $\frac{f}{\wz\qx\wx}$ sets \XXI in which bits get flipped. As such, there are at least

    \begin{equation}
        \frac{f}{\wz\qx\wx}-\qz d(u_1, \ker H_Z) \geq \left(\frac{1}{\wz\qx\wx}-\frac{\qz}{\alpha}\right)f
    \end{equation}
    Z - stabilisers for which there are bits flipped in the corresponding \XXI and the bits on which they are supported are not flipped in $C_1$ to turn $u_1$ into $u_1'$. For any such Z - stabiliser labelled by $i \in \{1, ..., N_Z\}$, there must be at least one Z - stabiliser of the cone code in \QQI violated, because it means that while $q_i$ did have even support, $u_2^{(i)}$ did not satisfy $q_i = \left(\partial_1^{\BBARI}\right)^T\left(u_2^{(i)}\right)$. This is true because $f$ is the minimum number of bits we flip in $u_2$ to turn $(u_1',u_2)$ into a codeword of $\ker H_Z'$, and therefore bits would not be flipped in some \XXI for which $u_2^{(i)}$ did already satisfy $q_i = \left(\partial_1^{\BBARI}\right)^T\left(u_2^{(i)}\right)$. There are therefore at least

    \begin{equation}
        \left(\frac{1}{\wz\qx\wx}-\frac{\qz}{\alpha}\right)f \geq \left(\frac{1}{\wz\qx\wx}-\frac{\qz}{\alpha}\right)\frac{\alpha}{\alpha+1}d(u, \ker H_Z').\label{coningXVio2}
    \end{equation}
    Z - stabilisers violated by $u$ in $C'$. We then let $\alpha$ be such that the right-hand sides of Equations \eqref{coningXVio1} and \eqref{coningXVio2} are equal. This is $\alpha = \left(\frac{N_Z}{N}\rho_X + \qz\right)\wz\qx\wx$ and then the number of Z - stabilisers in $C'$ violated by $u$ is at least

    \begin{equation}
        \frac{1}{1+\frac{N_Z}{N}\wz\qx\wx + \qz\wz\qx\wx}\frac{N_Z}{N}\rho_Xd(u, \ker H_Z')
    \end{equation}
    from which we deduce the required soundness.
\end{proof}

\subsubsection{Soundness of the Reduced Cone Code}\label{reducedConeCodeSoundnessSection}

We will now establish a lower bound on the soundness of the reduced cone code in terms of that of the cone code. As discussed, forming the reduced cone code comes with steps of ``thickening and choosing heights'' and ``cellulating'', but we already understand the change in soundness under the ``thickening and choosing heights'' step. There is a mild complication in the fact that we only ``choose heights'' for the X - stabilisers in $\RRI \otimes \mathcal{E}_0$, not those in $\mathcal{C}_0\otimes \mathcal{E}_0$, but we can easily take account of this by changing the value of $N_X$ used. The chain complex of the cone code after thickening and choosing heights is depicted in Figure \ref{thickendAndHeightChosenConeCodeChainComplex}.
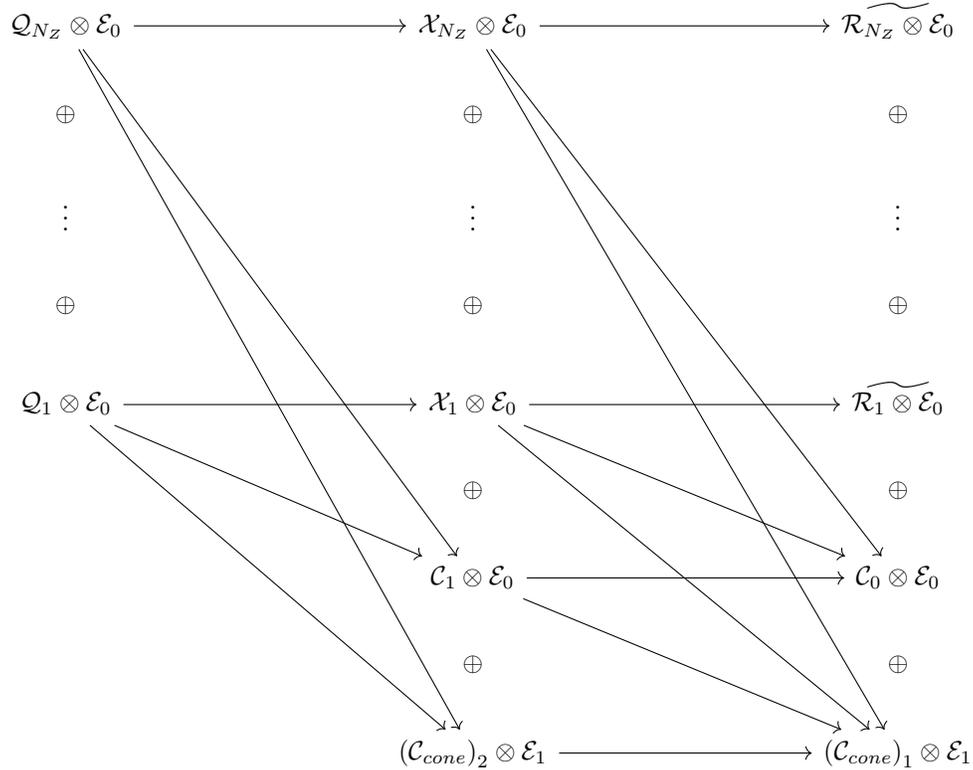
\begin{figure}[h]
\begin{center}
    
        \begin{tikzcd}
\mathcal{Q}_{N_Z}\otimes\mathcal{E}_0 \arrow[rrrr] \arrow[rrrrdddddd] \arrow[rrrrdddddddd] &  &  &  & \mathcal{X}_{N_Z}\otimes\mathcal{E}_0 \arrow[rrrr] \arrow[rrrrdddddd] \arrow[rrrrdddddddd] &  &  &  & \widetilde{\mathcal{R}_{N_Z}\otimes\mathcal{E}_0}     \\
\oplus                                                                                     &  &  &  & \oplus                                                                                     &  &  &  & \oplus                                                \\
\vdots                                                                                     &  &  &  & \vdots                                                                                     &  &  &  & \vdots                                                \\
\oplus                                                                                     &  &  &  & \oplus                                                                                     &  &  &  & \oplus                                                \\
\mathcal{Q}_1\otimes\mathcal{E}_0 \arrow[rrrrdddd] \arrow[rrrrdd] \arrow[rrrr]             &  &  &  & \mathcal{X}_1\otimes\mathcal{E}_0 \arrow[rrrr] \arrow[rrrrdddd] \arrow[rrrrdd]             &  &  &  & \widetilde{\mathcal{R}_{1}\otimes\mathcal{E}_0}       \\
                                                                                           &  &  &  & \oplus                                                                                     &  &  &  & \oplus                                                \\
                                                                                           &  &  &  & \mathcal{C}_1\otimes\mathcal{E}_0 \arrow[rrrrdd] \arrow[rrrr]                              &  &  &  & \mathcal{C}_0\otimes\mathcal{E}_0                     \\
                                                                                           &  &  &  & \oplus                                                                                     &  &  &  & \oplus                                                \\
                                                                                           &  &  &  & \left(\mathcal{C}_{cone}\right)_2\otimes\mathcal{E}_1 \arrow[rrrr]                         &  &  &  & \left(\mathcal{C}_{cone}\right)_1\otimes\mathcal{E}_1
\end{tikzcd}
\end{center}\caption{The chain complex of the cone code after thickening and choosing heights. The tildes on the $\RRI\otimes\mathcal{E}_0$ denote the fact that heights have been chosen in these spaces, i.e. that for every basis element $c \in \RRI$, we keep one and only one basis element of the form $c \otimes w_k$, where $\left(w_i\right)_{i=1}^l$ forms the standard basis for $\mathcal{E}_0$. To make the diagram as simple as is possible, we have abbreviated various sets of spaces into $\left(\mathcal{C}_{cone}\right)_i$.}\label{thickendAndHeightChosenConeCodeChainComplex}
\end{figure}

There will, however, be one slightly greater complication for us, and that is that we will thicken the cone code more than Hastings does. For the benefit of our later soundness proof, we will want to take a large enough $l$ and choose heights such that not only are qubits in $\XXI\otimes\mathcal{E}_0$ acted on by at most one X - stabiliser in $\widetilde{\RRI\otimes\mathcal{E}_0}$, which is what \cite{hastings2021quantum} does, but also each Z - stabiliser in $\QQI \otimes \mathcal{E}_0$ has overlap with at most one X - stabiliser in $\widetilde{\RRI\otimes\mathcal{E}_0}$. Let us now demonstrate that this is possible, in the style of an argument in \cite{hastings2021quantum}. This argument will go via ``strong hypergraph colouring''. The idea is as follows. Define a hypergraph\footnote{A hypergraph is a pair $(V,E)$ of a set of vertices, $V$, and a set of ``hyperedges'', $E$, which is a set of subsets of $V$.}, where vertices correspond to X - stabilisers of the cone code in \RRI, and hyperedges correspond to Z - stabilisers of the cone code in \QQI. A hyperedge corresponding to a Z - stabiliser will contain all the vertices corresponding to the X - stabilisers with which it has some overlap. Colouring the vertices of the hypergraph with $l$ colours, the objective behind strong hypergraph colouring is to make it such that no hyperedge contains two vertices of the same colour. Then, colouring the vertices corresponds to choosing heights for the stabilisers in $\widetilde{\RRI\otimes\mathcal{E}_0}$, and achieving a strong hypergraph colouring will make it such that no two X - stabilisers in $\widetilde{\RRI\otimes\mathcal{E}_0}$ have overlap with the same Z - stabiliser in $\QQI\otimes\mathcal{E}_0$.

We then use the trivial upper bound \cite{38853} that a hypergraph can be strongly coloured by a number of colours scaling with its maximum edge degree multiplied by its maximum vertex degree, which for us is $l = \Theta\left(\wz\qx\log\wz\right)$ colours, where these un-decorated parameters refer to those of the code pre-coning. Note that choosing heights such that no Z - stabiliser in $\QQI \otimes \mathcal{E}_0$ has overlap with more than one X - stabiliser in $\widetilde{\RRI\otimes\mathcal{E}_0}$ will automatically make it such that every qubit in $\XXI\otimes\mathcal{E}_0$ will be acted on by at most one X - stabiliser in $\widetilde{\RRI\otimes\mathcal{E}_0}$, assuming every qubit in \XXI is acted on by some Z - stabiliser in \QQI, which we see is true by the construction of \BBARI.

With this done, we may now cellulate the added 2 - discs, which we recall is necessary to turn the added X - stabilisers into low weight stabilisers. We do this according to the simple cellulation shown in Figure \ref{cellulationFigure} for every disc, where note in particular that no vertices are added, and therefore no Z - stabilisers are added in the cellulation step. The resulting chain complex - that of the reduced cone code - is shown in Figure \ref{reducedConeCodeChainComplex}.

\begin{figure}[h]
\begin{center}
\begin{tikzpicture}

  \draw[color = black](0,0) circle[radius = 2.5];

  \draw[fill=red] (2.500,0.000) circle (2pt);
  \draw[fill=red] (2.2839,1.0168) circle (2pt);
  \draw[fill=red] (1.6728,1.8579) circle (2pt);
  \draw[fill=red] (0.7725,2.3776) circle (2pt);
  \draw[fill=red] (-0.2613,2.4863) circle (2pt);
  \draw[fill=red] (-1.2500,2.1651) circle (2pt);
  \draw[fill=red] (-2.0225,1.4695) circle (2pt);
  \draw[fill=red] (-2.4454,0.5198) circle (2pt);
  \draw[fill=red] (-2.4454,-0.5198) circle (2pt);
  \draw[fill=red] (-2.0225,-1.4695) circle (2pt);
  \draw[fill=red] (-1.2500,-2.1651) circle (2pt);
  \draw[fill=red] (-0.2613,-2.4863) circle (2pt);
  \draw[fill=red] (0.7725,-2.3776) circle (2pt);
  \draw[fill=red] (1.6728,-1.8579) circle (2pt);
  \draw[fill=red] (2.2839,-1.0168) circle (2pt);

  \node at (3.5,0) {$\longmapsto$};

  \draw[color = black](7,0) circle[radius = 2.5];

  \draw[line width = 0.15mm] (7+2.2839,1.0168) -- (7+2.2839,-1.0168);
  \draw[line width = 0.15mm] (7+1.6728,1.8579) -- (7+1.6728,-1.8579);
  \draw[line width = 0.15mm] (7+0.7725,2.3776) -- (7+0.7725,-2.3776);
  \draw[line width = 0.15mm] (7+-0.2613,2.4863) -- (7+-0.2613,-2.4863);
  \draw[line width = 0.15mm] (7+-1.2500,2.1651) -- (7+-1.2500,-2.1651);
  \draw[line width = 0.15mm] (7+-2.0225,1.4695) -- (7+-2.0225,-1.4695);
  
  \draw[fill=red] (7+2.500,0.000) circle (2pt);
  \draw[fill=red] (7+2.2839,1.0168) circle (2pt);
  \draw[fill=red] (7+1.6728,1.8579) circle (2pt);
  \draw[fill=red] (7+0.7725,2.3776) circle (2pt);
  \draw[fill=red] (7+-0.2613,2.4863) circle (2pt);
  \draw[fill=red] (7+-1.2500,2.1651) circle (2pt);
  \draw[fill=red] (7+-2.0225,1.4695) circle (2pt);
  \draw[fill=red] (7+-2.4454,0.5198) circle (2pt);
  \draw[fill=red] (7+-2.4454,-0.5198) circle (2pt);
  \draw[fill=red] (7+-2.0225,-1.4695) circle (2pt);
  \draw[fill=red] (7+-1.2500,-2.1651) circle (2pt);
  \draw[fill=red] (7+-0.2613,-2.4863) circle (2pt);
  \draw[fill=red] (7+0.7725,-2.3776) circle (2pt);
  \draw[fill=red] (7+1.6728,-1.8579) circle (2pt);
  \draw[fill=red] (7+2.2839,-1.0168) circle (2pt);

\end{tikzpicture}
\end{center}\caption{The cellulation that we use for all discs. The high weight X - stabiliser represented by the left-hand disc is replaced by many low weight X - stabilisers, represented by the new faces on the right-hand side. Note that while new qubits have been added in the form of new edges, no new Z - stabiliser (vertices) have been added.}\label{cellulationFigure}
\end{figure}
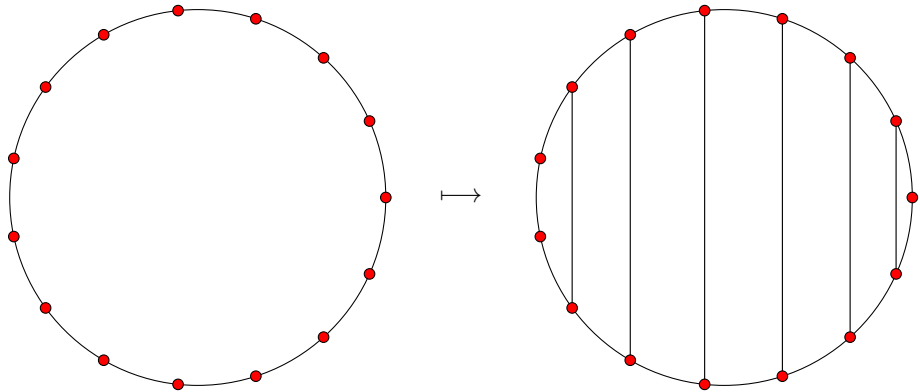

\begin{figure}[]
\begin{center}
    \begin{tikzcd}
                                                                                                              &  &  &  & \mathcal{X}_{N_Z}\otimes\mathcal{E}_0 \arrow[rrrrdddd] \arrow[rrrrdddddddddd] \arrow[rrrrdddddddddddd] &  &  &  &                                                         \\
                                                                                                              &  &  &  & \oplus                                                                                                 &  &  &  &                                                         \\
\mathcal{Q}_{N_Z}\otimes\mathcal{E}_0 \arrow[rrrruu] \arrow[rrrr] \arrow[rrrrdddddddd] \arrow[rrrrdddddddddd] &  &  &  & \widetilde{\mathcal{X}^\perp_{N_Z}\otimes\mathcal{E}_0} \arrow[rrrrdd]                                 &  &  &  &                                                         \\
\oplus                                                                                                        &  &  &  & \oplus                                                                                                 &  &  &  &                                                         \\
\vdots                                                                                                        &  &  &  & \vdots                                                                                                 &  &  &  & \widetilde{\hat{\mathcal{R}}_{N_Z}\otimes\mathcal{E}_0} \\
\oplus                                                                                                        &  &  &  & \oplus                                                                                                 &  &  &  & \oplus                                                  \\
\mathcal{Q}_1\otimes\mathcal{E}_0 \arrow[rrrr] \arrow[rrrrdd] \arrow[rrrrdddd] \arrow[rrrrdddddd]             &  &  &  & \mathcal{X}_1\otimes\mathcal{E}_0 \arrow[rrrrdd] \arrow[rrrrdddddd] \arrow[rrrrdddd]                   &  &  &  & \vdots                                                  \\
                                                                                                              &  &  &  & \oplus                                                                                                 &  &  &  & \oplus                                                  \\
                                                                                                              &  &  &  & \widetilde{\mathcal{X}^\perp_{1}\otimes\mathcal{E}_0} \arrow[rrrr]                                     &  &  &  & \widetilde{\hat{\mathcal{R}}_{1}\otimes\mathcal{E}_0}   \\
                                                                                                              &  &  &  & \oplus                                                                                                 &  &  &  & \oplus                                                  \\
                                                                                                              &  &  &  & \mathcal{C}_1\otimes\mathcal{E}_0 \arrow[rrrrdd] \arrow[rrrr]                                          &  &  &  & \mathcal{C}_0\otimes\mathcal{E}_0                       \\
                                                                                                              &  &  &  & \oplus                                                                                                 &  &  &  & \oplus                                                  \\
                                                                                                              &  &  &  & \left(\mathcal{C}_{cone}\right)_2\otimes\mathcal{E}_1 \arrow[rrrr]                                     &  &  &  & \left(\mathcal{C}_{cone}\right)_1\otimes\mathcal{E}_1  
\end{tikzcd}
\end{center}\caption{The chain complex of the reduced cone code, after each disc has been cellulated. The spaces $\widetilde{\mathcal{X}_i^\perp\otimes\mathcal{E}_0}$ are those spanned by the edges added to all of the discs in $\widetilde{\mathcal{R}_i\otimes\mathcal{E}_0}$. Note that we have used cellulations here that do not add vertices to the discs i.e. do not add Z - stabilisers to the code.}\label{reducedConeCodeChainComplex}
\end{figure}
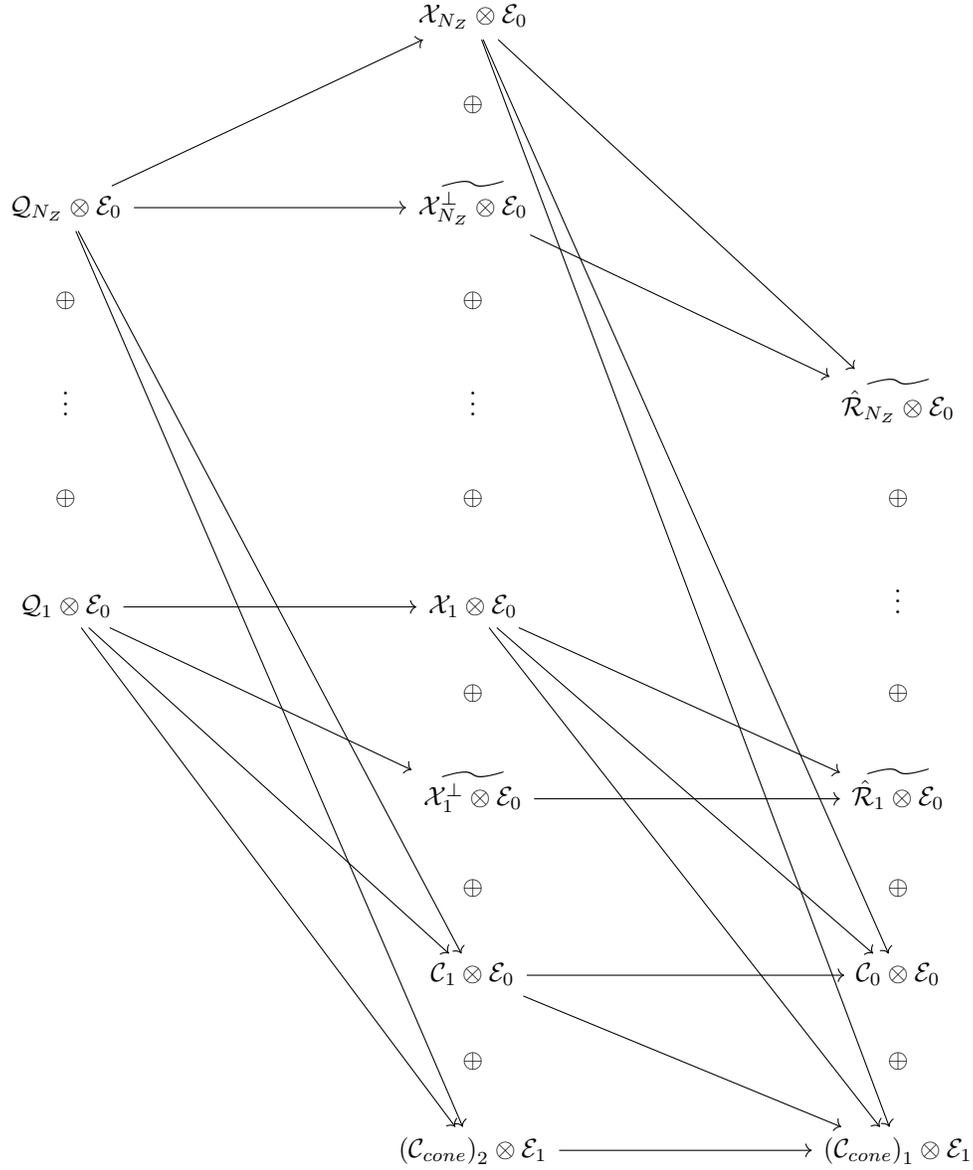

\begin{lemma}
    Let parameters with two primes denote those of the cone code after thickening and choosing heights. Let un-decorated parameters refer to those of the code pre-coning. Let parameters with a tilde refer to those of the reduced cone code. Then

    \begin{equation}
        \tilde{\rho}_Z \geq \frac{\tilde{N}}{\tilde{N}_X}\frac{N_X''}{N''}\frac{\rho_Z''}{1+\wz\left(\frac{N_X''\rho_Z''}{N''}+1\right)}.
    \end{equation}
\end{lemma}

\begin{proof}
    We recall that the chain complex is thickened enough and heights are chosen such that each Z - stabiliser in $\QQI \otimes \mathcal{E}_0$ has overlap with at most one X - stabiliser in $\widetilde{\RRI\otimes\mathcal{E}_0}$. Letting the cone code after thickening and choosing heights be denoted $C''$ and the reduced cone code be denoted $\tilde{C}$, choose some Z - operator in the reduced cone code represented by some bit string $u \in \mathbb{F}_2^{\tilde{N}}$. Without loss of generality, we may say that $u$ is not supported on any of the spaces $\widetilde{\mathcal{X}^\perp_{i}\otimes\mathcal{E}_0}$, which are the spaces of ``internal edges'' of the discs, i.e. those edges that are added in the cellulation step, as opposed to the external edges of the discs, i.e. those that were already there. This is true because we can without loss of generality add to $u$ the boundaries of various Z - stabilisers in $\QQI \otimes \mathcal{E}_0$. If $u$ is supported on some internal edge, we can add to $u$ the boundary of one of the Z - stabilisers represented by a vertex on the end of this edge to remove $u$'s support on that internal edge. Because the chain complex was thickened such that each Z - stabiliser in $\QQI \otimes \mathcal{E}_0$ has overlap with at most one X - stabiliser in $\widetilde{\RRI\otimes\mathcal{E}_0}$, each vertex is in contact with at most one internal edge in any disc, and as such we can indeed push $u$'s support off of the internal edges.

    With this done, $u$ is without loss of generality supported only on the qubits of $C''$. We will use our lower bound on the soundness of $C''$, the cone code post thickening and choosing heights, to establish a lower bound on the soundness of the reduced cone code. For this, we will now show how to put $u$ in $\tilde{C}$. Let $\hat{u} \in \mathbb{F}_2^{N''}$ be the projection of $u$ into the $N''$ bits corresponding to the qubits of $C''$. We can flip $d(\hat{u},\ker H_X'')$ bits in $\hat{u}$ to put $\hat{u}$ in $C''$, where $H_X''$ is the parity-check matrix associated with the X - stabilisers of $C''$, the cone code after thickening and choosing heights. We flip the bits in $u$ corresponding to these bits. Then, $\hat{u}$ must satisfy all of the X - stabilisers in $C''$, however, $u$ may not necessarily satisfy all of the X - stabilisers in $\tilde{C}$. The reason for this is that, in satisfying all of the X - stabilisers of $C''$, $\hat{u}$ must have even support on the external edges of every disc. However, even with this being true, $\hat{u}$ may have odd support on the boundaries of some of the surfaces of the cellulated disc i.e. some X - stabilisers in $\widetilde{\hat{\RRI}\otimes\mathcal{E}_0}$, the space of cellulated discs (see Figure \ref{reducedConeCodeChainComplex}), may be violated.
    
    We may, however, satisfy these stabilisers only by flipping bits in $u$ corresponding to internal edges. This will not affect any other X - stabilisers, because, as shown in Figure \ref{reducedConeCodeChainComplex}, qubits in $\widetilde{\XXI^{\perp}\otimes\mathcal{E}_0}$ are only acted on by X - stabilisers in $\widetilde{\hat{\RRI}\otimes\mathcal{E}_0}$. Furthermore, because $\hat{u}$ has even support on the external edges of the whole disc, there must now be an even number of violated stabilisers represented by faces in $\widetilde{\hat{\RRI}\otimes\mathcal{E}_0}$. Because every internal edge in each disc is adjacent to two faces, any even number of these violated stabilisers may be fixed by flipping bits in $u$ corresponding to internal edges of the discs. This finishes putting $u$ in $\tilde{C}$.

    Suppose that in this latter step, $f$ bits of $u$ are flipped in total that correspond to internal edges of various discs. Then

    \begin{equation}
        d(u, \ker \tilde{H}_X) \leq d(\hat{u}, \ker H_X'') + f
    \end{equation}
    where $\tilde{H}_X$ denotes the parity-check matrix associated with the X - stabilisers of the reduced cone code. Let us consider first the case that $f \leq \alpha d(\hat{u}, \ker H_X'')$ for some positive number $\alpha$ to be determined. In this case, $d(u, \ker \tilde{H}_X) \leq (1 + \alpha)d(\hat{u}, \ker H_X'')$. $\hat{u}$ violates at least $\frac{N''_X}{N''}\rho_Z''d(\hat{u}, \ker H_X'')$ X - stabilisers in $C''$. For every X - stabiliser that is violated by $\hat{u}$ in $C''$, at least one is violated by $u$ in $\tilde{C}$, because $\hat{u}$ having odd support on the external edges of a given disc implies that $u$ has odd support on the boundary of at least one of the faces in the cellulation of that disc. Therefore, $u$ violates at least

    \begin{equation}
        \frac{N''_X}{N''}\rho_Z''d(\hat{u}, \ker H_X'') \geq \frac{1}{1+\alpha}\frac{N''_X}{N''}\rho_Z''d(u, \ker \tilde{H}_X)\label{cellZVio1}
    \end{equation}
    X - stabilisers in $\tilde{C}$. Suppose now that $f \geq \alpha d(\hat{u}, \ker H''_X)$. Then we have $d(u, \ker \tilde{H}_X) \leq f\left(\frac{\alpha+1}{\alpha}\right)$. Because there are at most \wz external edges around one of the discs, there are at most \wz internal edges in a cellulated disc. There are therefore at least $\frac{f}{\wz}$ discs which have their internal edge bits flipped when we put $u$ in the code. Furthermore, there are at most $d(\hat{u}, \ker H_X'')$ discs which have their external edge bits flipped when we put $u$ in the code. There are therefore at least

    \begin{equation}
        \frac{f}{\wz}-d(\hat{u}, \ker H_X'') \geq f\left(\frac{1}{\wz}-\frac{1}{\alpha}\right)
    \end{equation}
    discs with inner edge bits flipped and no external edge bits flipped. For such a disc, at least one of the X - stabilisers given by the faces of that cellulated disc must have been violated by $u$ and so $u$ violates at least

    \begin{equation}
        \left(\frac{1}{\wz}-\frac{1}{\alpha}\right)\frac{\alpha}{\alpha+1}d(u, \ker \tilde{H}_X)\label{cellZVio2}
    \end{equation}
    X - stabilisers in $\tilde{C}$. We now let $\alpha$ be such that the right-hand sides of Equations \eqref{cellZVio1} and \eqref{cellZVio2} are equal. This turns out to be $\alpha = \wz\left(\frac{N_X''\rho_Z''}{N''}+1\right)$. Then, there are at least

    \begin{equation}
        \frac{N_X''}{N''}\frac{\rho_Z''}{1+\wz\left(\frac{N_X''\rho_Z''}{N''}+1\right)}d(u, \ker \tilde{H}_X)
    \end{equation}
    X - stabilisers violated by $u$ in $\tilde{C}$, from which we deduce the required soundness.
\end{proof}

\begin{lemma}
    Let parameters with two primes denote those of the cone code after thickening and choosing heights. Let parameters with a tilde be those of the reduced cone code. Then

    \begin{equation}
        \tilde{\rho}_X \geq \frac{\tilde{N}}{N''}\rho_X''.
    \end{equation}
\end{lemma}

\begin{proof}
Let us consider some X - operator represented by some bit string $u \in \mathbb{F}_2^{\tilde{N}}$. Without loss of generality, $u$ is not supported on internal edges of any cellulated disc. This is because, if $u$ is supported on the internal edge of some disc, we can add the boundary of some collection of faces in that disc to push the support of $u$ onto the external edges of that disc. Because the set of Z - stabilisers of the reduced cone code, $\tilde{C}$, is the same as that of the cone code post-thickening and choosing heights, $C''$, we can put $u$ in $\tilde{C}$ simply by flipping $d(u, \ker H_Z'')$ bits in it, and therefore $d(u, \ker \tilde{H}_Z) \leq d(u, \ker H_Z'')$, abusing notation slightly to view $u$ simultaneously as a vector in $\mathbb{F}_2^{\tilde{N}}$ supported only on the $N''$ bits of $C''$, and as a vector in $\mathbb{F}_2^{N''}$.

As an operator in $C''$, $u$ will violate at least $\frac{N_Z''}{N''}\rho_X''d(u, \ker H_Z'')$ Z - stabilisers in that code, and this many are then also violated by $u$ in $\tilde{C}$, again because $u$ is not supported on any of the spaces $\widetilde{\XXI^\perp\otimes\mathcal{E}_0}$. $u$ therefore violates at least

\begin{equation}
    \frac{N_Z''}{N''}\rho_X''d(u, \ker \tilde{H}_Z)
\end{equation}
Z - stabilisers in $\tilde{C}$ and using $\tilde{N}_Z = N_Z''$, we deduce the required soundness.
\end{proof}

\subsection{Weight Reducing a Quantum Locally Testable Code}\label{fullWeightRed}

We will now fully analyse the lengths, number of stabilisers, weights and soundnesses of any quantum locally testable code under the successive application of copying, gauging, thickening and choosing heights, and coning \textit{in this order}, under the assumption that the code is reasonable before the coning step. For this, it suffices that the original code (before any weight reduction) satisfies $d_Z = \omega(\wz\wx)$. In total, from these four successive weight reduction steps, we have the following:

\begin{lemma}[Formal Version of Theorem \ref{wtRedThm}]
    Consider a quantum locally testable CSS code on $N$ physical qubits with weights \qx, \wx, \qz and \wz and whose Z - distance asymptotically exceeds the maximum weight of any Z - stabiliser multiplied by the maximum weight of any X - stabiliser; $d_Z = \omega(\wz\wx)$\footnote{This condition corresponds to the code before the last weight reduction step, coning, being ``reasonable''. This term was defined in Section \ref{coningPrelimSection}.}. Then, there is another quantum locally testable code whose locality is constant, whose number of physical qubits is $\tilde{N} \leq N\text{poly\;}(\qx,\wx,\qz,\wz)$\footnote{\label{foot}Note that these polynomial factors come with an additional $\text{polylog}(N)$ in the main result - Theorem 1 - of \cite{hastings2021quantum}. This does not seem to be necessary, although removing it essentially only affects the case of weight reducing a qLTC with constant locality (to one with a lower constant locality).}, whose dimension is equal to that of the original code and whose X - and Z - distances are each equal to at least their original values divided by $\text{poly\;}(\qx,\wx,\qz,\wz)$\textsuperscript{\ref{foot}}. A full lower bound on the soundness of the resulting code is presented in tables \ref{ZOpoSoundness} and \ref{XOpoSoundness} using tables \ref{allWeights}, \ref{allLengths}, \ref{allXStabs} and \ref{allZStabs} in Appendix \ref{WRFullParams} but in the common case that the original code satisfies $N=\Theta(N_X) = \Theta(N_Z)$, $\rho_X = \OO(1)$ and $\rho_Z = \OO(1)$, the resulting values of $\rho_X$ and $\rho_Z$ are at least their old values divided by $\text{poly\;}(\qx,\wx,\qz,\wz)$.
\label{wtRedThmFormal}\end{lemma}

Tables \ref{allWeights} to \ref{XOpoSoundness} in Appendix \ref{WRFullParams} can be used to compute lower bounds on the resulting soundness of the code after a full weight reduction. While these provide full detail, we emphasise that this level of detail is usually unnecessary, as the assumptions $N=\Theta(N_X) = \Theta(N_Z)$ and $\rho_X, \rho_Z = \OO(1)$ hold in many cases. Nevertheless, for every stage of the construction, table \ref{allWeights} presents upper bounds on the weights, tables \ref{allLengths}, \ref{allXStabs} and \ref{allZStabs} present upper and lower bounds on the lengths and numbers of stabilisers, while tables \ref{ZOpoSoundness} and \ref{XOpoSoundness} present lower bounds on the soundnesses of the Z - and X - operators, respectively. Bounds may be calculated iteratively, calculating parameters at a given stage of the construction in terms of those at the previous stage, where the final weight reduced code (the reduced cone code) has parameters denoted with a superscript $(4)$. For the common case that $N=\Theta(N_X) = \Theta(N_Z)$ and $\rho_X, \rho_Z = \OO(1)$, we may use these tables to check that

\begin{equation}
    \rho^{(4)}\geq \frac{\rho}{\text{poly}(\qx,\wx,\qz,\wz)}
\end{equation}
for both soundness parameters $\rho \in \{\rho_Z, \rho_X\}$. Similar cases where, for example, $N$ and $N_X$ or $N_Z$ scale slightly differently will likely give very similar, or the same results.

\section{Soundness Amplification}\label{soundnessAmplificationSection}

We now present our soundness amplification procedure for quantum locally testable codes. We will show that this construction can amplify the soundness of a qLTC to a constant while maintaining its number of qubits, its distance and its dimension, also showing bounds on the effect on the locality of the code: both on the stabiliser weight and qubit degree.

The procedure relies heavily on the use of lossless expander graphs which we now define.

\begin{definition}
    A bi-regular graph $G = (L \sqcup R, E)$ that is $D$-left-regular is a $(K_{max},\epsilon)$-lossless expander if for every subset of the left vertices $S \subseteq L$ for which $|S| \leq K_{max}$, we have that $|N(S)| \geq (1-\epsilon)|S|D$, where $N(S)$ denotes the set of (right) vertices adjacent to some vertex in $S$.
\end{definition}
Explicit families of such graphs with constant $D$ were constructed in \cite{capalbo2002randomness,golowich2024new} and their existence for non-constant $D$ is also well-known. In our case, treating the situation in which $D$ is non-constant is important and for the sake of self-containment, as well as to show the existence of such graphs with the best possible bound on their right-degree, which for us is crucial, the existence of the needed graphs is shown in Section \ref{expanderExistenceSection}. The lemma we prove is as follows.

\begin{lemma}
    For each $N_X,M \in \mathbb{N}$ with $M \leq N_X$, and every $\epsilon \in (0,1)$, there exists a $(K_{max},\epsilon)$-lossless expander with $N$ left vertices, $M$ right vertices, left degree $D$ for which the degree of every right vertex is at most $\left\lceil\frac{N_X D}{M}\right\rceil$, where
    \begin{align}
        D &= \frac{1}{\epsilon}\log\left(8e^2 \frac{N_X}{M}\right)\\
        K_{max} &= \frac{1}{2e} \frac{\epsilon^{\frac{\epsilon D}{\epsilon D-1}}M}{D^{\frac{\epsilon D}{\epsilon D-1}}}.
    \end{align}\label{expanderExistenceLemma}
\end{lemma}
Let us comment that the bound on the degree of the right vertices $\left\lceil\frac{N_X D}{M}\right\rceil$ is optimal, and in the situation that $M$ divides $N_X D$, the graph is bi-regular.

We now describe the construction. The same procedure is applied to the $X$-stabilisers and to the $Z$-stabilisers; we describe it only for the $X$-stabilisers\footnote{In fact, this procedure could be used to amplify the soundness of a classical locally testable code.}. The construction is intended to be used to amplify sub-constant soundness in a qLTC to constant, and we only prove that it works in this way, although the proof can likely be easily extended to the case of amplifying constant soundness to a larger constant.

We start by fixing some universal constant $\alpha \in (0,1)$ and setting $\kappa = \frac{1+\alpha}{2}$. Amplifying sub-constant soundness to a constant takes multiple rounds. One `round' of soundness amplification takes place as follow. Letting $N_X$ be the number of $X$-stabilisers in the old code, the new code has these same $N_X$ $X$-stabilisers acting on the same $N$ qubits. The new code has this set of qubits and no more but has additional $X$-stabilisers. There are multiple groups of stabilisers introduced where each new group is labelled by the index $i$ taking values $i = \log(1/\rho^\kappa), \log(1/\rho^\kappa)+1, ..., \log(1/\rho)-1, \log(1/\rho)$. Because we have assumed that $\rho \to 0$ as $N \to \infty$, we note that we have $i \to \infty$ in the same limit. The new stabilisers in each group are simply sums of old stabilisers according to the edges of a lossless expander. For the $i$-th group, we consider a lossless expander with the following parameters:
\begin{align}
    M &= \frac{N_X}{2^{i\alpha}}\\
    \epsilon &= \sqrt{\frac{i\alpha}{2^{i(1-\alpha)}}},
\end{align}
and we take the values of $D$ and $K_{max}$ specified in Lemma \ref{expanderExistenceLemma}. In the $i$-th group, there are $M$ stabilisers introduced, each of which is the sum of the old stabilisers to which the corresponding right vertex is adjacent. The changes in the soundness and locality in one round of soundness amplification are given by the following lemma.

\begin{lemma}
    Assuming that the code's soundness $\rho$ is below some universal constant $\mu$\ to be determined, after one round of soundness amplification, the new soundness is at least $\tilde{\rho} \geq \gamma \rho^\kappa$ for some universal constant $\gamma$. If the original stabiliser weight and qubit degree were $w$ and $q$ respectively then the new stabiliser weight and qubit degree are respectively
    \begin{align}
        \tilde{w} &\leq \gamma_w\frac{w}{\rho^{\kappa+\delta_w}}\\
        \tilde{q} &\leq \gamma_q\frac{q}{\rho^{\frac{1-\alpha}{2}+\delta_q}}
    \end{align}
    for some universal constants $\gamma_w$, $\gamma_q$ and any desired universal constants $\delta_w$, $\delta_q > 0$.\label{oneRoundSALemma}
\end{lemma}

\begin{remark}
    In the statement of Lemma \ref{oneRoundSALemma}, and in its proof, we emphasise that anything we take to be a constant is a \textit{universal} constant, which will be important when we eventually apply multiple rounds of soundness amplification (a number of rounds scaling to infinity) to the code family (whose size scales to infinity).
\end{remark}

\begin{proof}[Proof of Lemma \ref{oneRoundSALemma}]
    In the lossless expander corresponding to the $i$-th group of new stabilisers, we note that
    \begin{equation}
        D = \Theta\left(\sqrt{i2^{i(1-\alpha)}}\right)
    \end{equation}
    where we note that the constants hidden in the $\Theta$ are universal. We then further find
    \begin{equation}
        \left(\frac{\epsilon}{D}\right)^{\frac{1}{\epsilon D-1}} = \Theta\left(\frac{1}{2^{i(1-\alpha)}}\right)^{\frac{1}{\Theta(i)-1}}
    \end{equation}
    which is then bounded above and below by constants for any $i \to \infty$ (or indeed any constant $i$). Choosing $\mu$ small enough ensures that $i$ is always large enough that this expression is between some universal constants. We may therefore write
    
    \begin{equation}
    K_{max} = \frac{1}{2e}\left(\frac{e}{D}\right)^{\frac{1}{\epsilon D-1}}\frac{\epsilon M}{D}= \Theta\left(\frac{\epsilon M}{D}\right) = \Theta\left(\frac{N_X}{2^i}\right)
    \end{equation}
    where again the constants hidden in the $\Theta$ are universal. This allows us to write $K_{max} \geq \nu \frac{N_X}{2^i}$ for some universal constant $\nu$. We take $\nu \leq 1$ without loss of generality.

    The number of $X$-stabilisers in the new code is
    \begin{equation}
        \tilde{N}_X = N_X + \sum_{i=\log(1/\rho^\kappa)}^{\log(1/\rho)}\frac{N_X}{2^{i\alpha}} = \Theta(N_X)
    \end{equation}
    where once again the corresponding constants are universal.

    Suppose now that a $Z$-operator represented by some bit string $x$ violates $K$ stabilisers in the old code, where we know that $K \geq \rho \frac{N_X}{N}d\left(x,\ker(H)\right)$ by the guarantee on soundness in the old code. We note that because the stabiliser group is unchanged in a round of soundness amplification (because the new code has the same stabilisers as the old code as well as additional ones that are the sum of old stabilisers), the value of $d\left(x,\ker(H)\right)$ is unchanged from the old code to the new code. Let us suppose further that $x$ violates fewer than $\nu\rho^\kappa\frac{N_X}{N}d\left(x,\ker(H)\right)$ $X$-stabilisers. Let us write
    \begin{equation}
        \frac{\nu}{2^{i+1}}\frac{N_X}{N}d\left(x,\ker(H)\right) \leq K < \frac{\nu}{2^i}\frac{N_X}{N}d\left(x,\ker(H)\right).
    \end{equation}
    where such an $i$ must exist because we have taken $\nu \leq 1$, and we may take $\mu$ to be a small enough universal constant that $\nu\rho^\kappa > \rho$, i.e. we take $\mu \leq \nu^{\frac{1}{1-\kappa}}$.
    
    We then consider the group of newly introduced stabilisers labelled by the index $i$. We note that we have $K < \frac{\nu N_X}{2^i} \leq K_{max}$. The set of left vertices $S$ corresponding to the violated stabilisers has at least $(1-\epsilon)KD$ neighbours on the right-hand side. It must, therefore have $(1-2\epsilon)KD$ `unique' neighbours, meaning neighbours on the right-hand side that are adjacent to exactly one vertex in $S$. This is true because
    \begin{equation}
        (1-\epsilon)KD \leq |N(S)| \leq KD - |N_{> 1}(S)|,
    \end{equation}
    where $N_{> 1}(S)$ is the set of non-unique neighbours of $S$, implying $|N_{> 1}(S)| \leq \epsilon KD$. Each unique neighbour of $S$ corresponds to a violated stabiliser in the new code, meaning that there are at least $(1-2\epsilon)KD$ violated stabilisers in the new code. Setting $\mu$ small enough ensures that $(1-2\epsilon) \geq 1/2$ because a larger $i$ leads to a smaller $\epsilon$, and so we may pick a small enough universal constant $\mu$ that $\epsilon \leq 1/4$. As a result, we see that the number of violated stabilisers is at least
    \begin{align}
        \Omega(KD) &\geq \frac{\Omega(\sqrt{i2^{i(1-\alpha)}})}{2^{i}}\frac{N_X}{N}d\left(x,\ker(H)\right)\\
        &\geq \Omega\left(\sqrt{\frac{i}{2^{i(1+\alpha)}}}\right)\frac{\tilde{N}_X}{N}d\left(x,\ker(H)\right) \\&\geq \frac{1}{2^{i\kappa}}\frac{\tilde{N}_X}{N}d\left(x,\ker(H)\right)\label{thirdLineIneq}
    \end{align}
    where going into the second line we have absorbed the (universal) constant difference between $N_X$ and $\tilde{N}_X$ into $\Omega$ and going into the third line we have picked a small enough $\mu$ (ensuring that $i$ is large enough) that the inequality of Equation \eqref{thirdLineIneq} holds. In the new code, $x$ therefore violates at least $\rho^\kappa \frac{\tilde{N}_X}{N}d\left(x,\ker(H)\right)$ stabilisers.

    We find that any $Z$-operator $x$ that violated fewer than $\nu\rho^\kappa \frac{N_X}{N}d\left(x,\ker(H)\right)$ stabilisers now violates at least $\rho^\kappa \frac{\tilde{N}_X}{N}d\left(x,\ker(H)\right)$ stabilisers. From this, we may conclude that the new code has soundness at least $\nu \rho^\kappa \frac{N_X}{\tilde{N}_X} \geq \gamma \rho^\kappa$ for some universal constant $\gamma$.

    For the stabiliser weight, suppose that our initial (maximum) stabiliser weight is $w$. The right degree in the $i$-th lossless expander is at most $\left\lceil\frac{N_XD}{M}\right\rceil = \Theta\left(\sqrt{i}2^{i\kappa}\right) \leq  \mathcal{O}\left(\sqrt{\log(1/\rho)}\frac{1}{\rho^\kappa}\right) \leq \mathcal{O}\left(\frac{1}{\rho^{\kappa + \delta_w}}\right)$ for any fixed $\delta_w > 0$. Accordingly, the new stabiliser weight is at most
    \begin{equation}
        \mathcal{O}\left(\frac{w}{\rho^{\kappa+\delta_w}}\right).
    \end{equation}
    Noting again that the hidden constant is universal, the new stabiliser weight is at most
    \begin{equation}
        \gamma_w \frac{w}{\rho^{\kappa+\delta_w}}
    \end{equation}
    for some universal constant $\gamma_w$.

    Lastly, suppose that the qubit degree before this round of soundness amplification is $q$. Noting that the left degree of the $i$-th expander is $D = \Theta\left(\sqrt{i2^{i(1-\alpha)}}\right)$, the qubit degree is now at most
    \begin{equation}
        q + q\Theta\left(\sum_{i=(1/\rho^\kappa)}^{\log(1/\rho)}\sqrt{i2^{i(1-\alpha)}}\right) \leq q + q \Theta\left(\sum_{i=\log(1/\rho^\kappa)}^{\log(1/\rho)}\sqrt{2^{i(1-\alpha+2\delta_q)}}\right) \leq \Theta\left(\frac{q}{\rho^{\frac{1-\alpha}{2}+\delta_q}}\right).
    \end{equation}
    for any fixed $\delta_q > 0$. There is therefore some universal constant $\gamma_q$ such that the new qubit degree is at most
    \begin{equation}
        \gamma_q\frac{q}{\rho^{\frac{1-\alpha}{2}+\delta_q}}.
    \end{equation}    
\end{proof}

Applying multiple rounds of soundness amplification to a code with $\rho \to 0$ as $N \to \infty$ allows us to amplify its soundness to a constant, as is captured in the following result.

\begin{lemma}
Consider a quantum locally testable CSS code with soundness $\rho$, maximum stabiliser weight $w$ and maximum qubit degree $q$ which has $N_X$ and $N_Z$ $X$- and $Z$-stabilisers respectively. Then, there is another quantum locally testable code with constant soundness and of the same dimension, distance and number of qubits as the old one that has stabiliser weight at most $w\;\text{poly}(1/\rho)$ and qubit degree at most $q\;\text{poly}(1/\rho)$. In both cases, the $\text{poly}(1/\rho)$ factors may be taken to be $\mathcal{O}\left(\frac{1}{\rho^{1+\delta}}\right)$ for any desired $\delta > 0$. Finally, the new code has $\tilde{N}_X$ and $\tilde{N}_Z$ $X$- and $Z$-stabilisers, where $N_X \leq \tilde{N}_X$ and $\tilde{N}_X = \mathcal{O}(\text{polylog}\;(1/\rho)N_X)$, and similarly for $\tilde{N}_Z$.\label{SALemma}
\end{lemma}
\begin{proof}[Proof of Lemma \ref{SALemma}]
    If the input code already has constant soundness, then we do nothing. If the input code has soundness approaching zero as $N \to \infty$, then we apply multiple rounds of soundness amplification. Soundness amplification leaves the number of qubits, distance and dimension unchanged because, as previously observed, the stabiliser group is invariant.
    
    For each member of the code family, we consider its soundness. If its soundness is below the universal constant $\min(\mu,\gamma^{\frac{1}{1-\kappa}}/2)$ (where we recall that $\mu$ and $\gamma$ are universal constants given by Lemma \ref{oneRoundSALemma} and its proof, $\kappa = \frac{1+\alpha}{2}$ and $\alpha \in (0,1)$ is a universal constant that may be chosen), then we amplify its soundness until it is above this value. We note that $\rho < \mu$ allows us to obtain a new soundness after one round of soundness amplification of at least $\gamma \rho^\kappa$. Furthermore, if $\rho < \gamma^{\frac{1}{1-\kappa}}$, then obtaining a new soundness of at least $\gamma \rho^\kappa$ represents an increase and indeed the soundness will approach the value $\gamma^{\frac{1}{1-\kappa}}$ by repeated application of soundness amplification; in particular it will eventually exceed $\gamma^{\frac{1}{1-\kappa}}/2$.

    After $n$ rounds of applying soundness amplification, we have soundness at least
    \begin{equation}
        \gamma^{1+\kappa + \kappa^2 + ... + \kappa^{n-1}}\rho^{\kappa^n} > \gamma^{\frac{1}{1-\kappa}}\rho^{\kappa^n}
    \end{equation}
    and so we apply $n = \Theta\left(\log\log(1/\rho)\right)$ rounds to obtain constant soundness. Moreover, after $n$ rounds, the stabiliser weight will be at most
    \begin{equation}
        \gamma_w^n\frac{w}{\rho^{(\kappa+\delta_w)(1+\kappa+\kappa^2+...+\kappa^{n-1})}} < \gamma_w^n\frac{w}{\rho^\frac{\kappa+\delta_w}{1-\kappa}} = \Theta\left(\log(1/\rho)\right)\frac{w}{\rho^\frac{\kappa+\delta_w}{1-\kappa}} = \mathcal{O}\left(\frac{w}{\rho^\frac{\kappa+\delta_w'}{1-\kappa}}\right)\label{finalStabWeight}
    \end{equation}
    where we recall that $\delta_w>0$ is a constant that may be freely chosen. $\delta_w'>0$ may similarly be freely chosen (depending on $\delta_w$). In the same way, after $n$ rounds, the qubit degree will be at most
    \begin{equation}
        \gamma_q^n\frac{q}{\rho^\frac{(1-\alpha)/2+\delta_q}{1-\kappa}} = \mathcal{O}\left(\frac{q}{\rho^\frac{(1-\alpha)/2+\delta_q'}{1-\kappa}}\right).\label{finalQubitDegree}
    \end{equation}
    for any $\delta_q'>0$ that we may choose freely (depending on $\delta_q$). We see that by choosing small enough universal constants $\alpha, \delta_w', \delta_q'$, we obtain a factor in the denominator of the right-hand sides of Equations \eqref{finalStabWeight} and \eqref{finalQubitDegree} of $\rho^{1+\delta}$ for any desired $\delta > 0$.
    
    Lastly, we recall from the proof of Lemma \ref{oneRoundSALemma} that after one round of soundness amplification, the new number of $X$-stabilisers is at least the old number and at most $\gamma_XN_X$ for some universal constant $\gamma_X$. After $n$ rounds, we will have at most $\gamma_X^nN_X = \mathcal{O}(\text{polylog}(1/\rho)N_X)$ $X$-stabilisers. The same argument holds for the $Z$-stabilisers.
\end{proof}

\subsection{Existence of the Lossless Expanders for Soundness Amplification}\label{expanderExistenceSection}

In this section, we prove Lemma \ref{expanderExistenceLemma} which gives the existence of the lossless expanders used in the soundness amplification procedure. The methods used are quite standard, although the proof is included because we are unaware of a similar proof with a similar bound on the degree of the right vertices, which is important for our construction. Indeed, to impose this, we use a slightly non-standard sampling procedure in the probabilistic method.

\begin{proof}[Proof of Lemma \ref{expanderExistenceLemma}]
    Consider a bipartite graph with $N_X$ left vertices and $M$ right vertices with left degree $D$. We consider there to be $N_XD$ half-edges leaving the left vertices and $N_XD$ half-edges leaving the right vertices. To do this, on the right side, we fix some set of the vertices to have degree $\left\lceil\frac{N_XD}{M}\right\rceil$, and the rest will have degree $\left\lfloor\frac{N_XD}{M}\right\rfloor$. One by one, we join the half-edges on the left side to the half-edges on the right side uniformly at random, therefore sampling uniformly from the $(N_XD)!$ possible graphs. We note that the distribution over the graphs obtained is the same regardless of the order in which the edges are joined.

    For any $K \leq K_{max}$, let $p_K$ be the probability that there is some set of left vertices $S$ of size $K$ for which $|N(S)| < (1-\epsilon)KD$. We consider some fixed such set $S$ and the $KD$ half-edges leaving $S$ and joining to vertices $v_1, ..., v_{KD}$ on the right-hand side, where there may be repeats in this list. The probability that $|N(S)| < (1-\epsilon)KD$ is less than or equal to the probability that there is some $\epsilon KD$-subset $\{v_{i_1}, ..., v_{i_{\epsilon KD}}\}$\footnote{This is not strictly a set but a multiset because vertices may be repeated within it.} for which every member is repeated outside of this subset within the list $v_1, ..., v_{KD}$. Consider some fixed such subset $\{v_{i_1}, ..., v_{i_{\epsilon KD}}\}$.

    Since the distribution over the graphs obtained is the same irrespective of the order in which the destinations of the half-edges are chosen, we imagine first that the vertices within $v_1, ..., v_{KD}$ are chosen outside of the subset $\{v_{i_1}, ..., v_{i_{\epsilon KD}}\}$, immediately followed by this subset. The probability that each of the vertices in this subset is chosen to be a vertex outside of this subset but within $v_1, ..., v_{KD}$ is upper bounded by the probability of the same event happening when $\{v_{i_1}, ..., v_{i_{\epsilon KD}}\}$ are instead chosen with a uniform and independent distribution over the right vertices\footnote{This is true because, under the sampling procedure, a right vertex becomes less likely to be chosen again if it has already been chosen.}, which is less than or equal to $\left(\frac{KD}{M}\right)^{\epsilon KD}$.

    By a union bound over all the $\epsilon KD$-subsets $\{v_{i_1}, ..., v_{i_{\epsilon KD}}\}$ of $v_1, ..., v_{KD}$, we find that the probability that this fixed $S$ has fewer than $(1-\epsilon)KD$ neighbours is at most $\begin{pmatrix} KD \\ \epsilon KD \end{pmatrix} \left(\frac{KD}{M}\right)^{\epsilon KD}$. By a further union bound over all possible such sets $S$, we find that
    \begin{equation}
        p_K \leq \begin{pmatrix} N_X \\ K \end{pmatrix} \begin{pmatrix} KD \\ \epsilon KD \end{pmatrix} \left(\frac{KD}{M}\right)^{\epsilon KD} \leq \left(\frac{N_X}{K}\frac{e^{\epsilon D+1}}{\epsilon^{\epsilon D}}\frac{K^{\epsilon D}D^{\epsilon D}}{M^{\epsilon D}}\right)^K = \left(\frac{N_X}{M}\frac{e^{\epsilon D+1}}{\epsilon^{\epsilon D}}\frac{K^{\epsilon D-1}D^{\epsilon D}}{M^{\epsilon D-1}}\right)^K
    \end{equation}
    where we used $\begin{pmatrix}n\\k\end{pmatrix}\leq \left(\frac{ne}{k}\right)^k$. Using $K \leq K_{max} = \frac{1}{2e} \frac{\epsilon^{\frac{\epsilon D}{\epsilon D-1}}M}{D^{\frac{\epsilon D}{\epsilon D-1}}}$ gives
    \begin{equation}
        p_K \leq \left(\frac{N_X}{M}\frac{2e^2}{2^{\epsilon D}}\right)^K
    \end{equation}
    and then using $D = \frac{1}{\epsilon}\log(8e^2\frac{N_X}{M})$ gives $p_K \leq 4^{-K}$. Then, the probability that some set of size $K \leq K_{max}$ exists with fewer than $(1-\epsilon)KD$ neighbours is at most
    \begin{equation}
        \sum_{K=1}^{K_{max}}p_K < \frac{1}{2}
    \end{equation}
    and so there exists some graph in the distribution that is a $(K_{max},\epsilon)$-lossless expander.
\end{proof}

\section{Distance Amplification}\label{DASection}

The Alon-Edmunds-Luby (AEL) distance amplification technique \cite{alon1995linear} is a well-established technique in classical coding theory. It is common for this technique to be used to boost the distance of a code, and then one may show that the resulting code (to some extent) inherits a desired property of the original code, such as local testability \cite{kopparty2017high, gopi2018locally, hemenway2019local}. This technique has only been used once before on the quantum side \cite{bergamaschi2022approaching} and has not been applied before to quantum locally testable codes. We show here how it could be used to construct a linear distance qLTC from an already-high distance qLTC, but we emphasise that this technique finds no present applications, but hopefully will in future, because no high-distance qLTCs are known to exist.

Our treatment is similar to that of \cite{bergamaschi2022approaching}, with modifications arising because we specifically want to work with codes over qubits, rather than codes over some larger alphabet. Suppose we have our original qLTC, called the `outer' code, already with a high distance which we wish to increase to linear. We label all of the parameters of this code with a subscript ``out''; for example, this code has $N_{out}$ physical qubits. Distance amplification then proceeds as follows:

\begin{enumerate}
    \item Consider some good quantum CSS LDPC code\footnote{Such a code has linear dimension, linear distance and constant locality. These codes are known to exist \cite{panteleev2022asymptotically,leverrier2022quantum,dinur2023good}.}, which we will call the `inner' code, whose parameters are labelled with a subscript ``in''. The first step of distance amplification is to divide the $N_{out}$ qubits into $b$ blocks of $K_{in}$ qubits each and encode each of these blocks into the inner code\footnote{For the sake of brevity, we avoid rounding issues by assuming that $K_{in}$ divides $N_{out}$.}. Each different set of qubits resulting from every copy of the inner code is referred to as a ``block'', so each block contains $N_{in}$ qubits at this stage, but the number of qubits in a block will change over the course of the construction.
    \item Second, we will apply a pseudorandom permutation to all the $b\cdot N_{in}$ qubits of the concatenated code, according to a bipartite expander graph that we will describe. We think of ``blocks'' as describing positions of qubits, not the qubits themselves, so we think of the qubits as moving from one block to another.
    \item The last step is then to encode each block into another good quantum LDPC CSS code, which we call the ``block'' code. Parameters of this code are denoted with a subscript ``block''. For clarity, we say that, for example, the qubits in the first block before the permutation may fill positions 1 to 10 and after the permutation may fill various positions. It is the new qubits in positions 1 to 10 that are then encoded into the ``block'' code, not the qubits that were previously in positions 1 to 10.
\end{enumerate}

Note that the result of this procedure is a CSS code since, as described in Section 2.3 of \cite{bergamaschi2022approaching}, the concatenation of two CSS codes may be chosen to be a CSS code. 

We will go on to show that if we use a certain permutation, if the inner and block codes each have linear distance, and if the outer code already has high distance, the resulting code will have linear distance, and we can obtain a lower bound on the soundness of the resulting code in terms of that of the outer code. Let us now describe the bipartite graphs used to define the pseudorandom permutation.

\begin{definition}
    An $N_{in}$-regular bipartite graph with vertex set $U \cup V$, where $|U| = |V| = b$, is called $\epsilon$-pseudorandom if for every pair of vertex sets $S \subseteq U$, $T \subseteq V$, 

    \begin{equation}
        \left||E(S,T)|-\frac{N_{in}|S||T|}{b}\right| \leq \epsilon N_{in}\sqrt{|S||T|}.
    \end{equation}
\label{pseudorandom}\end{definition}

By considering Ramanujan graphs, it is then true that explicit families of such graphs (with $b \rightarrow \infty$) exist as long as $N_{in} \geq \frac{4}{\epsilon^2}$ \cite{bergamaschi2022approaching}. With such a graph, we can define the permutation that is made in the second step of the above procedure. At the start of step 2, assume we have a code with $b = \frac{N_{out}}{K_{in}}$ blocks, where each block contains $N_{in}$ qubits. Then, we may associate the $j$-th qubit in the $i$-th block with the $j$-th edge leaving the $i$-th vertex of $U$ in the bipartite graph (where we have assigned some arbitrary numbering to the edges leaving each vertex). Assuming this edge forms the $j'$-th edge arriving at the $i'$-th vertex of $V$, this qubit is moved to the position of the $j'$-th qubit of the $i'$-th block. Given this definition of the permutation, and assuming that the graph satisfies Definition \ref{pseudorandom}, we have the following lemma, which is essential for our distance proof:

\begin{lemma}[Lemma 5.2 of \cite{bergamaschi2022approaching}]
    Consider any numbers $\alpha_{out}, \alpha_{in} \in (0,1)$ and some subset of the qubit blocks $T$ after the permutation of size $|T| \leq \left(\alpha_{in}-\epsilon\sqrt{\frac{\alpha_{in}}{\alpha_{out}}}\right)b$. There are then at most $\alpha_{out}b$ qubit blocks before the permutation containing more than $\alpha_{in}N_{in}$ qubits that get mapped to one of the blocks in $T$.
\label{permLemma}\end{lemma}

\begin{proof}
    Let $S$ be the set of qubit blocks which contain more than $\alpha_{in}N_{in}$ qubits that get mapped to one of the blocks of $T$. We aim to show that $|S| \leq \alpha_{out}b$. We imagine $S$ and $T$ as subsets of the vertex sets $U$ and $V$ respectively in the bipartite graph that we use for the permutation. We have 
    
    \begin{equation}
        |E(S,T)| \geq |S| \alpha_{in}N_{in}\label{pseudorandomEqn1}
    \end{equation}
    by the definition of $S$. Also, we have from Definition \ref{pseudorandom} that

    \begin{align}
        |E(S,T)| &\leq \frac{N_{in}|S||T|}{b} + \epsilon N_{in}\sqrt{|S||T|}\\
        &< |S|\alpha_{in}N_{in}-|S|\epsilon N_{in}\sqrt{\frac{\alpha_{in}}{\alpha_{out}}}+|S|\epsilon N_{in}\sqrt{\frac{\alpha_{in}}{|S|/b}}\label{pseudorandomEqn2}
    \end{align}
    where in the first term we use $|T| \leq \left(\alpha_{in}-\epsilon\sqrt{\frac{\alpha_{in}}{\alpha_{out}}}\right)b$ and in the second term we use simply $|T| < \alpha_{in}b$. Equations \eqref{pseudorandomEqn1} and \eqref{pseudorandomEqn2} are in contradiction if $\frac{|S|}{b} > \alpha_{out}$, and so we have $\frac{|S|}{b} \leq \alpha_{out}$, and then the result.
\end{proof}

The proof of the distance of the final code naturally follows from this lemma. We may prove that the final code has distance at least $d$ by showing that a word can be correctly decoded if there are errors on at most $\frac{d}{2}$ of its physical qubits. With this in mind, we have:

\begin{lemma}
    Let parameters with and without a tilde be those of the final, distance amplified code, and the original, `outer' code, respectively. For any code, we let the relative distance be denoted $\Delta = \frac{d}{N}$ i.e. the distance divided by the number of physical qubits. The final code has relative distance

    \begin{equation}
        \tilde{\Delta} \geq \Delta_{block}\left(\frac{\Delta_{in}}{2}-\epsilon\sqrt{\frac{\Delta_{in}}{\Delta_{out}}}\right).
    \end{equation}
\label{DAdistanceLemma}\end{lemma}

\begin{proof}
    Suppose there are errors on a fraction of the qubits of the final code of size at most

    \begin{equation}
        \frac{\Delta_{block}}{2}\left(\frac{\Delta_{in}}{2}-\epsilon\sqrt{\frac{\Delta_{in}}{\Delta_{out}}}\right), 
    \end{equation}
    meaning there are at most $\frac{d_{block}}{2}\left(\frac{\Delta_{in}}{2}-\epsilon\sqrt{\frac{\Delta_{in}}{\Delta_{out}}}\right)b$ errorful qubits in the final code. There are therefore at most $\left(\frac{\Delta_{in}}{2}-\epsilon\sqrt{\frac{\Delta_{in}}{\Delta_{out}}}\right)b$ blocks containing more than $\frac{d_{block}}{2}$ errorful qubits. We can decode all of the blocks from the block code that contain at most $\frac{d_{block}}{2}$ errorful qubits, leaving at most $\left(\frac{\Delta_{in}}{2}-\epsilon\sqrt{\frac{\Delta_{in}}{\Delta_{out}}}\right)b$ blocks containing errors. Undoing the permutation will, by Lemma \ref{permLemma}, result in at most $\frac{\Delta_{out}}{2}b$ blocks containing more than $\frac{\Delta_{in}}{2}N_{in}$ errorful qubits. We can decode all the blocks from the inner code that contain at most $\frac{\Delta_{in}}{2}N_{in}$ errorful qubits, leaving at most $\frac{\Delta_{out}}{2}b$ blocks containing errors. This implies that we have $N_{out}$ qubits, of which at most $\frac{\Delta_{out}}{2}N_{out} = \frac{d_{out}}{2}$ are errorful, which is a word we may decode from the outer code.
\end{proof}

With this expression, we must take care that our final relative distance is in fact constant. Indeed, given that $\Delta_{block} = \Delta_{in} = \Theta(1)$, we take $\epsilon = \Theta(\sqrt{\Delta_{out}})$ to ensure $\tilde{\Delta} = \Theta(1)$. We must keep in mind, however, that for the existence of the expander graphs, we require $N_{in} = \Omega\left(\frac{1}{\epsilon^2}\right)$. It will make sense for the sake of local testability to have $N_{in}$ as small as possible, and so we indeed let $N_{in} = \Theta\left(\frac{1}{\epsilon^2}\right) = \Theta\left(\frac{1}{\Delta_{out}}\right)$. Seeing this, we find that we can, at least in principle, apply this technique in all cases where $\Delta_{out} = \omega\left(\frac{1}{N_{out}}\right)$, because we require the number of blocks $b = \frac{N_{out}}{K_{in}} = \Theta\left(\frac{N_{out}}{N_{in}}\right) \rightarrow \infty$. In practice, this technique would only be useful with $\Delta_{out}$ being greater than inverse polynomial, because using this technique with inverse polynomial $\Delta_{out}$ would quickly lead to very poor soundness and locality.

With this in mind, we will state and prove the full parameters of the final code in terms of the original `outer' code.

\begin{lemma}[Formal Version of Theorem \ref{DAInformalThm}]\label{DALemma}
    Let parameters with a tilde refer to those of the final code and parameters with a subscript ``out'' refer to those of the original code, where we require $d_{out} = \omega(1)$. The `inner' and `block' codes, denoted respectively with a subscript ``in'' and ``block'', are both good quantum LDPC CSS codes, where $K_{block} = N_{in}$, $N_{in}$ = $\Theta(N_{X,in}) = \Theta(N_{Z,in})$ and $N_{block} = \Theta(N_{X,block}) = \Theta(N_{Z,block})$. We define
    
    \begin{align}
        b &= \frac{N_{out}}{K_{in}}\\
        \hat{\rho}_{Z,out} &= \frac{N_{X,out}\rho_{Z,out}}{N_{out}}\\
        \hat{\rho}_{X,out} &= \frac{N_{Z,out}\rho_{X,out}}{N_{out}}\\
        \alpha_Z &= \frac{\hat{\rho}_{Z,out}}{N_{in}K_{in}w_{out}+\hat{\rho}_{Z,out}+K_{in}(1+N_{in})\left(\hat{\rho}_{Z,out}+w_{out}\right)+1}\\
        \alpha_X &= \frac{\hat{\rho}_{X,out}}{N_{in}K_{in}w_{out}+\hat{\rho}_{X,out}+K_{in}(1+N_{in})\left(\hat{\rho}_{X,out}+w_{out}\right)+1}.
    \end{align}With $N_{in} = \Theta\left(\frac{N_{out}}{d_{out}}\right)$, we have

    \begin{enumerate}
        \item $\tilde{N} = bN_{block}$, $\tilde{N}_X = N_{X,out} + bN_{X,in} + bN_{X,block}$, $\tilde{N}_Z = N_{Z,out} + bN_{Z,in} + bN_{Z,block}$.
        \item $\tilde{K} = K$.
        \item $\tilde{w} = \OO(w_{out}N_{in}^2)$.
        \item $\tilde{d} = \Theta(\tilde{N})$.
        \item $\tilde{\rho}_Z \geq \frac{\tilde{N}}{\tilde{N}_X}\frac{\alpha_Z}{N_{in}N_{block}}$.
        \item $\tilde{\rho}_X \geq \frac{\tilde{N}}{\tilde{N}_Z}\frac{\alpha_X}{N_{in}N_{block}}$.
    \end{enumerate}
    where we recall that $w$ refers to the overall locality of a code.
\end{lemma}

\begin{proof}
    The proof of distance is provided in Lemma \ref{DAdistanceLemma} and the proof of the soundness will be provided in Lemma \ref{DAsoundnessLemma}. We prove the remaining parameters here. Item 1 is true because the code finishes with $b$ blocks filled with $N_{block}$ qubits. There are also three types of X - stabilisers. First, there are the $N_{X,out}$ X - stabilisers of the outer code which get encoded when the initial concatenation happens with the inner code. These then get permuted and then encoded again into the block code. Second, there are the X - stabilisers of the inner code, of which there are $N_{X,in}$ for each of the $b$ copies of the code. These get permuted and then encoded into the block code. Finally, there are the X - stabilisers of the block code, of which there are $N_{X,block}$ for each of its $b$ copies. The same all goes for the Z - stabilisers. Item 2 is immediate because the original qubits encoded into the outer code are the only qubits encoded in the final code.

    For item 3, let us first consider the maximum weight of a stabiliser in the final code, where the following can apply to either X - or Z - stabilisers. Consider the first type of stabiliser in the previous paragraph, one arising from the outer code that gets encoded, permuted, and then finally encoded again. As a stabiliser of the outer code, it has maximum weight $w_{out}$ and the qubits on which it acts may be spread across at most $w_{out}$ blocks. Upon encoding into the inner code, each operator within each block becomes a logical operator of the inner code, which may have weight at most $N_{in}$, so the whole stabiliser may have weight at most $w_{out}N_{in}$. This weight is the same upon permutation. Finally, the qubits on which this stabiliser acts are spread across at most $w_{out}N_{in}$ blocks, and again, each operator within each block becomes a logical operator of the block code, giving a final weight at most $w_{out}N_{in}N_{block} = \OO(w_{out}N_{in}^2)$. It may be checked that the weights of the other stabilisers do not exceed this.

    Let us also consider qubit degree, where the following argument holds for either X or Z type stabilisers. Let us first consider just the stabilisers arising from the outer code. At most $w_{out}$ of these (of either X or Z type) can act on a given qubit in the outer code. After each block is encoded into the inner code, a qubit in a particular block of this concatenated code may be acted on by the encoding of any stabiliser that acted on some qubit of that block in the outer code, because every operator in each block gets encoded to a logical operator of the inner code. Therefore, after this initial concatenation step, each qubit is acted on by at most $w_{out}K_{in}$ stabilisers. This remains true after the permutation step. The same consideration then goes for the final encoding step; the number of stabilisers that may act on a given qubit in a particular block in the final code is at most the number of stabilisers that act somewhere in that block before the concatenation with the block code. This is then at most $w_{out}K_{in}N_{in} = \OO(w_{out}N_{in}^2)$. Contributions to the qubit degree from the other types of stabilisers can be checked to not exceed this.
\end{proof}

We may now prove a lower bound on the soundness parameters of the final code in terms of the original, outer code. There are similar ideas to the local testability proof of \cite{kopparty2017high}, although here we must address the differences in our construction, as well as our different definition of local testability. 

\begin{lemma}
Let parameters with and without a tilde denote those of the final code and the original, `outer' code, respectively. We have
\begin{equation}
    \tilde{\rho}_Z \geq \frac{\tilde{N}}{\tilde{N}_X}\frac{\alpha}{N_{in}N_{block}}
\end{equation}
where
\begin{equation}
    \alpha = \frac{\hat{\rho}_{Z,out}}{N_{in}K_{in}w_{out}+\hat{\rho}_{Z,out}+K_{in}(1+N_{in})\left(\hat{\rho}_{Z,out}+w_{out}\right)+1}.
\end{equation}
and
\begin{equation}
    \hat{\rho}_{Z,out} = \frac{N_{X,out}\rho_{Z,out}}{N_{out}}.
\end{equation}
Under the reasonable assumptions that $N_{out} = \Theta(N_{X,out}) = \Theta(N_{Z,out})$ and $\rho_{out} = \OO(1)$, we have
\begin{equation}
    \tilde{\rho}_Z = \Omega\left(\frac{\rho_{Z,out}}{N_{in}^4w_{out}}\right).
\end{equation}
The corresponding expressions for $\tilde{\rho}_X$ can be found by swapping all X's and Z's in the above expressions.
\label{DAsoundnessLemma}\end{lemma}

\begin{proof}
    The construction is symmetric in $Z$ and $X$ and so it will suffice for us to prove the expression for $\tilde{\rho}_Z$ only.

    Let us first define the following codes. We consider the initial outer code, $C_{out}$ on $N_{out}$ qubits. We then define the code $C_{cat}$ to be the code immediately after the concatenation with the inner code, which is on $bN_{in}$ qubits, recalling the number of blocks is $b = \frac{N_{out}}{K_{in}}$. Next, we define the code $C_{perm}$ to be the code immediately after the permutation step, also on $bN_{in}$ qubits. Lastly, the final code, after the concatenation with the block code, is denoted $\tilde{C}$, and is on $\tilde{N}=bN_{block}$ qubits. As usual, we will consider a Z - operator on the final code represented by some bit string $z \in \mathbb{F}_2^{\tilde{N}}$. For brevity, we will abuse notation slightly and also refer to the classical codes defined by the X - stabilisers of various quantum codes by those quantum codes themselves, so for example the classical code defined by the X - stabilisers of the final code will also be denoted $\tilde{C} \subseteq \mathbb{F}_2^{\tilde{N}}$.

    It will make sense to give names to the three types of stabilisers. Recall that there are the stabilisers arising from the outer code that are encoded into the inner code, permuted, and finally encoded into the block code. We call these the ``encoded outer stabilisers''. Secondly, there are the stabilisers arising from the inner code that are permuted and then encoded into the block code. We call these the ``encoded inner stabilisers''. Lastly, there are the stabilisers arising from the block code, which we call the ``block stabilisers''.

    Given the vector $z \in \mathbb{F}_2^{\tilde{N}}$, we may define a vector 

    \begin{equation}
        w^z \in \left(\mathbb{F}_2\cup\{?\}\right)^{bN_{in}}
    \end{equation}
    meaning a vector of length $bN_{in}$ on the set of elements $\{0,1,?\}$, where a ``$?$'' is called an ``erasure''. $w^z$ is defined as follows. For each block of $z$ that is a valid codeword of the block code, we decode this codeword and place the result in the corresponding block of $w^z$. For each block of $z$ that is not a valid codeword of the block code, we define each character in the corresponding block of $w^z$ to be $?$. We let $E_w$ be the set of erasures in $w^z$ i.e.
    \begin{equation}
        E_w =\{i \in [bN_{in}] \text{ s.t.} \left(w^z\right)_i = \; ?\}.
    \end{equation}
    Next, we define a binary vector $w_0^z \in \mathbb{F}_2^{bN_{in}}$ that agrees with $w^z$ outside of $E_w$, but is otherwise arbitrary. We will then undo the action of the permutation on both $w^z$ and $w_0^z$ to define respectively $v^z \in \left(\mathbb{F}_2\cup\{?\}\right)^{bN_{in}}$ and $v_0^z \in \mathbb{F}_2^{bN_{in}}$. We may then define the set of erasures in $v^z$ as $E_v$ i.e.

    \begin{equation}
        E_v =\{i \in [bN_{in}] \text{ s.t.} \left(v^z\right)_i = \; ?\}.
    \end{equation}
    We will then define a vector $u^z \in \left(\mathbb{F}_2\cup\{?\}\right)^{bK_{in}}$ given by taking $v_0^z$ and decoding every block that is a valid codeword of the inner code, and replacing every block that is not with erasures. Lastly, we may define $E_u$ to be the set of erasures in $u^z$ i.e.

    \begin{equation}
        E_u =\{i \in [bK_{in}] \text{ s.t.} \left(u^z\right)_i = \; ?\}.
    \end{equation}
    and define a binary vector $u_0^z \in \mathbb{F}_2^{bK{in}}$ that agrees with $u^z$ outside of $E_u$ but is otherwise arbitrary.

    We may now give some preliminary facts that will be true in all cases. We have that

    \begin{equation}
        d(z,\tilde{C}) \leq N_{block}d(w^z,C_{perm})\label{blockRelation}
    \end{equation}
    because the change of one character in $w^z$ can affect the change of at most $N_{block}$ bits in $z$\footnote{Note that we are able to still talk about the quantity $d(w^z,C_{perm})$ even though $w^z$ is a vector over the set $\mathbb{F}_2\cup\{?\}$, simply as the minimum number of characters that must be changed in $w^z$ in order to put it in the code.}. Next, we have

    \begin{equation}
        d(w^z,C_{perm}) \leq d(w_0^z,C_{perm}) + \left|E_w\right|\label{wErasureRelation}
    \end{equation}
    since $w^z$ can be put in the code $C_{perm}$ by mirroring how $w_0^z$ is put in the code on its non-erasure bits, and then changing its erasure bits as necessary. Since the code $C_{cat}$ is just the ``un-permutation'' of the code $C_{perm}$, we have that 
    
    \begin{align}
        d(v^z,C_{cat}) &= d(w^z,C_{perm})\label{vwEquality}\\
        d(v_0^z,C_{cat}) &= d(w_0^z,C_{perm})\label{vw0Equality}\\
        |E_w| &= |E_v|
    \end{align}
    In particular, we have
    \begin{equation}
        d(v^z,C_{cat}) \leq d(v_0^z,C_{cat}) + |E_v|.
    \end{equation}
    Lastly, similar to Equations \eqref{blockRelation} and \eqref{wErasureRelation}, we have

    \begin{align}
        d(v_0^z,C_{cat}) &\leq N_{in}d(u^z,C_{out})\label{vToURelation}\\
        d(u^z,C_{out}) &\leq d(u_0^z,C_{out}) + |E_u|.
    \end{align}
    With all this in place, we may consider three cases. First, let us consider the case for which $|E_w| \geq \alpha d(w^z,C_{perm})$ for some $\alpha \in (0,1)$ to be determined. For each block in $w^z$ filled with erasures, there must be at least one violated block stabiliser, meaning that there are at least

    \begin{equation}
        \frac{|E_w|}{N_{in}} \geq \frac{\alpha}{N_{in}}d(w^z,C_{perm}) \geq \frac{\alpha}{N_{in}}\frac{d(z,\tilde{C})}{N_{block}}\label{firstVioStabs}
    \end{equation}
    violated X - stabilisers in this case. Next, let us consider the case for which $|E_w| \leq \alpha d(w^z,C_{perm})$ and $|E_u| \geq \beta d(u^z,C_{out})$ for some other number $\beta \in (0,1)$ to be determined. In this case, Equations \eqref{wErasureRelation}, \eqref{vwEquality} and \eqref{vw0Equality} give us that $d(v_0^z,C_{cat})\geq(1-\alpha)d(v^z,C_{cat})$. Let us define $F_v \subseteq [b]$ and $F_u \subseteq [b]$ respectively to be the sets of blocks in $v^z$ and $u^z$ that contain erasures. Note that $u^z$ contains erasures ``uniformly'' in its blocks i.e. if one character in a block of $u^z$ is an erasure, all of the characters in that block are erasures, whereas the erasures in $v^z$ may be spread ``non-uniformly''. We therefore have $|F_u| = \frac{|E_u|}{K_{in}}$, and we also have $|F_v| \leq |E_v|$. We now consider $F_u \setminus F_v$, the set of erasure blocks in $u^z$ for which there were no erasures in $v^z$. The number of violated encoded inner stabilisers is then at least

    \begin{align}
        \left|F_u \setminus F_v\right| &\geq |F_u| - |F_v|\\
        &\geq \frac{|E_u|}{K_{in}}-|E_v|\\
        &\geq \left(\frac{\beta(1-\alpha)}{K_{in}N_{in}}-\alpha\right)d(w^z,C_{perm})\\
        &\geq \left(\frac{\beta(1-\alpha)}{K_{in}N_{in}}-\alpha\right)\frac{d(z,\tilde{C})}{N_{block}}\label{secondVioStabs}
    \end{align}
    where going into the third line we have used Equations \eqref{vwEquality} and \eqref{vToURelation} as well as the facts specific to this case. Going into the last line, we have assumed that we have $\alpha$ and $\beta$ such that $\frac{\beta(1-\alpha)}{K_{in}N_{in}} - \alpha \geq 0$, which will turn out to be true for our eventual choices of $\alpha$ and $\beta$, and we have used Equation \eqref{blockRelation}.

    Lastly, we consider the case for which $|E_w| \leq \alpha d(w^z,C_{perm})$ and $|E_u| \leq \beta d(u^z,C_{out})$ for which we have $d(w_0^z,C_{perm}) \geq (1-\alpha)d(w^z,C_{perm})$ and $d(u_0^z,C_{out}) \geq (1-\beta)d(u^z,C_{out})$. Now, if the unencoded outer stabilisers act on $u_0^z$, there will be at least $\frac{N_{X,out}}{N_{out}}\rho_{Z,out}d(u_0^z,C_{out})$ violated stabilisers. We may then ask how many encoded outer stabilisers will be violated when acting on $z$. There are at most $(|E_v|K_{in} + |E_u|)w_{out}$ unencoded stabilisers from the outer code that, when acting on $u_0^z$, act on qubits that were erasures in $u^z$, or came from blocks containing erasures in $v^z$. There are therefore at least

    \begin{equation}
        \frac{N_{X,out}}{N_{out}}\rho_{Z,out}d(u_0^z,C_{out})-\left(|E_v|K_{in}+|E_u|\right)w_{out}
    \end{equation}
    violated encoded outer stabilisers when checking $z$. After some algebra, we find that this is at least

    \begin{equation}
        \left[  \left(  \frac{N_{X,out}}{N_{out}}\rho_{Z,out}(1-\beta) - \beta w_{out}      \right)\frac{1-\alpha}{N_{in}} - \alpha K_{in}w_{out}     \right]\frac{d(z,\tilde{C})}{N_{block}}\label{finalVioStabs}
    \end{equation}
    where in making this calculation, we assume that our $\alpha$ and $\beta$ are chosen such that the quantity in the square brackets of Equation \eqref{finalVioStabs} is non-negative, which will turn out to be the case.

    We then choose $\alpha$ and $\beta$ such that the right-hand sides of Equations \eqref{firstVioStabs}, \eqref{secondVioStabs} and \eqref{finalVioStabs} are equal, which are, with $\hat{\rho}_{Z,out} = \frac{N_{X,out}\rho_{Z,out}}{N_{out}}$, 

    \begin{align}
    \alpha &= \frac{\hat{\rho}_{Z,out}}{N_{in}K_{in}w_{out}+\hat{\rho}_{Z,out}+K_{in}(1+N_{in})\left(\hat{\rho}_{Z,out}+w_{out}\right)+1}\\
    \beta &=\frac{\hat{\rho}_{Z,out}\left(1+N_{in}\right)K_{in}}{2K_{in}N_{in}w_{out}+1+\hat{\rho}_{Z,out}\left(1+N_{in}\right)K_{in} + K_{in}w_{out}}.
    \end{align}
    The required soundness then follows from Equation \eqref{firstVioStabs}.
\end{proof}

\section*{Acknowledgements}

We are grateful for extensive and fruitful discussions with Matthew Hastings on his weight reduction techniques for quantum codes. We also thank Sergii Strelchuk for his continued support of our work. Thanks to Louis Golowich for a discussion on the asymptotically good testable quantum code. TCL was supported in part by funds provided by the U.S. Department of Energy (D.O.E.) under the cooperative research agreement DE-SC0009919.

\bibliographystyle{unsrt}
\bibliography{references}
\pagebreak
\appendix
\section{Tables of Parameters Under the Weight Reduction Construction}\label{WRFullParams}

We now present detailed tables showing the change in each parameter over each stage of the weight reduction procedure, as discussed in Section \ref{fullWeightRed}. Bounds may be calculated iteratively, moving in rows down the tables, where un-decorated parameters are those of the code pre-weight reduction, and those labelled with a superscript $(4)$ are those of the fully weight reduced code.

\renewcommand{\arraystretch}{1.5}
\begin{table}[h]
\centering
\begin{tabular}{c||c|c|c|c}
Original&\qx&\wx&\qz&\wz\\\hline\hline
Post-Copying&$\qx^{(1)}=\OO(1)$&$\wx^{(1)}=\OO(\wx)$&$\qz^{(1)}=\OO(\qz)$&$\wz^{(1)}=\OO(\qx\wz)$\\\hline
Post-Gauging&$\qx^{(2)}=\OO(1)$&$\wx^{(2)}=\OO(1)$&$\qz^{(2)}=\OO(\wx^{(1)}\qz^{(1)})$&$\wz^{(2)}=\OO(\wz^{(1)}\wx^{(1)})$\\\hline
\makecell{Post-Thickening \\ and Choosing Heights}&$\qx^{(3)}=\OO(1)$&$\wx^{(3)}=\OO(1)$&$\qz^{(3)}=\OO(1)$&$\wz^{(3)}=\OO(\wz^{(2)})$\\
\hline
Reduced Cone Code&$\qx^{(4)}=\OO(1)$&$\wx^{(4)}=\OO(1)$&$\qz^{(4)}=\OO(1)$&$\wz^{(4)}=\OO(1)$
\end{tabular}\caption{Upper bounds on the four weights throughout the four primary steps of the construction.}\label{allWeights}
\end{table}
\renewcommand{\arraystretch}{1}

\vspace{1cm}

\renewcommand{\arraystretch}{1.5}
\begin{table}[h]
\centering
\begin{tabular}{c||c}
Original&$N$\\\hline\hline
Post-Copying&$N^{(1)}=N\qx$\\\hline
Post-Gauging&$N^{(2)}=\Theta\left(N^{(1)}\right)$\\\hline
\makecell{Post-Thickening \\ and Choosing Heights}&$N^{(3)}=\Theta\left(l_1(N^{(2)}+N_X^{(2)})\right)$\\\hline
Cone Code&$N^{(4a)} = \OO(N^{(3)}+N_Z^{(3)}\wz^{(3)})$; $N^{(4a)} = \Omega(N^{(3)}+N_Z^{(3)})$\\\hline
\makecell{Thickened \\ Cone Code} & $N^{(4b)} = \Theta\left(l_2(N_Z^{(4a)}+N^{(4a)})\right)$\\\hline
\makecell{Thickened Cone Code \\ with Full Height Choice} & $N^{(4c)} = N^{(4b)}$\\\hline
\makecell{Thickened Cone Code \\ with Partial Height Choice} & $N^{(4d)} = N^{(4c)}$\\\hline
Reduced Cone Code&$N^{(4)} = \Theta(N^{(4d)})$
\end{tabular}\caption{Upper and lower bounds on the length of the code as it undergoes the four steps. We must define several lengths over the course of the coning procedure for the benefit of the later soundness bounds. In particular, it will be necessary for the benefit of $N_X$ that we differentiate between the parameters labelled $4c$ and $4d$, which are respectively those of the code under a ``full height choice'' (where we make a choice of heights for every X - stabiliser in $\left(\mathcal{C}_{cone}\right)_0\otimes\mathcal{E}_0$), and under the partial height choice that we actually use, where we only make a choice of heights for each X - stabiliser in $\RRI \otimes \mathcal{E}_0$, and not $\mathcal{C}_0\otimes\mathcal{E}_0$. The values of $l_1$ and $l_2$ are the amounts we thicken by the first and second time we thicken respectively, first to reduce \qz to \OO(1), second during the coning construction to reduce \qx to \OO(1). These are, where we may pick any $\epsilon > 0$, $l_1 = \Theta\left((\qz^{(2)})^{1+\epsilon}\min(\qz^{(2)}\wz^{(2)},N^{(2)})^{\OO(\epsilon)}\right)$ and $l_2 = \Theta\left(\wz^{(3)}\log(\wz^{(3)})\right)$.}\label{allLengths}
\end{table}
\renewcommand{\arraystretch}{1}

\renewcommand{\arraystretch}{1.5}
\begin{table}[]
\centering
\begin{tabular}{c||c}
Original&$N_X$\\\hline\hline
Post-Copying&$N_X^{(1)}=\Theta(N_X+N\qx)$\\\hline
Post-Gauging&$N_X^{(2)} = \Theta(N_X^{(1)})$\\\hline
\makecell{Post-Thickening \\ and Choosing Heights}&$N_X^{(3)} = N_X^{(2)}l_1$\\\hline
Cone Code&$N_X^{(4a)} = \OO(N_X^{(3)}+N_Z^{(3)}\wz^{(3)}\log\wz^{(3)})$; $N_X^{(4a)} = \Omega(N_X^{(3)}+N_Z^{(3)})$\\\hline
\makecell{Thickened \\ Cone Code} & $N_X^{(4b)} = \Theta\left(l_2\left(N_X^{(4a)}+N^{(4a)}\right)\right)$\\\hline
\makecell{Thickened Cone Code \\ with Full Height Choice} & $N_X^{(4c)}=\Theta\left(l_2N^{(4a)}+N_X^{(4a)}\right)$\\\hline
\makecell{Thickened Cone Code \\ with Partial Height Choice} & $N_X^{(4d)} = \Theta\left(l_2N^{(4a)}+l_2N^{(3)}+N_X^{(4a)}\right)$\\\hline
Reduced Cone Code&$N_X^{(4)} = \OO\left(l_2\left(N^{(3)}+N_Z^{(3)}\wz^{(3)}\right)\right)$; $N_X^{(4)} = \Omega\left(N_X^{(4d)}\right)$
\end{tabular}\caption{Upper and lower bounds on the number of X - stabilisers of the code as it undergoes the four steps. The amounts we thicken by, $l_1$ and $l_2$, are defined in Table \ref{allLengths}.}\label{allXStabs}
\end{table}
\renewcommand{\arraystretch}{1}

\renewcommand{\arraystretch}{1.5}
\begin{table}[]
\centering
\begin{tabular}{c||c}
Original&$N_Z$\\\hline\hline
Post-Copying&$N_Z^{(1)}=N_Z$\\\hline
Post-Gauging&$N_Z^{(2)} = N_Z^{(1)}$\\\hline
\makecell{Post-Thickening \\ and Choosing Heights}&$N_Z^{(3)} = \Theta\left(N_Z^{(2)}+l_1N^{(2)}\right)$\\\hline
Cone Code&$N_Z^{(4a)} = \OO(N_Z^{(3)}\wz^{(3)})$; $N_Z^{(4a)} = \Omega(N_Z^{(3)})$\\\hline
\makecell{Thickened \\ Cone Code} & $N_Z^{(4b)} = l_2N_Z^{(4a)}$\\\hline
\makecell{Thickened Cone Code \\ with Full Height Choice} & $N_Z^{(4c)} = N_Z^{(4b)}$\\\hline
\makecell{Thickened Cone Code \\ with Partial Height Choice} & $N_Z^{(4d)}=N_Z^{(4c)}$\\\hline
Reduced Cone Code&$N_Z^{(4)} = N_Z^{(4d)}$
\end{tabular}\caption{Upper and lower bounds on the number of Z - stabilisers of the code as it undergoes the four steps. The amounts we thicken by, $l_1$ and $l_2$, are defined in Table \ref{allLengths}.}\label{allZStabs}
\end{table}
\renewcommand{\arraystretch}{1}

\renewcommand{\arraystretch}{2}
\begin{table}[]
\centering
\begin{tabular}{c||c}
Original&$\rho_Z$\\\hline\hline
Post-Copying&$\rho_Z^{(1)} = \Omega\left(\frac{N^{(1)}}{N_X^{(1)}}\frac{\rho_Z}{\qx\rho_Z+\frac{N}{N_X}\qx^3}\right)$\\\hline
Post-Gauging&$\rho_Z^{(2)} = \Omega\left(\frac{N^{(2)}}{N_X^{(2)}}\frac{N_X^{(1)}}{N^{(1)}}\frac{\rho_Z^{(1)}}{\wx^{(1)}\left(1+\frac{N_X^{(1)}}{N^{(1)}}\rho_Z^{(1)}\right)}\right)$\\\hline
\makecell{Post-Thickening \\ and Choosing Heights}&$\rho_Z^{(3)}=\Omega\left(\frac{N^{(3)}}{N_X^{(3)}}\min\left(\frac{N_X^{(2)}}{N^{(2)}}\rho_Z^{(2)},1\right)\frac{1}{l_1}\right)$\\\hline
Cone Code&$\rho_Z^{(4a)} = \Omega\left(\frac{N^{(4a)}}{N_X^{(4a)}}\frac{\rho_Z^{(3)}}{\wz^{(3)}\left(\rho_Z^{(3)}+\frac{N^{(3)}}{N_X^{(3)}}\right)}\right)$\\\hline
\makecell{Thickened \\ Cone Code} & $\rho_Z^{(4b)} = \Omega\left(\frac{N^{(4b)}}{N_X^{(4b)}}\min\left(\frac{N_X^{(4a)}\rho_Z^{(4a)}}{N^{(4a)}},1\right)\frac{1}{l_2}\right)$\\\hline
\makecell{Thickened Cone Code \\ with Full Height Choice} & $\rho_Z^{(4c)} = \Omega\left(\frac{N_X^{(4b)}}{N_X^{(4c)}}\frac{\rho_Z^{(4b)}}{\wz^{(3)}\log\wz^{(3)}l_2}\right)$\\\hline
\makecell{Thickened Cone Code \\ with Partial Height Choice} & $\rho_Z^{(4d)} = \Omega\left(\frac{N_X^{(4c)}}{N_X^{(4d)}}\rho_Z^{(4c)}\right)$\\\hline
Reduced Cone Code&$\rho_Z^{(4)} = \Omega\left(\frac{N^{(4)}}{N_X^{(4)}}\frac{N_X^{(4d)}}{N^{(4d)}}\frac{\rho_Z^{(4d)}}{\wz^{(3)}\left(\frac{N_X^{(4d)}\rho_Z^{(4d)}}{N^{(4d)}}+1\right)}\right)$
\end{tabular}\caption{Lower bounds on the soundness of the Z - operators at each stage of the construction. The amounts we thicken by, $l_1$ and $l_2$, are defined in Table \ref{allLengths}. The final soundness of the Z - operators, after the full weight reduction, is given by $\rho_Z^{(4)}$.}\label{ZOpoSoundness}
\end{table}
\renewcommand{\arraystretch}{1}

\renewcommand{\arraystretch}{2}
\begin{table}[]
\centering
\begin{tabular}{c||c}
Original&$\rho_X$\\\hline\hline
Post-Copying&$\rho_X^{(1)} = \Omega\left(\qx\rho_X\right)$\\\hline
Post-Gauging&$\rho_X^{(2)} = \Omega\left(\frac{N^{(2)}}{N^{(1)}}\rho_X^{(1)}\right)$\\\hline
\makecell{Post-Thickening \\ and Choosing Heights}&$\rho_X^{(3)}=\Omega\left(\frac{N^{(3)}}{N_Z^{(3)}}\frac{1}{\wz^{(2)}\qz^{(2)}l_1^2}\min\left(\frac{N_Z^{(2)}}{N^{(2)}}\rho_X^{(2)},1\right)\right)$\\\hline
Cone Code&$\rho_X^{(4a)} = \Omega\left(\frac{N^{(4a)}}{N_Z^{(4a)}}\frac{N_Z^{(3)}}{N^{(3)}}\frac{\rho_X^{(3)}}{\wz^{(3)}\left(1+\frac{N_Z^{(3)}}{N^{(3)}}\rho_X^{(3)}\right)}\right)$\\\hline
\makecell{Thickened \\ Cone Code} & $\rho_X^{(4b)} = \Omega\left(\frac{N^{(4b)}}{N_Z^{(4b)}}\min\left(\frac{N_Z^{(4a)}\rho_X^{(4a)}}{N^{(4a)}},1\right)\frac{1}{l_2}\right)$\\\hline
\makecell{Thickened Cone Code \\ with Full Height Choice} & $\rho_X^{(4c)} = \rho_X^{(4b)}$\\\hline
\makecell{Thickened Cone Code \\ with Partial Height Choice} & $\rho_X^{(4d)} = \rho_X^{(4c)}$\\\hline
Reduced Cone Code&$\rho_X^{(4)} = \Omega\left(\frac{N^{(4)}}{N^{(4d)}}\rho_X^{(4d)}\right)$
\end{tabular}\caption{Lower bounds on the soundness of the X - operators at each stage of the construction. The amounts we thicken by, $l_1$ and $l_2$, are defined in Table \ref{allLengths}. The final soundness of the X - operators, after the full weight reduction, is given by $\rho_X^{(4)}$.}\label{XOpoSoundness}
\end{table}
\renewcommand{\arraystretch}{1}

\end{document}